\documentclass[10pt]{article}
\usepackage[hmargin=1in,vmargin=1in]{geometry}
\usepackage{setspace}
\usepackage{amsfonts,amsmath,amssymb,amsthm}
\usepackage[colorlinks,citecolor=blue]{hyperref}
\usepackage{natbib}
\usepackage[pdftex]{graphicx}
\usepackage{enumerate}
\usepackage{booktabs}
\usepackage[toc,title,titletoc,header]{appendix}
\usepackage{multirow}
\usepackage{threeparttable}
\usepackage{etoolbox}
\allowdisplaybreaks

\renewcommand{\qed}{\rule{2mm}{2mm}}

\newcommand{\indep}{\perp \!\!\! \perp}

\newtheorem{theorem}{Theorem}[section]
\newtheorem{lemma}{Lemma}[section]

\theoremstyle{definition}
\newtheorem{example}{Example}[section]
\newtheorem{remark}{Remark}[section]
\newtheorem{assumption}{Assumption}[section]

\newtheorem{condition}{Condition}[section]

\AtEndEnvironment{remark}{~\qed}
\AtEndEnvironment{example}{~\qed}

\DeclareMathOperator{\var}{Var}

\begin{document}

\author{
Yuehao Bai \\
Department of Economics\\
University of Southern California \\
\url{yuehao.bai@usc.edu}
\and
Jizhou Liu \\
HSBC Business School \\
Peking University \\
\url{jizhou.liu@phbs.pku.edu.cn}
\and
Azeem M.\ Shaikh\\
Department of Economics\\
University of Chicago \\
\url{amshaikh@uchicago.edu}
\and
Max Tabord-Meehan\\
Department of Economics\\
University of Toronto \\
\url{m.tabordmeehan@utoronto.ca}
}

\bigskip

\title{On the Efficiency of Highly Stratified Experiments \thanks{An earlier version of this paper was circulated under the title ``On the Efficiency of Finely Stratified Experiments.'' We thank Xinran Li, Wei Ma, Kirill Ponomarev, Joseph Romano, Andres Santos, Panos Toulis, Ke-Li Xu, and the seminar participants at the 2024 CEME Young Econometricians Conference, City University of Hong Kong, UC Riverside, UC Irvine, Cornell University, Syracuse University, the University of Toronto, the University of Wisconsin, Vanderbilt University, Carleton University, Queen's University, and McGill University for helpful comments. The third author acknowledges support from the National Science Foundation through grant SES-2419008. The fourth author acknowledges support from the National Science Foundation through grant SES-2149408.}}

\maketitle


\begin{spacing}{1.1}
\begin{abstract}
This paper studies the use of highly stratified designs for the efficient estimation of a large class of treatment effect parameters that arise in the analysis of experiments. By a ``highly stratified'' design, we mean experiments in which units are divided into blocks of a fixed size and a proportion within each block is assigned to a binary treatment uniformly at random. The class of parameters considered are those that can be expressed as the solution to a set of moment conditions constructed using a known function of the observed data. They include, among other things, average treatment effects, quantile treatment effects, and local average treatment effects as well as the counterparts to these quantities in experiments in which the unit is itself a cluster. In this setting, we establish three results. First, we show that under a highly stratified design, the na\"ive method of moments estimator achieves the same asymptotic variance as what could typically be attained under alternative treatment assignment mechanisms only through {\it ex post} covariate adjustment. Second, we argue that the na\"ive method of moments estimator under a highly stratified design is asymptotically efficient by deriving a lower bound on the asymptotic variance of regular estimators of the parameter of interest in the form of a convolution theorem. In this sense, highly stratified experiments are attractive because they lead to efficient estimators of treatment effect parameters ``by design.'' Finally, we strengthen this conclusion by establishing conditions under which a ``fast-balancing'' property of highly stratified designs is in fact necessary for the na\"ive method of moments estimator to attain the efficiency bound.
\end{abstract}
\end{spacing}



\thispagestyle{empty} 
\newpage
\setcounter{page}{1}

\section{Introduction}

This paper studies the use of highly stratified designs for the efficient estimation of a large class of treatment effect parameters that arise in the analysis of experiments. By a ``highly stratified'' design, we mean experiments in which units are divided into blocks of a \emph{fixed} size based on their covariate values and a proportion within each group is assigned to a binary treatment uniformly at random. The canonical example of such a design is a matched pairs design: here the fixed size of the blocks equals two, and exactly one of the two units in each block is assigned to treatment at random, so that the marginal probability of treatment assignment is one half. More broadly, a ``highly" stratified experiment is also intended to generalize the notion of a ``finely" stratified experiment \citep[as defined in][]{fogarty2018mitigating}, for which in every block there is exactly one treated or control individual. On the other hand, our use of the term ``highly stratified" in this context contrasts with what we could call a ``coarsely stratified" design, in which units are divided into a small set of \emph{large} blocks: see Example \ref{ex:CAR} below for a formal definition. The class of parameters considered are those that can be expressed as the solution to a set of moment conditions constructed using a known function of the observed data.  This class of parameters includes many causal parameters of interest: average treatment effects (ATEs), quantile treatment effects, and local average treatment effects as well as the counterparts to these quantities in experiments in which the unit is itself a cluster. 

In the setting described above, we establish three results. First, we study the asymptotic properties of a na\"ive method of moments estimator under a highly stratified design. Here, by a na\"ive method of moments estimator, we mean an estimator constructed using a direct sample analog of the moment conditions. For example, in the case of the ATE, such an estimator is given by the Horvitz-Thompson estimator for the difference in means. We show that under a highly stratified design, the na\"ive method of moments estimator achieves the same asymptotic variance as what could typically be attained under alternative treatment assignment mechanisms only through {\it ex post} covariate adjustment using the same set of covariates. Such adjustment strategies frequently involve the nonparametric estimation of conditional expectations or similar quantities; see, for example, \cite{zhang2008improving}, \cite{tsiatis2008covariate}, \cite{jiang2022improving}, \cite{jiang2022regression-adjusted} and \cite{rafi2023efficient}. We further illustrate that this feature of highly stratified experiments stems from the way in which highly stratified designs balance treatment status across covariate values, a property we define formally below and refer to as ``fast-balancing.'' Second, we derive a lower bound on the asymptotic variance of regular estimators of the parameter of interest in the form of a convolution theorem. This convolution theorem accommodates a large class of possible treatment assignment mechanisms, including covariate adaptive randomization \citep[][]{efron1971forcing, wei1978adaptive,zelen1974randomization, pocock1975sequential,hu2012asymptotic,bugni2018inference,ye2022inference,ma2020statistical,ma2024new}, re-randomization \citep[][]{li2017general,li2018asymptotic,li2020rerandomization,li2020rerandomization-1, cytrynbaum2024finely}, and highly stratified designs \citep[][]{jiang2021bootstrap, bai2022inference, cytrynbaum2023designing}.  We show that the lower bound is attained by the na\"ive method of moments estimator under a highly stratified design. In this sense, the na\"ive method of moments estimator under a highly stratified design is asymptotically efficient. More succinctly, we say that highly stratified experiments lead to efficient estimators ``by design.'' Finally, we strengthen this conclusion by characterizing all regular asymptotically linear estimators for a large class of treatment assignment mechanisms and use this characterization to establish conditions under which the fast-balancing property of highly stratified experiments is in fact a necessary condition for the na\"ive method of moments estimator to attain the efficiency bound.

Together, these results demonstrate that highly stratified experiments lead to efficient estimators that prioritize transparency in that they preclude the researcher from ``data snooping'' associated with {\it ex post} nonparametric covariate adjustment.  Importantly, concerns with this type of data snooping are not completely eliminated by typical pre-analysis plans because such adjustments involve choices, such as the choice of nonparametric estimator or tuning parameters, that are often not pre-registered prior to the experiment. The estimators are therefore attractive because they avoid performing nonparametric covariate adjustment in order to achieve efficiency and thereby remain ``hands above the table'' \citep[][]{freedman2008regression, lin2013agnostic}. 

Our paper builds upon two strands of literature. The first strand of literature concerns the analysis of highly stratified experiments.  Within this literature, our analysis is most closely related to \cite{bai2022inference}, who derive the asymptotic behavior of the difference-in-means estimator of the ATE when treatment is assigned according to a matched pairs design, and \cite{cytrynbaum2023designing}, who develops related results for an experimental design referred to as ``local randomization'' that permits the proportion of units assigned to treatment to vary with the baseline covariates. Beyond settings that study estimation of the ATE, \cite{bai2024inference-1} develops results for the analysis of different cluster-level average treatment effects and \cite{jiang2021bootstrap} develop results analogous to those in \cite{bai2022inference} for suitable estimators of the quantile treatment effect. This paper, like those just mentioned, operates in a ``super-population'' framework, in which the outcomes and covariates are assumed to be drawn as an i.i.d.\ sample from a population distribution. This is in contrast to an alternative strand of the literature that studies highly stratified experiments from the design-based perspective \citep[see, in particular,][]{imai2008variance, imai2009essential, fogarty2018mitigating, fogarty2018regression-assisted, liu2020regression-adjusted, pashley2021insights, bai2025new}. To our knowledge, our paper is the first to analyze the properties of highly stratified experiments in a general framework that accommodates any parameter that can be characterized as the solution to a set of moment conditions involving a known function of the observed data. In more recent work, \cite{cytrynbaum2024finely} considers a similar framework to study highly stratified \emph{re-randomized} experiments, which nest highly stratified experiments as a special case. However, his results assume that the moment functions which define the parameter of interest are continuous in a way that precludes parameters like the Quantile Treatment Effect. Moreover, we emphasize that none of the above papers formally establish the asymptotic efficiency of highly stratified experiments. The second strand of literature concerns bounds on the efficiency with which treatment effect parameters can be estimated in experiments. We note that, due to the potential for dependence in treatment assignments across individuals, we cannot immediately appeal to standard semi-parametric efficiency results \citep[see, for example,][]{van_der_vaart1998asymptotic, chen2008semiparametric}. Two important recent papers in this literature studying efficiency bounds in the special case of estimating the ATE are \cite{armstrong2022asymptotic} and \cite{rafi2023efficient}. Even in this special case, their results differ from ours in important and empirically relevant ways; Remark \ref{rem:other_SPEB} provides an in-depth discussion of the connection between these results and ours.  See also \cite{bai2022optimality} for some finite-sample optimality properties of matched pairs designs for estimation of the ATE.

The remainder of this paper is organized as follows.  In Section \ref{sec:setup}, we describe our setup and notation.  We emphasize, in particular, the way in which our framework can accommodate various treatment effect parameters of interest.  Section \ref{sec:variance} derives the asymptotic behavior of the na\"ive method of moments estimator of our parameter of interest when treatment is assigned using a highly stratified design and studies estimation of the asymptotic variance. In Section \ref{sec:semi}, we develop our lower bound on the asymptotic variance of regular estimators of these parameters and show that it is achieved by the the na\"ive method of moments estimator in a highly stratified design. In Section \ref{sec:regular}, we characterize all regular asymptotically linear estimators and argue that the fast-balancing property is necessary for the na\"ive method of moments estimator to attain the efficiency bound. In Section \ref{sec:sims}, we illustrate the practical relevance of our theoretical results through a simulation study. Finally, we conclude in Section \ref{sec:recs} with some recommendations for empirical practice guided by both these simulations and our theoretical results. Proofs of all results can be found in the supplementary material.

\section{Setup and Motivation}\label{sec:setup}

Let $A_i \in \{0, 1\}$ denote the treatment status of the $i$th unit, and let $X_i \in \mathbf R^{d_x}$ denote their observed, baseline covariates. For $a \in \{0, 1\}$, let $R_i(a) \in \mathbf R^{d_r}$ denote a vector of potential responses. As we illustrate below, considering a vector of responses allows us to accommodate many parameters of interest. Let $R_i \in \mathbf R^{d_r}$ denote the vector of observed responses obtained from $R_i(a)$ once treatment is assigned. As usual, the observed responses and potential responses are related to treatment status by the relationship
\begin{equation}\label{eq:PO}
R_i = R_i(1) A_i + R_i(0) (1 - A_i)~. 
\end{equation}
We assume throughout that our sample consists of $n$ units. For any random vector indexed by $i$, for example $A_i$, we define $A^{(n)} = (A_1, \ldots, A_n)$. Let $P_n$ denote the distribution of the observed data $(R^{(n)}, A^{(n)}, X^{(n)})$, and $Q_n$ the distribution of $(R^{(n)}(1), R^{(n)}(0), X^{(n)})$. We assume $Q_n = Q^n$, where $Q$ is the marginal distribution of $(R_i(1), R_i(0), X_i)$. Given $Q_n$, $P_n$ is then determined by \eqref{eq:PO} and the mechanism for determining treatment assignment. We assume that treatment assignment is performed such that a standard unconfoundedness assumption holds and such that the probability of assignment given $X_i$ is some known constant for every $1 \le i \le n$, as is often the case in most experiments:

\begin{assumption} \label{ass:unconfounded}
Treatment status is assigned so that
\begin{equation} \label{eq:unconfounded}
    (R^{(n)}(1), R^{(n)}(0)) \indep A^{(n)} | X^{(n)}~,
\end{equation}
and such that $P\{A_i = 1|X_i=x\} = \eta$, for some $\eta \in (0, 1)$ for all $1 \le i \le n$.
\end{assumption}

Assumption \ref{ass:unconfounded} restricts the probability of assignment to be the fixed fraction $\eta$ across the entire experimental sample, but this restriction can be weakened so that $\eta$ is replaced by $\eta(X_i)$ for many of our subsequent results: see Remark \ref{rem:generaleta} for a discussion. Given Assumption \ref{ass:unconfounded}, it can be shown that $(X_i, A_i, R_i)$ are identically distributed for $1 \le i \le n$, and their marginal distribution does not change with $n$ (see Lemma A.7 in the Supplement). As a consequence, we denote the marginal distribution of $(X_i, A_i, R_i)$ by $P$. We consider parameters $\theta_0 \in \Theta \subset \mathbf R^{d_\theta}$ that can be defined as the solution to a set of moment equalities. As we show below, these parameters include a large class of causal parameters defined in terms of potential outcomes and potential treatments. Formally, let $m: \mathbf R^{d_x} \times \{0, 1\} \times \mathbf R^{d_r} \to \mathbf R^{d_\theta}$ be a known measurable function, then we consider parameters $\theta_0$ that uniquely solve the moment equality
\begin{equation} \label{eq:moments}
E_P[m(X_i, A_i, R_i, \theta_0)] = 0~.
\end{equation}
We emphasize that $m(\cdot)$ is not a function of any unknown nuisance parameters, but may depend on the known value of $\eta$ in Assumption \ref{ass:unconfounded}. We present five examples of well-known parameters that can be described as (functions of) solutions to a set of moment conditions as in \eqref{eq:moments}.

\begin{example}[Average Treatment Effect]\label{ex:ATE}
Let $Y_i(a) = R_i(a)$ denote a scalar potential outcome for the $i$th unit under treatment $a \in \{0, 1\}$, and let $Y_i = R_i$ denote the observed outcome. Let $\theta_0 = E_Q[Y_i(1) - Y_i(0)]$ denote the average treatment effect (ATE). Under Assumption \ref{ass:unconfounded}, $\theta_0$ solves the moment condition in \eqref{eq:moments} with
\begin{equation} \label{eq:moments-ate}
   m(X_i, A_i, R_i, \theta) = \frac{Y_i A_i}{\eta} - \frac{Y_i (1 - A_i)}{1 - \eta} - \theta~. 
\end{equation}
For papers that consider estimators based on \eqref{eq:moments-ate}, see \cite{hirano2001estimation} and \cite{hirano2003efficient}.
\end{example}

\begin{example}[Quantile Treatment Effect]\label{ex:QTE}
Let $Y_i(a) = R_i(a)$ denote a scalar potential outcome for the $i$th unit under treatment $a \in \{0, 1\}$, and let $Y_i = R_i$ denote the observed outcome. Let $\tau \in (0, 1)$ and $\theta_0 = (\theta_0(1), \theta_0(0))' = (q_{Y(1)}(\tau), q_{Y(0)}(\tau))'$, where
\[ q_{Y(a)}(\tau) = \inf \{\lambda \in \mathbf R: Q \{Y_i(a) \leq \lambda \} \geq \tau\}~. \]
In other words, $\theta_0$ is defined to be the vector of $\tau$th quantiles of the marginal distributions of $Y_i(1)$ and $Y_i(0)$. If we assume $q_{Y(a)}(\tau)$ is unique for $a \in \{0, 1\}$ in the sense that $Q\{Y(a) \leq q_{Y(a)}(\tau) + \epsilon\} > Q\{Y(a) \leq q_{Y(a)}(\tau)\}$ for all $\epsilon > 0$, then it follows from Assumption \ref{ass:unconfounded} and Lemma 1 in \cite{firpo2007efficient} that $\theta_0$ solves the moment condition in \eqref{eq:moments} with
\[ m(X_i, A_i, R_i, \theta) = \begin{pmatrix}
\displaystyle \frac{A_i (\tau - I \{Y_i \leq \theta(1)\})}{\eta} \\
\displaystyle \frac{(1 - A_i) (\tau - I \{Y_i \leq \theta(0)\})}{1 - \eta}
\end{pmatrix}~, \]
for $\theta = (\theta^{(1)}, \theta^{(0)})'$. Note that the quantile treatment effect $q_{Y(1)}(\tau) - q_{Y(0)}(\tau)$ can then be defined as $h(\theta_0)$ where $h:\mathbf R^2 \to \mathbf R$ is given by $h(s,t) = s - t$.
\end{example}

\begin{example}[Local Average Treatment Effect]\label{ex:LATE}
Let $(\tilde{Y}_i(a), D_i(a)) = R_i(a)$ denote the vector of potential outcomes ($\tilde{Y}_i(a) \in \mathbf R$) and treatment take-up ($D_i(a) \in \{0, 1\}$) under treatment $a \in \{0, 1\}$, and let $(Y_i, D_i) = R_i$ denote the vector of observed outcomes and treatment take-up. Note here that $\tilde{Y}_i(a)$ corresponds to the potential outcome under assignment $a \in \{0,1\}$ and not to the potential outcome for a given take-up $D_i = d$. Suppose $E_Q[D_i(1) - D_i(0)] \ne 0$ and let
\[ \theta_0 = \frac{E_Q[\tilde{Y}_i(1) - \tilde{Y}_i(0)]}{E_Q[D_i(1) - D_i(0)]}~. \]
It then follows from Assumption \ref{ass:unconfounded} that $\theta_0$ solves the moment condition in \eqref{eq:moments} with
\begin{equation} \label{eq:moments-late}
    m(X_i, A_i, R_i, \theta) = \frac{Y_i A_i}{\eta} - \frac{Y_i (1 - A_i)}{1 - \eta} - \theta \left ( \frac{D_i A_i}{\eta} - \frac{D_i (1 - A_i)}{1 - \eta} \right )~.
\end{equation}
If we further assume instrument monotonicity (i.e., $P\{D_i(1) \ge D_i(0)\} = 1$) and instrument exclusion, then $\theta_0$ could be re-interpreted as the local average treatment effect (LATE) in the sense of \cite{imbens1994identification}.
\end{example}

\begin{example}[Weighted Average Treatment Effect]\label{ex:Clust_ATE}
Let $Y_i(a) = R_i(a)$ denote a scalar potential outcome for the $i$th unit under treatment $a \in \{0, 1\}$, and let $Y_i = R_i$ denote the observed outcome. 
Let 
\[\theta_0 = E_Q\left[\frac{\omega(X_i)}{E_Q[\omega(X_i)]}\left(Y_i(1) - Y_i(0)\right)\right]~,\]
for some known function $\omega: \mathbf R^{d_x} \to \mathbf R$.
It then follows from Assumption \ref{ass:unconfounded} that $\theta_0$ solves the moment condition in \eqref{eq:moments} with
\[ m(X_i, A_i, R_i, \theta) = \omega(X_i)\left(\frac{Y_i A_i}{\eta} - \frac{Y_i (1 - A_i)}{1 - \eta}\right) - \omega(X_i)\theta~.\] 
By defining $Y_i$ to be the average outcome in the $i$th cluster, $\theta_0$ defined in this way can accommodate the (cluster) size-weighted and equally-weighted average treatment effects considered in \cite{bugni2022inference} and \cite{bai2024inference-1} in the context of cluster-level randomized controlled trials.
\end{example}

\begin{example}[Log-Odds Ratio]\label{ex:log_odds}
Let $Y_i(a) = R_i(a) \in \{0, 1\}$ denote a binary potential outcome for the $i$th unit under treatment $a \in \{0, 1\}$, and let $Y_i = R_i$ denote the observed outcome.
Suppose $0 < P\{Y_i(a) = 0\} < 1$ for $a \in \{0, 1\}$, and let $\theta_0 = (\theta_0(1), \theta_0(2))'$, where
\[\theta_0(1) = \text{logit}(E_Q[Y_i(0)])~,\]
\[\theta_0(2) = \text{logit}(E_Q[Y_i(1)]) - \text{logit}(E_Q[Y_i(0)])~,\]
with $\text{logit}(z) = \log(\frac{z}{1-z})$, so that $\theta_0(2)$ denotes the log-odds ratio of treatment $1$ relative to treatment $0$. It follows from Assumption \ref{ass:unconfounded} that $\theta_0$ solves the moment condition in \eqref{eq:moments} with
\[m(X_i, A_i, R_i, \theta) =  \begin{pmatrix}
\displaystyle 1 - A_i \\
\displaystyle A_i
\end{pmatrix}\left(Y_i - \text{expit}(\theta(1) + \theta(2)A_i)\right)~,\]
where $\text{expit}(z) = \frac{\exp(z)}{1 + \exp(z)}$. The log-odds ratio can then be defined as $h(\theta_0)$ where $h: \mathbf R^2 \to \mathbf R$ is given by $h(s,t) = t$. This parameter appears in, for example, \cite{zhang2008improving}.
\end{example}
Additional examples could be obtained by considering combinations of Examples \ref{ex:ATE}--\ref{ex:log_odds}. For instance, combining the moment functions from Examples \ref{ex:LATE} and \ref{ex:Clust_ATE} would result in a weighted LATE parameter. Beyond these examples, certain treatment effect contrasts could also be related to the structural parameters in, for instance, an economic model of supply and demand: see, for example, the model estimated in \cite{casaburi2022using}.

Throughout the rest of the paper we consider the asymptotic properties of the method  of moments estimator $\hat{\theta}_n$ for $\theta_0$ which is constructed as a solution to the sample analogue of \eqref{eq:moments}:
\begin{equation} \label{eq:est}
\frac{1}{n} \sum_{1 \leq i \leq n} m(X_i, A_i, R_i, \hat \theta_n) = 0~.
\end{equation}
Because $\hat{\theta}_n$ is constructed directly using the moment function $m(\cdot)$, we call $\hat{\theta}_n$ the na\"ive method of moments estimator. Note that $\hat \theta_n$ as defined in \eqref{eq:est} is closely related to standard estimators of the parameter $\theta_0$ in specific examples. For instance, in Example \ref{ex:ATE}, \[\hat{\theta}_n = \frac{1}{\eta}\sum_{1 \le i \le n}Y_iA_i - \frac{1}{1 - \eta}\sum_{1 \le i \le n}Y_i(1 - A_i)~,\]
so that $\hat{\theta}_n$ is a Horvitz-Thompson analogue of the standard difference-in-means estimator for the ATE. In Example \ref{ex:LATE},
\[\hat{\theta}_n = \frac{\frac{1}{\eta}\sum_{1 \le i \le n}Y_iA_i - \frac{1}{1 - \eta}\sum_{1 \le i \le n}Y_i(1 - A_i)}{\frac{1}{\eta}\sum_{1 \le i \le n}D_iA_i - \frac{1}{1 - \eta}\sum_{1 \le i \le n}D_i(1 - A_i)}~,\]
so that $\hat{\theta}_n$ is a Horvitz-Thompson analogue of the standard Wald estimator for the local average treatment effect. 

Before proceeding, in the remainder of this section, we provide a more detailed summary of the main contributions of our paper.  To this end, first note that if $A^{(n)}$ were assigned i.i.d., independently of $X^{(n)}$, then it can be shown under mild conditions on $m(\cdot)$ \citep[see, for instance, Theorem 5.1 in][]{van_der_vaart1998asymptotic} that the na\"ive method of moments estimator satisfies 
\[\sqrt{n}(\hat{\theta}_n - \theta_0) \overset{d}{\to} N(0, \mathbb{V})~,\]
where
\begin{equation}\label{eq:naive_variance}
\mathbb V = M^{-1}E_P[m(X_i, A_i, R_i, \theta_0)m(X_i, A_i, R_i, \theta_0)'](M^{-1})^{\prime}~,
\end{equation}
with $M = \frac{\partial}{\partial \theta'} E_P[m(X, A, R, \theta)] \Big |_{\theta = \theta_0}$. In Section \ref{sec:variance}, we show that if we assign $A^{(n)}$ using a highly stratified design (see Assumption \ref{ass:a} below for a formal definition) then, under appropriate assumptions so that we achieve ``fast balance'' of the treatment across covariate values (see Assumptions \ref{ass:pair}, \ref{ass:normal} below),
\[\sqrt{n}(\hat{\theta}_n - \theta_0) \overset{d}{\to} N(0, \mathbb{V}_*)~,\]
where $\mathbb{V} \ge \mathbb{V}_*$ (see Theorem \ref{thm:normal}). Under i.i.d.\ assignment, the na\"ive method of moment estimator $\hat \theta_n$ cannot generally attain $\mathbb V_\ast$, but an estimator that attains $\mathbb{V}_*$  could instead be constructed by appropriately ``augmenting'' the moment function, and then considering an estimator which solves the augmented moment equation. For instance, if we consider the ATE in Example \ref{ex:ATE}, then it is straightforward to show that the following augmented moment function identifies $\theta_0$:
\begin{equation}\label{eq:augment_m}
m^*(X_i, A_i, R_i, \theta) = \frac{A_i(Y_i - \mu_1(X_i))}{\eta} - \frac{(1 - A_i)(Y_i - \mu_0(X_i))}{1 - \eta} + \mu_1(X_i) - \mu_0(X_i) - \theta~,
\end{equation}
where $\mu_a(X_i) = E_Q[Y_i(a)|X_i]$. 
This choice of $m^*(\cdot)$ produces the well known doubly-robust moment condition for estimating the ATE \citep{robins1995analysis,hahn1998role}. It can then be shown that an appropriately constructed two-step estimator, in which $\mu_1(\cdot)$ and $\mu_0(\cdot)$ are non-parametrically estimated in a first step, attains $\mathbb{V}_*$  \citep{tsiatis2008covariate,farrell2015robust,chernozhukov2017doubledebiasedneyman, rafi2023efficient}. Intuitively, the estimator obtained from the augmented moment function $m^*(\cdot)$ performs nonparametric covariate adjustment by exploiting the information contained in $X^{(n)}$ that may not have been captured in the original moment function $m(\cdot)$.  Similar nonparametric covariate adjustments based on augmented moment equations have been developed for other parameters of interest \citep{zhang2008improving, belloni2017program,jiang2022improving,jiang2022regression-adjusted}.   In this sense, we show that highly stratified designs can perform nonparametric covariate adjustment ``by design'' for the large class of parameters that can be expressed in terms of moment conditions of the form given in \eqref{eq:moments}, thus generalizing similar observations made in \cite{bai2022inference}, \cite{bai2022optimality}, and \cite{cytrynbaum2023designing} in the special case of estimating the ATE. As we explain in the discussion following Theorem \ref{thm:normal}, highly stratified experiments have this feature because they lead to fast-balancing of the treatment across covariate values, as defined formally in \eqref{eq:fast-m} below.

Earlier work on efficient treatment effect estimation has noted that the variance $\mathbb{V}_*$ is in fact the efficiency bound for estimating $\theta_0$ under i.i.d.\ assignment \citep[see, for instance,][]{cattaneo2010efficient}. A natural follow-up question is whether or not $\mathbb{V}_*$ continues to be the efficiency bound for estimating $\theta_0$ under a highly stratified design, or more generally for complex experimental designs which induce dependence in the treatment assignments across individuals in the experiment. In Section \ref{sec:semi}, we show that $\mathbb V_\ast$ continues to be the efficiency bound for estimating $\theta_0$ for a large class of treatment assignment mechanisms with a fixed marginal probability of treatment assignment, which includes highly stratified designs as a special case. We can thus conclude that, from the perspective of asymptotic efficiency, highly stratified designs are optimal experimental designs for a broad range of treatment effect estimation problems. In Section \ref{sec:regular} we build on this result and establish conditions under which efficient estimation of $\theta_0$ using the na\"ive method of moments estimator can be achieved only if the experimental design leads to fast-balancing of the treatment across covariate values. In this sense, we show that the fast-balancing property of highly stratified experiments is in fact a \emph{necessary} condition for achieving efficient estimation of $\theta_0$ ``by design.''

\section{The Asymptotic Variance of Highly Stratified Experiments}\label{sec:variance}
In this section, we derive the asymptotic distribution of the method of moments estimator $\hat \theta_n$ when treatment is assigned by a highly stratified design over the baseline covariates $X^{(n)}$. This assignment mechanism uses the covariates $X^{(n)}$ to group units with similar covariate values into blocks of fixed size, and then assigns treatment completely at random within each block.  In order to describe the assignment mechanism formally, we require some further notation to define the blocks of units. Let $\ell$ and $k$ be arbitrary positive integers with $\ell < k$ and set $\eta = \ell/k$. Here, $k$ is the total number of units in each block and $\ell$ is the number of treated units in each block. For simplicity, assume that $n$ is divisible by $k$. We then represent blocks of units using a partition of $\{1, \ldots, n\}$ given by
\[\left\{\lambda_j = \lambda_j(X^{(n)}) \subseteq \{1, \ldots, n\}, 1 \le j \le n/k\right\}~,\]
with $|\lambda_j| = k$.  Because of its possible dependence on $X^{(n)}$, $\{\lambda_j: 1 \le j \le n/k\}$ encompasses a variety of different ways of blocking the $n$ units according to the observed, baseline covariates. We note, however, that our framework is not intended to reflect settings in which the blocks themselves are randomly sampled from the population of interest \citep[see, for instance, the discussion in Section 4.4 of][]{pashley2021insights}.  Given such a partition, we assume that treatment status is assigned as described in the following assumption:
\begin{assumption} \label{ass:a}
Treatment status is assigned so that $(R^{(n)}(1), R^{(n)}(0)) \indep A^{(n)} \big | X^{(n)}$ and, conditional on $X^{(n)}$, 
\[\{(A_i: i \in \lambda_j): 1 \le j \le n/k\}\]
are i.i.d.\ and each uniformly distributed over all permutations of $(\underbrace{0, 0, \ldots, 0}_{k - \ell}, \underbrace{1, 1, \ldots, 1}_{\ell})$.
\end{assumption}

The assignment mechanism described in Assumptions \ref{ass:a} generalizes the definition of a matched pairs design. In particular, we recover a matched pairs design if we set $(\ell, k) = (1, 2)$, with $\eta = 1/2$. Indeed, suppose $n$ is even and consider pairing the experimental units into $n / 2$ pairs, represented by the sets
\[ \lambda_j = \{\pi(2j - 1), \pi(2j)\} \text{ for } j = 1, \ldots, n / 2~, \]
where $\pi = \pi_n(X^{(n)})$ is a permutation of $n$ elements.  Because of its possible dependence on $X^{(n)}$, $\pi$ encompasses a broad variety of ways of pairing the $n$ units according to the observed, baseline covariates $X^{(n)}$. Given such a $\pi$, we assume that treatment status is assigned so that Assumption \ref{ass:a} holds and, conditional on $X^{(n)}$, $(A_{\pi(2j-1)}, A_{\pi(2j)}), j = 1, \ldots, n / 2$ are i.i.d.\ and each uniformly distributed over the values in $\{(0,1), (1,0)\}$. For some examples of such an assignment mechanism being used in practice, see, for instance, \cite{angrist2009effects}, \cite{banerjee2015miracle}, and \cite{bruhn2016impact}.

\begin{remark}\label{rem:matched_pairs}
Note that Assumption \ref{ass:a} generalizes matched pairs designs along two dimensions: first, it allows for treatment fractions other than $\eta = 1/2$. Second, it allows for choices of $\ell$ and $k$ which are not relatively prime. For instance, if we set $(\ell, k) = (2, 4)$, then $\eta = 1/2$ as in matched pairs, but now the assignment mechanism blocks units into groups of size $4$ and assigns two units to treatment, two units to control. Although Theorem \ref{thm:normal} below establishes that allowing for this level of flexibility has no effect on the asymptotic properties of our estimator, in our experience we have found that designs which employ these treatment ``replicates'' in each block can simplify the construction of variance estimators in practice; see Section \ref{sec:var-est} for details, and \cite{imbens2011experimental} for an early discussion. 
\end{remark}

Our analysis will require some discipline on the way in which the blocks are formed.  In particular, we will require that the units in each block be close in terms of their baseline covariates in the sense described by the following assumption:
\begin{assumption} \label{ass:pair}
The blocks used in determining treatment status satisfy
\[ \frac{1}{n} \sum_{1 \leq j \leq n/k} \max_{i, i' \in \lambda_j} \|X_{i} - X_{i'}\|^2 \stackrel{P}{\rightarrow} 0~. \]
\end{assumption}
\cite{bai2022inference} and \cite{cytrynbaum2023designing} discuss blocking algorithms that satisfy Assumption \ref{ass:pair}. When $X_i \in \mathbf R$ and $E_Q[X_i^2] < \infty$, a simple algorithm that satisfies Assumption \ref{ass:pair} is simply to order units from smallest to largest and then block adjacent units into blocks of size $k$. In the case of matched pairs, if $\mathrm{dim}(X_i) > 1$ and $E_Q[\|X_i\|^d] < \infty$ for $d \geq \mathrm{dim}(X_i) + 1$, then Assumption \ref{ass:pair} is satisfied by the \texttt{nbpmatching} algorithm in \texttt{R} that minimizes the sum of squared distances of $X$ within pairs: see Appendix A of \cite{bai2024inference-1} for details. Beyond the case of pairs, \cite{cytrynbaum2023designing} demonstrates that the optimal blocking satisfies the following bound
\[
\frac{1}{n} \sum_{1 \leq j \leq n/k} \max_{i, i' \in \lambda_j} \|X_{i} - X_{i'}\|^2 = O_P(n^{{2/d} - 2/(\rm{dim}(X_i)+1)})~,
\]
from which we can deduce that the rate of convergence depends on both the dimension of $X_i$ as well as the number of moments it possesses. Note further that Assumption \ref{ass:pair} can be satisfied even if $X_i$ contains discrete components, as may arise when summarizing a categorical variable numerically using, e.g., one-hot encoding.

Finally, we impose the following assumptions to derive the large-sample properties of $\hat{\theta}_n$. In what follows, when writing expectations and variances, we suppress the subscripts $P$ and $Q$ whenever doing so does not lead to confusion.

\begin{assumption} \label{ass:normal}
Let $m(\cdot) = (m_s(\cdot): 1 \le s \le d_{\theta})'$. The moment functions are such that
\begin{enumerate}[\rm (a)]
    \item For every $\epsilon > 0$, $\inf\limits_{\theta \in \Theta: \|\theta - \theta_0\| > \epsilon} \| E[m(X_i, A_i, R_i, \theta)] \| > 0$.
    \item $E[m(X_i, A_i, R_i, \theta)]$ is differentiable at $\theta_0$ with a nonsingular derivative $M = \frac{\partial}{\partial \theta'} E[m(X, \allowbreak A, R, \theta)] \Big |_{\theta = \theta_0}$.
    \item For $1 \leq s \leq d_\theta$, $E[((m_s(X, a, R(a), \theta) - m_s(X, a, R(a), \theta_0))^2] \to 0$ as $\theta \to \theta_0$ for $a \in \{0, 1\}$.
    \item For $1 \leq s \leq d_\theta$, $\{m_s(x, a, r, \theta): \theta \in \Theta\}$ is pointwise measurable in the sense that there exists a countable set $\Theta^\ast$ such that for each $\theta \in \Theta$, there exists a sequence $\{\theta_m\} \subset \Theta^\ast$ such that $m_s(x, a, r, \theta_m) \to m_s(x, a, r, \theta)$ as $m \to \infty$ for all $x, a, r$.
    \item (i) $\sup_{\theta \in \Theta} E[\|m(X, a, R(a), \theta)\|] < \infty$ for $a \in \{0, 1\}$. (ii) $\{m_s(x, 1, r, \theta): \theta \in \Theta^\ast\}$ and $\{m_s(x, 0, r, \theta): \theta \in \Theta^\ast\}$ are $Q$-Donsker for $1 \leq s \leq d_\theta$.
    \item For $a \in \{0, 1\}$, $E[ m_s(X, a, R(a), \theta_0) | X = x]$ is $C$-Lipschitz for $1 \leq s \leq d_\theta$, for some constant $C < \infty$.
\end{enumerate}
\end{assumption}

Assumption \ref{ass:normal}(a) is a standard assumption to ensure the solution to \eqref{eq:moments} is ``well separated.'' It appears as a condition, for instance, in Theorem 5.9 in \cite{van_der_vaart1998asymptotic}. Assumption \ref{ass:normal}(b) is a standard assumption used when deriving the properties of $Z$-estimators. See, for instance, Theorem 3.1 in \cite{newey1994large} and Theorem 5.21 in \cite{van_der_vaart1998asymptotic}. Because differentiability is imposed on their expectations instead of the moment functions themselves, the moment functions are allowed to be nonsmooth, as in Example \ref{ex:QTE}. Assumption \ref{ass:normal}(c) requires the moment function to be mean-square continuous in $\theta$. Assumption \ref{ass:normal}(d) is a standard condition to guarantee the measurability of the supremum of a suitable class of functions. In particular, it allows us to define expectations of suprema without invoking outer expectations. See Example 2.3.4 in \cite{van_der_vaart1996weak} for details. Assumption \ref{ass:normal}(e) is a standard assumption which guarantees the existence of an integrable envelope and allows us to invoke a uniform law of large numbers and a uniform central limit theorem (see, for instance, page 81 of \cite{van_der_vaart1996weak} for a definition of a Donsker class). In particular, this assumption can be verified for Examples \ref{ex:ATE}--\ref{ex:log_odds}. Assumption \ref{ass:normal}(f) is a common assumption which simplifies some arguments when studying highly stratified designs, and ensures units that are close in terms of the baseline covariates are also close in terms of their moments. Note that Assumption \ref{ass:normal}(f) could be dropped following the approximation arguments in Lemma C.5 of \cite{cytrynbaum2023designing}; see also Examples \ref{ex:MP} and \ref{ex:MP2} for further discussion.

The following theorem establishes the asymptotic variance of the na\"ive method of moments estimator when the treatment assignment mechanism is highly stratified in the sense of satisfying Assumptions \ref{ass:a}--\ref{ass:pair}. Its proof relies on a crucial technical result in \cite{han2021complex}, which allows us to compare the empirical process that depends on the treatment assignments with the empirical process that only depends on i.i.d.\ quantities. As a consequence of us leveraging this result, Assumption \ref{ass:normal} is comparable to the typical assumptions imposed to study the properties of method of moments estimators using i.i.d.\ data.

\begin{theorem} \label{thm:normal}
Suppose the treatment assignment mechanism satisfies Assumptions \ref{ass:a}--\ref{ass:pair} and the moment functions satisfy Assumption \ref{ass:normal}. Let $\hat{\theta}_n$ be defined as in \eqref{eq:est}. Then,
\begin{equation} \label{eq:normal}
\sqrt n(\hat \theta_n -  \theta_0) = \frac{1}{\sqrt{n}} \sum_{1 \leq i \leq n} \psi^\ast(X_i, A_i, R_i, \theta_0) + o_P(1)~.
\end{equation}
where
\begin{align*}
& \psi^\ast(X_i, A_i, R_i, \theta_0) \\
& = - M^{-1} \Big ( I\{A_i = 1\} (m(X_i, 1, R_i, \theta_0) - E[m(X_i, 1, R_i(1), \theta_0) | X_i]) \\
& \hspace{3.5em} + I\{A_i = 0\} (m(X_i, 0, R_i, \theta_0) - E[m(X_i, 0, R_i(0), \theta_0) | X_i]) \\
& \hspace{3.5em} + \eta E[m(X_i, 1, R_i(1), \theta_0) | X_i] + (1 - \eta) E[m(X_i, 0, R_i(0), \theta_0) | X_i] \Big )~.
\end{align*}
Further, we have that
\begin{equation} \label{eq:normal_convergence}
\sqrt n(\hat \theta_n -  \theta_0) \stackrel{d}{\to} N(0, \mathbb V_\ast)~,
\end{equation}
where
\begin{equation} \label{eq:Vstar}
    \mathbb V_\ast = \var[\psi^\ast(X_i, A_i, R_i, \theta_0)]~.
\end{equation}
\end{theorem}

In order to make Theorem \ref{thm:normal} useful for inference about $\theta_0$, we describe in Section \ref{sec:var-est} an estimator $\hat{\mathbb{V}}_n$ of $\mathbb{V}_*$. We now sketch an argument of the proof of Theorem \ref{thm:normal} to highlight the fundamental role played by a ``fast-balancing'' property of highly stratified designs (see \eqref{eq:fast-m} below). In the proof of Theorem \ref{thm:normal}, we first establish that
\begin{equation} \label{eq:vdv}
\sqrt n(\hat \theta_n -  \theta_0) = -M^{-1}\frac{1}{\sqrt{n}}\sum_{1 \le i \le n}m(X_i,A_i,R_i,\theta_0) + o_P(1)~.
\end{equation}
To further establish \eqref{eq:normal}, it thus suffices to show that, under a highly stratified design,
\begin{equation}\label{eq:by_design}
-M^{-1}\frac{1}{\sqrt{n}}\sum_{1 \le i \le n}m(X_i,A_i,R_i,\theta_0) = \frac{1}{\sqrt{n}} \sum_{1 \leq i \leq n} \psi^\ast(X_i, A_i, R_i, \theta_0) + o_P(1)~.
\end{equation}
To obtain this equivalence, consider the following decomposition of $m(\cdot)$:
\begin{align*}
& m(X_i, A_i, R_i, \theta_0) \\
& = \eta E[m(X_i, 1, R_i(1), \theta_0) | X_i] + (1 - \eta) E[m(X_i, 0, R_i(0), \theta_0) | X_i] \\
& \hspace{3.5em} + I\{A_i = 1\} (m(X_i, 1, R_i, \theta_0) - E[m(X_i, 1, R_i(1), \theta_0) | X_i]) \\
& \hspace{3.5em} + I\{A_i = 0\} (m(X_i, 0, R_i, \theta_0) - E[m(X_i, 0, R_i(0), \theta_0) | X_i]) \\
& \hspace{3.5em} + (A_i - \eta) (E[m(X_i, 1, R_i(1), \theta_0) - m(X_i, 0, R_i(0), \theta_0) | X_i]) ~.
\end{align*}
Then the equivalence follows if we can show that
\begin{equation}\label{eq:fast-m}\frac{1}{\sqrt{n}}\sum_{1 \le i \le n}(A_i - \eta)\left(E[m(X_i, 1,R_i(1),\theta_0) - m(X_i, 0,R_i(0),\theta_0)|X_i]\right) = o_P(1)~.
\end{equation}
We call \eqref{eq:fast-m} the fast-balancing condition for the function $m(\cdot)$. Intuitively, the fast-balancing condition imposes that the experimental design should balance the treatment across covariate values at a rate which is faster than sampling variation. To see why \eqref{eq:fast-m} holds for a highly stratified design, let $\Omega(X_i) = E[m(X_i, 1,R_i(1),\theta_0) - m(X_i, 0,R_i(0),\theta_0)|X_i]$ and first note that, by Assumption \ref{ass:unconfounded},
\[E\bigg[\frac{1}{\sqrt{n}}\sum_{1 \le i \le n}(A_i - \eta)\Omega(X_i) \bigg \vert X^{(n)}\bigg] = 0~.\]
Next, for $1 \leq s \leq d_\theta$, let $\Omega^{(s)}(X_i)$ denote the $s$th component of $\Omega(X_i)$. Then it can be shown using Assumption \ref{ass:a} and \ref{ass:normal}(f) that for $1 \leq s \leq d_\theta$,
\[\var\bigg[\frac{1}{\sqrt{n}}\sum_{1 \le i \le n}(A_i - \eta)\Omega^{(s)}(X_i) \bigg \vert X^{(n)}\bigg] \le C^2\frac{\ell(k - \ell)}{k - 1}\left(\frac{1}{n}\sum_{1 \le j \le n/k} \max_{i, i' \in \lambda_j} \|X_i - X_{i'}\|^2\right)~,\]
where $C$ denotes the Lipschitz constant in Assumption \ref{ass:normal}(f), and so the conditional variance converges in probability to zero under Assumption \ref{ass:pair}. The fast-balancing condition \eqref{eq:fast-m} then follows by an application of Chebyshev's inequality conditional on $X^{(n)}$ and the dominated convergence theorem. In Section \ref{sec:regular}, we further argue that the fast-balancing condition \eqref{eq:fast-m} is in fact a \emph{necessary} condition which a given experimental design must satisfy to ensure \eqref{eq:by_design}. 

\begin{remark}\label{rem:psi_examples_ate}
Note it follows from \eqref{eq:moments} that
\begin{equation} \label{eq:meanzero}
\eta E_Q[m(X_i, 1, R_i(1), \theta_0)] + (1 - \eta) E_Q[m(X_i, 0, R_i(0), \theta_0)] = E_P[m(X_i, A_i, R_i, \theta_0)] = 0~,     
\end{equation}
so that $E[\psi^\ast(X_i, A_i, R_i, \theta_0)] = 0$. It is further straightforward to show using Assumption \ref{ass:unconfounded} that
\begin{align}
    \label{eq:var} \mathbb V_\ast & = \var[\psi^\ast(X_i, A_i,R_i, \theta_0)] \\
    \nonumber & = M^{-1} \big ( E \big [ \eta \var[m(X_i, 1, R_i(1), \theta_0) | X_i] + (1 - \eta) \var[m(X_i, 0, R_i(0), \theta_0) | X_i] \big ] \\
    \nonumber & \hspace{2em} + \var \big [ \eta E[m(X_i, 1, R_i(1), \theta_0) | X_i] + (1 - \eta) E[m(X_i, 0, R_i(0), \theta_0) | X_i] \big ] \big ) (M^{-1})'
\end{align}
For instance, in the special case of the ATE (Example \ref{ex:ATE}) we obtain that 
\begin{align}
     \nonumber \var[\psi^\ast(X_i, A_i,R_i, \theta_0)] & = E\left[\frac{\var[Y_i(1)|X_i]}{\eta} + \frac{\var[Y_i(0)|X_i]}{1 - \eta} \right. \\
    \label{eq:var-ate} & \hspace{5em} + \left. \left(E[Y_i(1)- Y_i(0)|X_i] - E[Y_i(1) - Y_i(0)]\right)^2\right]~,
\end{align}
which matches the asymptotic variance derived in \cite{bai2022inference} for matched pairs. Theorem \ref{thm:normal} however accommodates a much larger class of parameters, including those introduced in Examples \ref{ex:QTE}--\ref{ex:log_odds}.
\end{remark}

\begin{remark} \label{rem:V_compare}
By comparing the variance expression in \eqref{eq:naive_variance} to the variance expression for $\mathbb{V}_*$, we obtain
\begin{equation} \label{eq:V-Vast}
\mathbb{V} - \mathbb{V}_* = \eta (1 - \eta) M^{-1} \var[E[m(X_i, 1, R_i(1), \theta_0) - m(X_i, 0, R_i(0), \theta_0) | X_i]](M^{-1})^{\prime}~,    
\end{equation}
which is positive semidefinite. From this, we conclude that the asymptotic variance of the naive method of moments estimator $\hat{\theta}_n$ is lower in a highly stratified design compared to i.i.d.\ assignment. In Section \ref{sec:semi}, we will further show that $\mathbb V_\ast$ is the lowest possible asymptotic variance among regular estimators for $\theta_0$ in a large class of treatment assignment mechanisms, including both i.i.d.\ assignment and highly stratified designs. When $d_\theta = 1$, we may express $\mathbb V - \mathbb V_\ast$ in terms of the ``nonparametric $R^2$.''  In particular, $\mathbb V - \mathbb V_\ast$ is proportional to $E[R_{g, X}^2(g_i, X_i) \var[g_i]]$, where 
\begin{equation} \label{eq:R2}
R_{g, X}^2(g_i, X_i) = \frac{\var[E[g_i|X_i]]}{\var[g_i]}~,
\end{equation}
and $g_i = m(X_i, 1, R_i(1), \theta_0) - m(X_i, 0, R_i(0), \theta_0)$.  The quantity in \eqref{eq:R2} measures how much of the variation in $g_i$ can be explained nonparametrically by $X_i$. See, for instance, \cite{chernozhukov2024long}.
\end{remark}

\begin{remark} \label{rem:complete_rand}
Note that highly stratified designs include as a special case completely randomized designs, i.e., experiments in which a fixed proportion of the entire sample is assigned to treatment uniformly at random. To see this, consider, for instance, simply matching on an exogenously generated covariate. In this special case, we find from \eqref{eq:V-Vast} that $\mathbb{V}_{\ast} = \mathbb{V}$ as defined in Section \ref{sec:setup}. In this way, we see that completely randomized experiments are asymptotically no more efficient than i.i.d.\ assignment.
\end{remark}

\subsection{Variance Estimation} \label{sec:var-est}
In this subsection, we provide a consistent variance estimator for the asymptotic variance $\mathbb V_\ast$ in \eqref{eq:var}. We suppose that a consistent estimator $\widehat M_n$ for $M$ is available, i.e., $\widehat M_n \xrightarrow{P} M$. In examples where $m$ is differentiable in $\theta$, including Examples \ref{ex:ATE} and \ref{ex:LATE}--\ref{ex:log_odds}, the analog principle suggests that a natural estimator for $M$ is given by
\[ \widehat M_n = \frac{1}{n} \sum_{1 \leq i \leq n} \frac{\partial}{\partial \theta'} m(X_i, A_i, R_i, \theta)\bigg|_{\theta = \hat \theta_n}~. \]
In examples including Example \ref{ex:QTE} where $m$ is nonsmooth in $\theta$, $M$ may consist of components that require nonparametric estimators. See, for instance, \cite{jiang2021bootstrap}.

It then suffices to construct a consistent estimator for the ``meat" in \eqref{eq:var}. To motivate such an estimator, consider the expression in \eqref{eq:var} when $d_\theta = 1$. By the law of total variance, this middle component equals $\Sigma_1 + \Sigma_2$, where
\begin{align*}
\Sigma_1 & = \eta \var[m(X_i, 1, R_i(1), \theta_0)] + (1 - \eta) \var[m(X_i, 0, R_i(0), \theta_0)] \\
\Sigma_2 & = - \eta (1 - \eta) E \big [ \big ( E[m(X_i, 1, R_i(1), \theta_0) | X_i] - E[m(X_i, 1, R_i(1), \theta_0)] \\
& \hspace{5em} - (E[m(X_i, 0, R_i(0), \theta_0) | X_i] - E[m(X_i, 0, R_i(0), \theta_0)]) \big )^2 \big ] \\
& = - \eta (1 - \eta) \Big (E[E[m(X_i, 1, R_i(1), \theta_0) | X_i]^2] + E[E[m(X_i, 0, R_i(0), \theta_0) | X_i]^2] \\
& \hspace{7em} - 2E[E[m(X_i, 1, R_i(1), \theta_0) | X_i] E[m(X_i, 0, R_i(0), \theta_0) | X_i]] \\
& \hspace{7em} - (E[m(X_i, 1, R_i(1), \theta_0)] - E[m(X_i, 0, R_i(0), \theta_0)])^2  \Big )~.
\end{align*}
For $a \in \{0, 1\}$, define
\[ \hat \mu_n(a) = \frac{1}{\eta_a n} \sum_{1 \leq i \leq n} I \{A_i = a\} m(X_i, A_i, R_i, \hat \theta_n)~, \]
where $\eta_1 = \eta$ and $\eta_0 = 1 - \eta$.
The analog principle suggests that a natural estimator for $\Sigma_1$ is
\begin{align*}
\hat \Sigma_{1, n} = \frac{1}{n} \sum_{1 \leq i \leq n} I \{A_i = 1\} (m(X_i, A_i, R_i, \hat \theta_n) - \hat \mu_n(1)) (m(X_i, A_i, R_i, \hat \theta_n) - \hat \mu_n(1))' \\
\hspace{3em} + \frac{1}{n} \sum_{1 \leq i \leq n} I \{A_i = 0\} (m(X_i, A_i, R_i, \hat \theta_n) - \hat \mu_n(0)) (m(X_i, A_i, R_i, \hat \theta_n) - \hat \mu_n(0))'~.
\end{align*}
To estimate $\Sigma_2$, we first define
\begin{align*}
\hat \varsigma_n(1, 0) & = \frac{k}{n} \sum_{1 \leq j \leq n / k} \frac{1}{\ell(k - \ell)} \sum_{i, i' \in \lambda_j: A_i = 1, A_{i'} = 0} m(X_i, A_i, R_i, \hat \theta_n) m(X_{i'}, A_{i'}, R_{i'}, \hat \theta_n)'
\end{align*}
and $\hat \varsigma_n(0, 1)$ similarly. Next, define
\[ \hat \varsigma_n(1, 1) = \begin{cases}
\frac{k}{n} \sum\limits_{1 \leq j \leq n / k} \frac{1}{\binom{\ell}{2}} \sum\limits_{i < i' \in \lambda_j: A_i = A_{i'} = 1} m(X_i, A_i, R_i, \hat \theta_n) \\
\hspace{6cm} \times m(X_{i'}, A_{i'}, R_{i'}, \hat \theta_n)' & \text{ if } \ell > 1 \\
\frac{2k}{n} \sum\limits_{1 \leq j \leq \frac{n}{2k}} {\sum\limits_{i \in \lambda_{2j}, i' \in \lambda_{2j - 1}: A_i = A_{i'} = 1}} m(X_i, A_i, R_i, \hat \theta_n) \\
\hspace{6cm} \times m(X_{i'}, A_{i'}, R_{i'}, \hat \theta_n)' & \text{ if } \ell = 1~.
\end{cases} \]
Similarly, define
\[ \hat \varsigma_n(0, 0) = \begin{cases}
\frac{k}{n} \sum\limits_{1 \leq j \leq n / k} \frac{1}{\binom{k - \ell}{2}} \sum\limits_{i < i' \in \lambda_j: A_i = A_{i'} = 0} m(X_i, A_i, R_i, \hat \theta_n) \\
\hspace{6cm} \times m(X_{i'}, A_{i'}, R_{i'}, \hat \theta_n)' & \text{ if } k - \ell > 1 \\
\frac{2k}{n} \sum\limits_{1 \leq j \leq \frac{n}{2k}} {\sum\limits_{i \in \lambda_{2j}, i' \in \lambda_{2j - 1}: A_i = A_{i'} = 0}} m(X_i, A_i, R_i, \hat \theta_n) \\
\hspace{6cm} \times m(X_{i'}, A_{i'}, R_{i'}, \hat \theta_n)' & \text{ if } k - \ell = 1~.
\end{cases} \]
Finally, define
\[ \hat \Sigma_{2, n} = - \eta(1 - \eta) \big ( \hat \varsigma_n(1, 1) + \hat \varsigma_n(0, 0) - \hat \varsigma_n(1, 0) - \hat \varsigma_n(0, 1) - (\hat \mu_n(1) - \hat \mu_n(0)) (\hat \mu_n(1) - \hat \mu_n(0))' \big )~. \]
The estimator $\hat \varsigma_n(1, 1)$ is constructed in one of two ways depending on the number of treated units in each block. If more than one unit in each block is treated, then we take the averages of all pairwise products of the treated units in each block, and average them across all blocks. We call this a ``within block" estimator. If instead only one unit in each block is treated, then we take the product of two treated units in \emph{adjacent} blocks. We call this a ``between block" estimator, and note that similar constructions have been used previously in \cite{abadie2008estimation}, \cite{bai2022inference}, \cite{bai2024inference-1}, and \cite{cytrynbaum2023designing}. The estimator $\hat \varsigma_n(0, 0)$ is constructed similarly.
A natural estimator for $\mathbb{V}_\ast$ is then given by
\[\hat{\mathbb{V}}_n = \widehat{M}_n^{-1} \left(\hat{\Sigma}_{1,n} + \hat{\Sigma}_{2,n}\right) \left (\widehat{M}_n^{-1} \right )'~. \]
Note that for specific choices of $m(\cdot)$, $\hat{\mathbb{V}}_n$ recovers estimators which have been studied in prior work on inference in highly stratified experiments.  For instance, in the case of matched pairs with $(\ell, k) = (1,2)$ and $m(\cdot)$ as in Example \ref{ex:ATE}, so that $\theta_0$ is the ATE, $\hat{\mathbb{V}}_n$ exactly coincides with the estimator defined in equation (28) of \cite{bai2022inference}.

In addition to Assumption \ref{ass:pair}, we will now also require that the distances between units in \emph{adjacent} blocks be ``close" in terms of their baseline covariates:
\begin{assumption} \label{ass:pairsofblocks}
The blocks used in determining treatment status satisfy
\[ \frac{1}{n} \sum_{1 \leq j \leq \lfloor n / 2 \rfloor} \max_{i \in \lambda_{2j - 1}, i' \in \lambda_{2j}} \|X_i - X_{i'}\|^2 \stackrel{P}{\to} 0~. \]
\end{assumption}

Note that given blocks which satisfy Assumption \ref{ass:pair}, it is always possible to re-order the blocks such that the pairs $\{(\lambda_{2j-1}, \lambda_{2j})\}_{1 \le j \le \lfloor n/2 \rfloor}$ satisfy Assumption \ref{ass:pairsofblocks}, as long as we maintain the sufficient condition that $E[\|X_i\|^d] < \infty$ for $d \geq \mathrm{dim}(X_i) + 1$. This property could be achieved, for instance, by applying the {\tt nbpmatching} algorithm to the block-means $\{\bar{X}_j\}_{1 \le j \le n/k}$, where $\bar{X}_j := \frac{1}{k}\sum_{i \in \lambda_j}X_i$: see Lemma A.6 in \cite{cytrynbaum2023covariate} for details.

To formally establish the consistency of $\hat{\mathbb V}_n$, we impose the following mild uniform integrability condition, as well as a Glivenko-Cantelli and uniform Lipschitz condition. All three conditions need only hold in an arbitrarily small neighborhood of $\theta_0$.
\begin{assumption} \label{ass:var-est}
There exists $\delta > 0$ such that
\begin{enumerate}[(a)]
    \item For $a \in \{0, 1\}$,
    \begin{equation*}
        \lim_{\lambda \to \infty} E \left [ \sup_{\theta \in \Theta: \|\theta - \theta_0\| < \delta} \|m(X_i, a, R_i(a), \theta)\|^2 I \left \{ \sup_{\theta \in \Theta} \|m(X_i, a, R_i(a), \theta)\| > \lambda \right \} \right ] = 0~.
    \end{equation*}    
    \item $\{E[m_s(X_i, a, R_i(a), \theta) | X_i = x]: \|\theta - \theta_0\| < \delta\}$ and $\{E[m_s(X_i, a, R_i(a), \theta) m(X_i, a, \allowbreak R_i(a), \theta)'| X_i = x]: \|\theta - \theta_0\| < \delta\}$ are $Q$-Glivenko Cantelli for $1 \leq s \leq d_\theta$.
    \item For $a \in \{0, 1\}$, each component of $E[m(X, a, R(a), \theta) | X = x]$ and $E[m(X, a, R(a), \theta) \allowbreak m(X, a, R(a), \theta)' | X = x]$ is Lipschitz with a common Lipschitz constant across $\{\theta \in \Theta: \|\theta - \theta_0\| < \delta\}$.
\end{enumerate}
\end{assumption}

Assumption \ref{ass:var-est}(a) is a mild uniform integrability condition for the envelope function of the moment function in an arbitrarily small neighborhood of $\theta_0$. Assumption \ref{ass:var-est}(b) is a mild condition that requires the conditional expectation of the moment functions to satisfy a uniform law of large numbers. See page 81 of \cite{van_der_vaart1996weak} for a definition of a Glivenko-Cantelli class of functions. Assumption \ref{ass:var-est}(c) strengthens Assumption \ref{ass:normal}(f) to hold uniformly in an arbitrarily small neighborhood of $\theta_0$.

The following theorem establishes the consistency of $\hat{\mathbb V}_n$ for $\mathbb V_\ast$:

\begin{theorem} \label{thm:var-est}
Suppose the treatment assignment mechanism satisfies Assumptions \ref{ass:a}, \ref{ass:pair}, and \ref{ass:pairsofblocks} and the moment functions satisfy Assumptions \ref{ass:normal} and \ref{ass:var-est}. Further suppose $\widehat M_n \xrightarrow{P} M$. Then, $\hat{\mathbb V}_n \xrightarrow{P} \mathbb V_\ast$. 
\end{theorem}

\section{An Efficiency Bound and the Necessity of ``Fast-Balancing''}\label{sec:semi-reg}
In Section \ref{sec:semi} we establish that $\mathbb{V}_{\ast}$ is the efficiency bound for a large class of experimental designs. As a consequence, we can conclude that highly stratified designs are asymptotically efficient ``by design.'' Building on this result, in Section \ref{sec:regular} we establish that a necessary condition for achieving the bound $\mathbb{V}_{\ast}$ when estimating $\theta_0$ using the na\"ive method of moments estimator is that the experimental design be fast-balancing, in the sense of \eqref{eq:fast-m}.

\subsection{Efficiency Bound}\label{sec:semi}
An inspection of the asymptotic variance in \eqref{eq:Vstar} reveals that $\mathbb V_\ast$ in fact coincides with the classical efficiency bound for estimating $\theta_0$ with i.i.d.\ assignment. For example, the variance derived in \eqref{eq:var-ate} coincides with the efficiency bound derived in \cite{hahn1998role} for estimating the ATE with a known marginal treatment probability $\eta$.   Therefore, another way to interpret our result in Theorem \ref{thm:normal} is that the standard i.i.d.\ efficiency bound can be attained by a na\"ive method of moments estimator under a highly stratified design. On the other hand, because treatment status is not independent in a highly stratified design, a natural follow-up question is whether or not the efficiency bound for estimating $\theta_0$ changes relative to what can be obtained under i.i.d.\ assignment once we allow for more general assignment mechanisms. In this section, we show that $\mathbb{V}_*$ continues to be the efficiency bound for the class of parameters introduced in Section \ref{sec:setup}, while allowing for a more general class of treatment assignment mechanisms.  The main restriction on treatment assignment is given by Assumption \ref{ass:unconfounded}, which requires the marginal treatment probability to be known and equal to $\eta$.  As mentioned earlier and explained in Remark \ref{rem:generaleta} below, it is possible to relax this requirement so that $\eta$ can be replaced by a known function $\eta(X_i)$.  For a discussion of how our efficiency bound compares with other results in the literature, see Remark \ref{rem:other_SPEB}. 

We impose the following high-level assumption on the assignment mechanism:

\begin{assumption} \label{ass:LLN}
The treatment assignment mechanism is such that for any integrable function $\gamma: \mathbf R^{d_x} \to \mathbf R$,
\[ \frac{1}{n} \sum_{1 \leq i \leq n} A_i \gamma(X_i) \stackrel{P}{\to} \eta E[\gamma(X_i)]~. \]
\end{assumption}
\noindent In words, Assumption \ref{ass:LLN} requires that the assignment mechanism admits a law of large numbers for integrable functions of the covariate values. Examples \ref{ex:iid}--\ref{ex:rerand} illustrate that the assumption holds for common treatment assignment mechanisms used in practice.

\begin{example}[i.i.d.\ assignment]\label{ex:iid}
Let $A^{(n)}$ be assigned i.i.d., independently of $X^{(n)}$, such that $P\{A_i = 1\} = \eta$. Then it follows immediately by the law of large numbers that Assumption \ref{ass:LLN} is satisfied.
\end{example}

\begin{example}[Covariate-adaptive randomization (CAR)]\label{ex:CAR}
Let $S: \mathbf R^{d_x} \to \mathcal S = \{1, \ldots, \allowbreak |\mathcal S|\}$ be a function that maps the covariates into a fixed, finite set of discrete strata.  We call such a stratification ``coarse", to distinguish it from highly stratified designs as defined in Section \ref{sec:variance}. Define $S_i = S(X_i)$ and assume that treatment status is assigned so that
\[ (R^{(n)}(1), R^{(n)}(0), X^{(n)}) \indep A^{(n)} \big | S^{(n)}~, \]
and that for $s \in \mathcal S$,
\[ \frac{\sum_{1 \leq i \leq n} I \{S_i = s, A_i = 1\}}{\sum_{1 \leq i \leq n} I \{S_i = s\}} \stackrel{P}{\to} \eta~. \]
This high-level assumption accommodates a large class of stratified assignment mechanisms, including stratified biased coin designs \citep[]{efron1971forcing, wei1978adaptive}, minimization methods \citep[][]{pocock1975sequential,hu2012asymptotic} and stratified block randomization \citep[see][for an early discussion]{zelen1974randomization}. It follows from Lemma C.4 in \cite{bugni2019inference} that for any integrable function $\gamma(\cdot)$,
\[ \frac{1}{n} \sum_{1 \leq i \leq n} A_i \gamma(X_i) \stackrel{P}{\to} 
\eta \sum_{s \in \mathcal S} P \{S_i = s\} E[\gamma(X_i)| S_i = s] = \eta E[\gamma(X_i)]~. \]
Therefore, Assumption \ref{ass:LLN} is satisfied.
\end{example}

\begin{example}[CAR with general covariate features] \label{ex:ma}
\cite{ma2024new} propose a family of covariate adaptive randomization procedures which assign treatment sequentially based on an imbalance metric defined by ``feature maps'' of (potentially continuous) covariates. It follows by Theorem 3.5 of their paper that Assumption \ref{ass:LLN} is satisfied under appropriate conditions.
\end{example}

\begin{example}[Matched pairs] \label{ex:MP}
Suppose $n$ is even and we assign treatment using a highly stratified design with $(\ell, k) = (1,2)$. As discussed at the beginning of Section \ref{sec:variance}, such a design is also known as a matched pairs design. Assume that the pairing algorithm $\pi_n(X^{(n)})$ results in pairs that are close in the sense of Assumption \ref{ass:pair}. It then follows from the same argument used to establish \eqref{eq:fast-m} in the discussion following Theorem \ref{thm:normal} that for any Lipschitz integrable function $\gamma(\cdot)$,
\[ \frac{1}{n} \sum_{1 \leq i \leq n} A_i \gamma(X_i) \stackrel{P}{\to} \frac{1}{2} E[\gamma(X_i)]~. \]
By approximating integrable functions by Lipschitz integrable functions as in Lemma A.1 in \cite{hanneke2021universal}, it can be shown that the convergence holds for any integrable function $\gamma(\cdot)$. Therefore, Assumption \ref{ass:LLN} is satisfied.
\end{example}

\begin{example}[Re-randomization]\label{ex:rerand}
Re-randomization is an assignment mechanism in which researchers specify a balance criterion for the covariates, and then repeatedly generate assignments using a completely randomized design until an assignment is found which achieves an acceptable covariate distribution according to the balance criterion. The properties of re-randomization procedures have been studied in \cite{li2017general,li2020rerandomization}, \cite{li2018asymptotic,li2020rerandomization-1}, and \cite{cytrynbaum2024finely}. It follows from Corollary 3.7 in \cite{cytrynbaum2024finely} that Assumption \ref{ass:LLN} holds for re-randomization designs, under appropriate assumptions.
\end{example}

We now present an efficiency bound for the parameter $\theta_0$ introduced in Section \ref{sec:setup}. Formally, we characterize the bound via a convolution theorem that applies to all regular estimators of the parameter $\theta_0$. Following the definition on page 365 of \cite{van_der_vaart1998asymptotic}, by a regular estimator, we mean an estimator whose asymptotic distribution is invariant to ``local" perturbations of the data generating process:
\[ \sqrt n(\tilde \theta_n - \theta(P_{t/\sqrt n, g})) \xrightarrow{P_{t/\sqrt n, g}} L ~, \]
where $P_{t/\sqrt n,g}$ represents a ``local" perturbation of the distribution $P$ along a ``path'' with ``score" $g$. We leave the precise definition of regularity and related assumptions to Supplement A.2. In the paragraph following the statement of the theorem we provide some more details on the nature of our result.

\begin{theorem} \label{thm:efficiencybound}
Suppose Assumptions \ref{ass:unconfounded}, \ref{ass:normal}(b), and \ref{ass:LLN} hold, as well as Condition A.1 described in Supplement A.2. Further suppose $\mathbb V_\ast < \infty$. Let $\tilde{\theta}_n$ be any regular estimator of the parameter $\theta_0$ in the sense of (S.16) in Supplement A.2. Then,  
\[\sqrt{n}(\tilde{\theta}_n - \theta_0) \xrightarrow{d} L~,\]
where
\[L = N(0, \mathbb V_\ast) \ast B~, \]
for $\mathbb V_\ast$ in \eqref{eq:Vstar} and some fixed probability measure $B$ which is specific to the estimator $\tilde{\theta}_n$.
\end{theorem}

Given Theorem \ref{thm:efficiencybound} we call $\mathbb V_\ast = \var[\psi^\ast(X_i, A_i,R_i, \theta_0)]$ the efficiency bound for $\theta_0$, since our result shows that this is the lowest asymptotic variance attainable by any regular estimator under our assumptions. Indeed, it follows from Anderson's lemma \citep[Lemma 8.5 in][]{van_der_vaart1998asymptotic} that the asymptotic loss of any regular estimator is bounded below by the loss under $N(0, \mathbb V_\ast)$ for any ``bowl-shaped'' loss function (including, in particular, square loss). We note that our assumptions on the assignment mechanism preclude us from immediately appealing to standard semi-parametric convolution theorems \citep[see, for instance, Theorem~25.20 in][]{van_der_vaart1989asymptotic}. Instead, we proceed by justifying an application of Theorem 3.1 in \cite{armstrong2022asymptotic} combined with the convolution Theorem 3.11.2 in \cite{van_der_vaart1996weak} to each $d_\theta$-dimensional parametric submodel separately, and then arguing that the supremum over all such submodels is attained by $\var[\psi^{\ast}]$. A key observation is that in order to apply Theorem 3.1 in \cite{armstrong2022asymptotic} to argue that the likelihood ratio process is locally asymptotically normal, the conditional information needs to settle down in the limit, which is guaranteed as long as Assumption \ref{ass:LLN} is satisfied.

\begin{remark}\label{rem:psi_examples}
Following similar arguments as those in Remark \ref{rem:psi_examples_ate}, we can deduce that our efficiency bound agrees with well-known bounds for common parameters (like those presented in Examples \ref{ex:ATE}--\ref{ex:LATE}) in the setting of i.i.d.\ assignment. For example, we have noted in the case of the ATE (Example \ref{ex:ATE}) that \eqref{eq:var-ate} matches the efficiency bound under i.i.d.\ assignment derived in \cite{hahn1998role}. See \cite{rafi2023efficient} and \cite{armstrong2022asymptotic} for related results in the context of stratified and response-adaptive experiments.  Straightforward calculation also implies that, for the quantile treatment effect (Example \ref{ex:QTE}), the efficiency bound is given by
\begin{multline*}
    E \bigg [ \frac{1}{\eta} \frac{F_1 \big (\theta_0(1) | X_i \big ) \big (1 - F_1 \big (\theta_0(1) | X_i \big ) \big )}{f_1 \big (\theta_0(1) \big )^2} + \frac{1}{1- \eta} \frac{F_0 \big (\theta_0(0) | X_i \big ) \big (1 - F_0 \big (\theta_0(0) | X_i \big ) \big )}{f_0 \big (\theta_0(0) \big )^2} \\
    + \bigg ( \frac{F_1 \big (\theta_0(1) | X_i \big ) - \tau}{f_1 \big (\theta_0(1) \big )} - \frac{F_0 \big (\theta_0(0) | X_i \big ) - \tau}{f_0 \big (\theta_0(0) \big )} \bigg)^2 \bigg ]~,
\end{multline*}
which matches the efficiency bound under i.i.d.\ assignment derived in \cite{firpo2007efficient} when the propensity score is set to $\eta$. 
\end{remark}

\begin{remark} \label{rem:adj}
The efficiency bound in Theorem \ref{thm:efficiencybound} is attained by highly stratified experiments as in Theorem \ref{thm:normal} if no additional covariates are available for estimation beyond the set of covariates $X_i$ used in the design. In practice, researchers may consider adjusting for additional baseline covariates in order to improve efficiency. Suppose additional covariates $W^{(n)}$ are available and Assumption \ref{ass:a} is modified such that
\[ (R^{(n)}(1), R^{(n)}(0), W^{(n)}) \indep A^{(n)} \big | X^{(n)}~. \]
It can be shown that the efficiency bound, allowing for additional covariate adjustment based on $X_i$ and $W_i$, is
\begin{equation} \label{eq:adj}
    \mathbb V_\ast - \eta(1 - \eta) M^{-1} E[\var[E[g_i|X_i,W_i]|X_i]] (M^{-1})'~,
\end{equation}
where $g_i = m(X_i, 1, R_i(1), \theta_0) - m(X_i, 0, R_i(0), \theta_0)$. Then, as in Remark \ref{rem:V_compare}, the potential gain in efficiency from exploiting $W_i$ in addition to $X_i$ is proportional to
\[ E[R_{g, X, W}^2(g_i, X_i, W_i) \var[g_i | X_i]]~, \]
where
\begin{equation*}
R_{g, X, W}^2(g_i, X_i, W_i) = \frac{\var[E[g_i|X_i,W_i]|X_i]}{\var[g_i | X_i]}
\end{equation*}
is the nonparametric $R^2$ from regressing $g_i$ on $X_i$ and $W_i$ conditional on $X_i$. As a result, the scope for improving efficiency by adjusting for additional covariates is limited if $R_{g, X, W}^2$ is small. In the case of estimating the ATE,
\[ g_i = \frac{Y_i(1)}{\eta} + \frac{Y_i(0)}{1 - \eta}~, \]
so the scope for improvement depends on how much additional variation in the weighted potential outcomes can be explained by $W_i$ beyond $X_i$.
\end{remark}

\begin{remark}\label{rem:other_SPEB}
Here, we comment on how Theorem \ref{thm:efficiencybound} relates to prior efficiency bounds in experiments with general assignment mechanisms. For the case of estimating the ATE, \cite{armstrong2022asymptotic} derives an efficiency bound over a very large class of assignment mechanisms, including even response-adaptive designs, and shows that the bound is attained when units are assigned to treatment (control) with conditional probability proportional to the conditional standard deviation of the potential outcome under treatment (control).  This type of assignment is sometimes referred to as the Neyman allocation.  On the other hand, our results show that this bound may be quite loose whenever the assignment proportions are restricted to be anything not equal to the Neyman allocation, which is, of course, unknown.  For example, the bound is not informative about what can be achieved if the assignment proportions were set to one half regardless of whether or not the conditional outcome variances across treatment and control are equal.  Such settings frequently arise in practice due to logistical constraints or the absence of pilot data with which to estimate conditional variances of potential outcomes under treatment and control.  Furthermore, as argued in  \cite{cai2022performance}, even if pilot data is available, these quantities may be estimated so poorly that exogenously constraining the assignment proportions to one half leads to more efficient estimates of the ATE in practice.  Motivated by such concerns, \cite{rafi2023efficient} derives an efficiency bound for the ATE over the class of ``coarsely-stratified'' assignment mechanisms studied in \cite{bugni2019inference}, where the stratum-level assignment proportions are restricted {\it a priori} by the experimenter. This framework, however, rules out highly stratified designs. Finally, we once again emphasize that our analysis, unlike these other papers, applies to a general class of treatment effect parameters, including the ATE as a special case.
\end{remark}

\begin{remark} \label{rem:generaleta}
Although we focus on the case where $\eta_i(X_i) = P\{A_i = 1|X_i\} = \eta$ is a constant, the proof of Theorem \ref{thm:efficiencybound} holds when $\eta_i(x) = \eta(x)$ for $1 \leq i \leq n$, where $\eta(x)$ is an arbitrary known and fixed function. In these settings, Lemma A.5 shows that the efficiency bound equals
\begin{equation} \label{eq:var-generaleta}
\begin{split}
\mathbb V_\ast & = \var[\psi^\ast(X_i, A_i,R_i, \theta_0)] \\
& = M^{-1} \big ( E \big [ \eta(X_i) \var[m(X_i, 1, R_i(1), \theta_0) | X_i] \\
& \hspace{5em} + (1 - \eta(X_i)) \var[m(X_i, 0, R_i(0), \theta_0) | X_i] \big ] \\
& \hspace{3em} + \var \big [ \eta(X_i) E[m(X_i, 1, R_i(1), \theta_0) | X_i] \\
& \hspace{5em}+ (1 - \eta(X_i)) E[m(X_i, 0, R_i(0), \theta_0) | X_i] \big ] \big ) (M^{-1})'~,    
\end{split}
\end{equation}
so that the only difference from \eqref{eq:var} is that $\eta$ is replaced by $\eta(X_i)$.  

If we additionally impose that $\eta(X_i) $ takes on a finite set of values $\{\eta_1, \dots, \eta_S\}$, then this bound could be achieved by separately implementing a highly stratified experiment over each set $\{i: \eta(X_i) = \eta_s\}$ for $1 \leq s \leq S$. In other words, separately within each stratum defined by the units for which $\eta(X_i) = \eta_s$, employ the assignment mechanism described in Assumptions \ref{ass:a}--\ref{ass:pair} with $\ell/k = \eta_s$. For more general functions $\eta(\cdot)$, we conjecture one could employ the local randomization procedure in \cite{cytrynbaum2023designing}.
\end{remark}
\subsection{The Necessity of ``Fast-Balancing''} \label{sec:regular}
In this subsection, we provide conditions under which the fast-balancing condition described in \eqref{eq:fast-m} is a necessary condition for efficient estimation of $\theta_0$ ``by design.'' As a supplement, we also provide necessary and sufficient conditions for an asymptotically linear estimator to be regular (in the sense of (S.16) in Supplement A.2) for a large class of treatment assignment mechanisms. Concretely, given an assignment mechanism, suppose $\tilde \theta_n$ is an asymptotically linear estimator for $\theta_0$ in the sense that
\begin{equation} \label{eq:linear}
\sqrt n (\tilde \theta_n - \theta_0) = \frac{1}{\sqrt n} \sum_{1 \leq i \leq n} \psi(X_i, A_i, R_i, \theta_0) + o_P(1)~,
\end{equation}
where $E[\psi(X_i, A_i, R_i, \theta_0)] = 0$ and $\var[\psi(X_i, A_i, R_i, \theta_0)] < \infty$. The results in this section derive necessary and sufficient conditions for $\tilde \theta_n$ to be regular, and further demonstrate that in order for $\tilde{\theta}_n$ to be regular and efficient, either $\psi = \psi^\ast$ or a fast-balancing condition involving $\psi(\cdot)$ needs to be satisfied. Therefore, highly stratified designs are not only sufficient to guarantee efficiency when estimating $\theta_0$ using the na\"ive method of moments estimator, but their fast-balancing property is also necessary.

In order to study the behavior of the estimator under local alternatives, we will impose the following high-level assumption on the ``imbalance'' of the treatment assignments. To describe the assumption, let $\rho$ denote any metric that metrizes weak convergence.
\begin{assumption} \label{ass:imbalance}
The treatment assignment mechanism is such that for any square-integrable function $\gamma: \mathbf R^{d_x} \to \mathbf R^{d_\theta}$ with $E[\gamma(X_i)] = 0$ ,
\[ \rho \bigg ( \mathcal L \Big ( \frac{1}{\sqrt n} \sum_{1 \leq i \leq n} (A_i - \eta) \gamma(X_i) \Big \vert X^{(n)} \Big ), ~N(0, V_\gamma^{\rm imb}) \bigg ) \xrightarrow{P} 0 \]
for some deterministic variance $V_\gamma^{\rm imb}$, where $\mathcal L(\cdot | X^{(n)})$ denotes the conditional distribution given $X^{(n)}$.
\end{assumption}

In the following examples, we discuss Assumption \ref{ass:imbalance} in the context of some common treatment assignment mechanisms.

\begin{example}
Revisiting Example \ref{ex:iid}, let $A^{(n)}$ be assigned i.i.d., independently of $X^{(n)}$, such that $P \{A_i = 1\} = \eta$. Then, by verifying the conditions of the Lindeberg-Feller CLT conditional on $X^{(n)}$, it can be shown that Assumption \ref{ass:imbalance} is satisfied with $V^{\rm imb}_\gamma = \eta (1 - \eta) \var[\gamma(X_i)]$.
\end{example}

\begin{example}
Revisiting Example \ref{ex:CAR}, suppose treatment status is assigned using stratified block randomization, which is a special case of covariate-adaptive randomization where $A^{(n)}$ is such that \[\sum_{1 \leq i \leq n} A_i I \{S_i = s\} = \bigg\lfloor \eta \sum_{1 \leq i \leq n} I \{S_i = s\} \bigg\rfloor~,\] with all such assignments being drawn uniformly at random and independently across strata. It follows from Theorem 12.2.1 in \cite{lehmann2022testing} combined with a subsequencing argument that Assumption \ref{ass:imbalance} is satisfied with $V_\gamma^{\rm imb} = \eta (1 - \eta) E[\var[\gamma(X_i) | S_i]]$.
\end{example}

\begin{example} \label{ex:ma2}
Revisiting Example \ref{ex:ma}, suppose treatment is assigned using the covariate adaptive randomization procedure described in \cite{ma2024new}. Then it follows from Theorem 3.6 in their paper that, under appropriate assumptions,
\[\frac{1}{\sqrt{n}}\sum_{1 \le i \le n}(A_i - \eta)\gamma(X_i) \xrightarrow{d} N(0, \tilde{V})~,\]
for some variance $\tilde{V}$. Note, however, that this result is \emph{not} conditional on $X^{(n)}$ and thus does not immediately imply Assumption \ref{ass:imbalance}. We conjecture that a similar result could be established conditional on $X^{(n)}$ and thus Assumption \ref{ass:imbalance} would be satisfied.
\end{example}

\begin{example} \label{ex:MP2}
Revisiting Example \ref{ex:MP}, suppose $n$ is even and we assign treatment using a matched pairs design. It then follows by arguing as in the discussion following Theorem \ref{thm:normal} that for any square-integrable Lipschitz function $\gamma(\cdot)$,
\[ \var\bigg[\frac{1}{\sqrt{n}} \sum_{1 \leq i \leq n}(A_i - \eta) \gamma(X_i)\bigg\vert X^{(n)}\bigg] \stackrel{P}{\to} 0~. \]
Therefore, by Markov's inequality, Assumption \ref{ass:imbalance} is satisfied with
$V_\gamma^{\rm imb} = 0$. By approximating square-integrable functions by square-integrable Lipschitz functions as in Lemma C.5 in \cite{cytrynbaum2023designing}, it can be shown that the convergence holds for any square-integrable function $\gamma(\cdot)$. 
\end{example}

\begin{example}\label{ex:rerand2}
Revisiting Example \ref{ex:rerand}, we note that, following Corollary 3.7 in \cite{cytrynbaum2024finely}, we do not expect Assumption \ref{ass:imbalance} to hold for re-randomization designs in general.
\end{example}

We are now ready to state a theorem that characterizes all regular asymptotically linear estimators and establishes the necessity of the fast-balancing condition for efficient estimation using $\hat{\theta}_n$.

\begin{theorem} \label{thm:regular}
Suppose the treatment assignment mechanism satisfies Assumptions \ref{ass:unconfounded} and \ref{ass:LLN}--\ref{ass:imbalance}. Suppose $\tilde \theta_n$ is an asymptotically linear estimator for $\theta_0$ in the sense of \eqref{eq:linear}. Then, $\tilde \theta_n$ is regular if and only if
\[ \psi(x, a, r, \theta_0) = \psi^\ast(x, a, r, \theta_0) + \psi^\perp(x, a, \theta_0)~, \]
for some function $\psi^\perp$ such that $E[\psi^\perp(X_i, A_i, \theta_0) | X_i] = \eta \psi^\perp(X_i, 1, \theta_0) + (1 - \eta) \psi^\perp(X_i, 0, \allowbreak \theta_0) = 0$. Furthermore, if $\tilde \theta_n$ is regular, it attains the efficiency bound if and only if
\begin{equation} \label{eq:fast}
\frac{1}{\sqrt n} \sum_{1 \leq i \leq n} (A_i - \eta) E[\psi(X_i, 1, R_i(1), \theta_0) - \psi(X_i, 0, R_i(0), \theta_0) | X_i] = o_P(1)~.
\end{equation}
\end{theorem}

To further understand condition \eqref{eq:fast}, note that $E[\psi^\ast(X_i, 1, R_i(1)) | X_i] = E[\psi^\ast(X_i, 0, \allowbreak R_i(0)) | X_i]$, so
\begin{align*}
& \frac{1}{\sqrt n} \sum_{1 \leq i \leq n} (A_i - \eta) E[\psi(X_i, 1, R_i(1), \theta_0) - \psi(X_i, 0, R_i(0), \theta_0) | X_i] \\
& = \frac{1}{\sqrt n} \sum_{1 \leq i \leq n} (A_i - \eta) (\psi^\perp(X_i, 1, \theta_0) - \psi^\perp(X_i, 0, \theta_0)) \\
& = \frac{1}{\sqrt n} \sum_{1 \leq i \leq n} \psi^\perp(X_i, A_i, \theta_0)~,
\end{align*}
where the last equality follows from the fact that $E[\psi^\perp(X_i, A_i, \theta_0) | X_i] = 0$. As a result, \eqref{eq:fast} holds either when the treatment assignment mechanism groups units with similar values of $\psi^\perp(X_i, 1, \theta_0) - \psi^\perp(X_i, 0, \theta_0)$, or when the estimator is based on the efficient influence function, so that $\psi^\perp = 0$. Recall that as an intermediate step in the proof of Theorem \ref{thm:normal}, we showed in \eqref{eq:vdv} that for the na\"ive method of moments estimator $\hat \theta_n$, $\psi(\cdot) = - M^{-1} m(\cdot)$, so \eqref{eq:fast} coincides with the fast-balancing condition in \eqref{eq:fast-m}. We can thus conclude from Theorem \ref{thm:regular} that the fast-balancing condition is necessary to achieve efficient estimation based on the na\"ive method of moments estimator, when the class of assignment mechanisms satisfy Assumption \ref{ass:imbalance}.

\begin{example}
Revisiting Example \ref{ex:ATE}, recall $\hat \theta_n$ estimates the ATE based on the moment conditions in \eqref{eq:moments-ate}. Direct calculation shows that for $\hat \theta_n$,
\[ \psi^\perp(x, a, \theta_0) = (a - \eta) \Big ( \frac{\mu_1(x)}{\eta} + \frac{\mu_0(x)}{1 - \eta} \Big )~. \]
As a result, Theorem \ref{thm:regular} demonstrates that $\hat{\theta}_n$ does not achieve the efficiency bound, unless
\[ \frac{1}{\sqrt n} \sum_{1 \leq i \leq n} (A_i - \eta) \left ( \frac{\mu_1(X_i)}{\eta} + \frac{\mu_0(X_i)}{1 - \eta} \right ) = o_P(1)~, \]
which is indeed the case in highly stratified experiments when the treatment assignment mechanism satisfies Assumption \ref{ass:pair}. \end{example}

\section{Simulations}\label{sec:sims}
In this section, we illustrate the theoretical results in Sections \ref{sec:variance} and \ref{sec:semi-reg} through a simulation study. Throughout this section, we set $\eta = 1/2$, and compare the mean-squared error (MSE), bias, the length of the confidence interval and its coverage rate for the following combinations of treatment assignment mechanisms and estimators:
\begin{enumerate}[--]
    \item i.i.d.\ treatment assignment and the naïve method of moments estimator
    \item i.i.d.\ treatment assignment and covariate adjusted estimators
    \item Matched pairs, i.e., a highly stratified design with $(\ell, k) = (1, 2)$, and the naïve method of moments estimator
    \item Matched quadruplets, i.e., a highly stratified design with $(\ell, k) = (2, 4)$, and the naïve method of moments estimator
\end{enumerate}

In Section \ref{sec:sims_ATE}, we present the model specifications and estimators for estimating the ATE as in Example \ref{ex:ATE}. Supplement B contains the model specifications and estimators for estimating the LATE as in Example \ref{ex:LATE}. Section \ref{sec:sims_results} reports the simulation results for the MSE. The tables and discussion for bias and coverage are contained in Supplement B.
 
\subsection{Average Treatment Effect}\label{sec:sims_ATE}
In this section, we present model specifications and estimators for estimating the ATE as in Example \ref{ex:ATE}. The  Recall that in this case the moment function we consider is given by  
\begin{equation*}
    m(X_i, R_i, A_i, \theta) = \frac{Y_i A_i}{\eta} - \frac{Y_i (1 - A_i)}{1 - \eta} - \theta~,
\end{equation*}
with $R_i = Y_i$. For $a \in \{0, 1\}$ and $1\leq i \leq n$, the potential outcomes are generated according to the equation:
\begin{equation} \label{eq:sims-outcome}
Y_i(a) = \mu_a(X_i) + \epsilon_{i}~,
\end{equation}
where $\mu_0(X_i) = \sum_{1 \leq l \leq 8} w_l \big ( X_{i, l} + \frac{1}{3} (X_{i, l}^2 - 1) \big )$ for $w = (2, 1, 1, 0.5, 0.01, 0.001, 0.0001, \allowbreak 0.00001)$, $\mu_1(X_i) = 0.2 + \mu_0(X_i)$, $\epsilon_i \sim N(0, 4)$, $(X_i, \epsilon_i)$, $1 \leq i \leq n$ are i.i.d., and for each $1 \leq i \leq n$, $(X_i, \epsilon_i)$ are independent. In each of the simulations to follow, when forming pairs/quadruplets or performing regression adjustment, we use only a subvector of the covariates $X_i$ consisting of the first $T$ covariates, for $T \in \{2, 4, 8\}$.

We consider the following three estimators for the ATE:

\begin{description}
    \item[Unadjusted Estimator]
    \begin{equation*}
        \hat\theta_n^{\rm unadj} = \frac{1}{n / 2} \sum_{1\leq i \leq n} ( Y_i A_i - Y_i(1-A_i) )~.
    \end{equation*}
    \item[Adjusted Estimator 1]
    \begin{equation*}
        \hat\theta_n^{\rm adj, 1} = \frac{1}{n} \sum_{1 \leq i \leq n} \big ( 2 A_i(Y_i - \hat\mu_1^Y(X_i)) - 2(1- A_i)(Y_i - \hat\mu_0^Y(X_i)) + \hat\mu_1^Y(X_i) - \hat\mu_0^Y(X_i) \big )~,
    \end{equation*}
    where $\hat\mu_a^Y(X_i)$ is the linear projection of $Y_i$ on $(1, (X_{i, l}, X_{i, l}^2: 1 \leq l \leq T))$ in the subsample with $A_i=a$.
    \item[Adjusted Estimator 2] 
     \begin{equation*}
        \hat\theta_n^{\rm adj, 2} = \frac{1}{n} \sum_{1 \leq i \leq n} \big ( 2 A_i(Y_i - \hat\mu_1^Y(X_i)) - 2(1- A_i)(Y_i - \hat\mu_0^Y(X_i)) + \hat\mu_1^Y(X_i) - \hat\mu_0^Y(X_i) \big )~,
    \end{equation*}
    where $\hat\mu_a^Y(X_i)$ is the linear projection of $Y_i$ on on $(1, (X_{i, l}, X_{i, l}^2, X_{i, l} \allowbreak I\{X_{i, l} > \hat t_l\}: 1 \leq l \leq T))$ in the subsample with $A_i=a$, where $\hat t_l$ is the sample median of $X_{i, l}$, $1 \leq i \leq n$.
\end{description}

The first estimator $\hat\theta_n^{\rm unadj}$ is the naïve method of moments estimator given by the solution to \eqref{eq:est}. The second and third estimators $\hat\theta_n^{\rm adj,1}$ and $\hat\theta_n^{\rm adj,2}$  are covariate-adjusted estimators which can be obtained as two-step method of moments estimators from solving the ``augmented'' moment equation \eqref{eq:augment_m} described in the discussion at the end of Section \ref{sec:setup}. $\hat\theta_n^{\rm adj,1}$ and $\hat\theta_n^{\rm adj,2}$ differ in the choice of basis functions used in the construction of the estimators $\hat{\mu}_a(x)$. Note that by the double-robustness property of the augmented estimating equation \eqref{eq:augment_m}, it can be shown that the adjusted estimators $\hat{\theta}_n^{\rm adj,1}$, $\hat{\theta}_n^{\rm adj,2}$ are consistent and asymptotically normal regardless of the choice of estimators $\hat{\mu}_a(x)$, but consistency of $\hat{\mu}_a(x)$ to $\mu_a(x)$ would ensure that $\hat{\theta}_n^{\rm adj,1}$, $\hat{\theta}_n^{\rm adj,2}$ are efficient under i.i.d.\ assignment \citep[][]{robins1995analysis, tsiatis2008covariate, chernozhukov2017doubledebiasedneyman}.

\subsection{Simulation Results}\label{sec:sims_results}
We focus on the MSE and leave the discussion for bias and coverage to Supplement B. Table \ref{table:sims2} displays the ratio of the empirical MSE for each design/estimator pair relative to the MSE of the unadjusted estimator under i.i.d.\ assignment, computed across $4000$ Monte Carlo replications. As expected given our theoretical results, we find that the empirical MSEs of the na\"ive unadjusted estimator under a matched pairs/quads design closely match the empirical MSEs of the covariate adjusted estimators under i.i.d.\ assignment; this feature is particularly noteworthy given that the adjusted estimators are in fact correctly specified, and the correct specification would be unknowable in practice. We note that the MSE improvement of the unadjusted estimator with a matched pairs/quads design relative to i.i.d.\ assignment is typically worse with 2 or 8 covariates than with 4 covariates: this phenomenon stems from the fact that the first 4 covariates are much stronger predictors of the control outcome than the last 4 covariates, which are almost uninformative. Although we have found in prior work \citep[][]{bai2024inference} that the matched pairs design delivers a lower MSE than the matched quads design when the potential outcomes depend on the covariates linearly, we do not find that this is the case here with a nonlinear model. However, we do consistently find that the MSE with 8 covariates is smaller for the matched pairs design than the matched quads design. This comparison illustrates that, although as discussed in Remark \ref{rem:matched_pairs}, our theoretical results imply matched pairs and matched quads designs are not distinguishable asymptotically, their finite-sample properties may differ, especially when the number of covariates is large. In particular, note Assumption \ref{ass:pair} is more stringent for matched quads, for which $k = 4$, than matched pairs, for which $k = 2$.

\begin{table}[ht!]
\centering
\begin{threeparttable}
\caption{MSE ratios relative to unadjusted estimator under i.i.d.\ assignment}
\setlength{\tabcolsep}{3pt} 
\begin{tabular}{cccccccccc}
\toprule
& & & \multicolumn{3}{c}{\textbf{i.i.d.\ assignment}} & & \textbf{Matched pairs} & & \textbf{Matched tuples} \\ 
\cmidrule{4-6} \cmidrule{8-8} \cmidrule{10-10}
& & \# of covariates ($T$) & Unadjusted & Adjusted 1 & Adjusted 2 & & Unadjusted & & Unadjusted \\ 
\midrule

\multicolumn{10}{c}{$n=100$} \\ \addlinespace[0.5em]
& \multirow{3}{*}{ATE} & 2 & 1.0000 & 0.7691 & 0.7809 & & 0.7608 & & 0.7253 \\
& & 4 & 1.0000 & 0.7166 & 0.7356 & & 0.7123 & & 0.6853 \\
& & 8 & 1.0000 & 0.7527 & 0.8102 & & 0.7619 & & 0.8083 \\ \addlinespace[0.5em]
& \multirow{3}{*}{LATE} & 2 & 1.0000 & 0.7550 & 0.7664 & & 0.7527 & & 0.7069 \\
& & 4 & 1.0000 & 0.6991 & 0.7165 & & 0.6941 & & 0.6682 \\
& & 8 & 1.0000 & 0.7209 & 0.7781 & & 0.7312 & & 0.7676 \\ \addlinespace[1.0em]

\multicolumn{10}{c}{$n=200$} \\ \addlinespace[0.5em]
& \multirow{3}{*}{ATE} & 2 & 1.0000 & 0.7598 & 0.7619 & & 0.7500 & & 0.7254 \\
& & 4 & 1.0000 & 0.7084 & 0.7154 & & 0.7016 & & 0.7213 \\
& & 8 & 1.0000 & 0.7256 & 0.7420 & & 0.7586 & & 0.8042 \\ \addlinespace[0.5em]
& \multirow{3}{*}{LATE} & 2 & 1.0000 & 0.7562 & 0.7581 & & 0.7219 & & 0.7056 \\
& & 4 & 1.0000 & 0.7021 & 0.7095 & & 0.6806 & & 0.7068 \\
& & 8 & 1.0000 & 0.7081 & 0.7238 & & 0.7317 & & 0.7811 \\ \addlinespace[1.0em]

\multicolumn{10}{c}{$n=400$} \\ \addlinespace[0.5em]
& \multirow{3}{*}{ATE} & 2 & 1.0000 & 0.7427 & 0.7441 & & 0.7473 & & 0.7698 \\
& & 4 & 1.0000 & 0.6828 & 0.6880 & & 0.6936 & & 0.7372 \\
& & 8 & 1.0000 & 0.6904 & 0.7013 & & 0.7526 & & 0.7986 \\ \addlinespace[0.5em]
& \multirow{3}{*}{LATE} & 2 & 1.0000 & 0.7420 & 0.7436 & & 0.7528 & & 0.7769 \\
& & 4 & 1.0000 & 0.6792 & 0.6845 & & 0.6928 & & 0.7379 \\
& & 8 & 1.0000 & 0.6864 & 0.6974 & & 0.7501 & & 0.7950 \\ \addlinespace[1.0em]

\multicolumn{10}{c}{$n=1000$} \\ \addlinespace[0.5em]
& \multirow{3}{*}{ATE} & 2 & 1.0000 & 0.7526 & 0.7531 & & 0.7853 & & 0.7338 \\
& & 4 & 1.0000 & 0.6846 & 0.6862 & & 0.7178 & & 0.7123 \\
& & 8 & 1.0000 & 0.6889 & 0.6934 & & 0.7722 & & 0.7969 \\ \addlinespace[0.5em]
& \multirow{3}{*}{LATE} & 2 & 1.0000 & 0.7532 & 0.7537 & & 0.7841 & & 0.7318 \\
& & 4 & 1.0000 & 0.6842 & 0.6859 & & 0.7145 & & 0.7119 \\
& & 8 & 1.0000 & 0.6879 & 0.6925 & & 0.7653 & & 0.7916 \\ \addlinespace[1.0em]

\multicolumn{10}{c}{$n=2000$} \\ \addlinespace[0.5em]
& \multirow{3}{*}{ATE} & 2 & 1.0000 & 0.7509 & 0.7512 & & 0.7367 & & 0.7316 \\
& & 4 & 1.0000 & 0.6926 & 0.6932 & & 0.6847 & & 0.6788 \\
& & 8 & 1.0000 & 0.6933 & 0.6934 & & 0.6975 & & 0.7253 \\ \addlinespace[0.5em]
& \multirow{3}{*}{LATE} & 2 & 1.0000 & 0.7501 & 0.7504 & & 0.7377 & & 0.7337 \\
& & 4 & 1.0000 & 0.6915 & 0.6922 & & 0.6858 & & 0.6788 \\
& & 8 & 1.0000 & 0.6914 & 0.6916 & & 0.6966 & & 0.7223 \\
\bottomrule
\end{tabular}
\begin{tablenotes}
\footnotesize
\item Note: For each specification, the MSE of the unadjusted estimator under i.i.d.\ assignment is normalized to one, and the other columns report the ratios of MSEs relative to this baseline.
\end{tablenotes}
\label{table:sims2}
\end{threeparttable}
\end{table}

\section{Recommendations for Empirical Practice}\label{sec:recs}
We conclude with some recommendations for empirical practice based on our theoretical results. Overall, our findings highlight the general benefit of highly stratified designs for designing efficient experiments: highly stratified experiments ``automatically'' perform fully-efficient covariate adjustment for a large class of interesting parameters. This finding generalizes similar observations made by \cite{bai2022inference}, \cite{bai2022optimality} and \cite{cytrynbaum2023designing} for the special case of estimating the ATE. 

Our simulation evidence suggests, however, that highly stratified experiments may produce less precise estimates than (correctly specified) covariate adjustment when the the dimension of $X_i$ is large relative to the sample size.  For this reason, we recommend that practitioners construct their blocks using a subset of the baseline covariates that they believe have the highest explanatory power in terms of the nonparametric $R^2$ in \eqref{eq:R2}; the pre-treatment measure of the outcomes of interest, for example, is typically believed to be one such covariate \citep[see, in particular,][]{bruhn2009pursuit}. The experimental data can then be analyzed efficiently using an unadjusted method-of-moments estimator. 

If one wishes to perform covariate adjustment with additional covariates beyond those used for blocking, then this can be done {\it ex-post}. As discussed in Remark \ref{rem:adj}, the scope for improvement from covariate adjustment is limited by the nonparametric $R^2$ from the regression of the moment functions on the additional covariates, conditional on the ones used for matching; if one has already matched on the covariates with the highest explanatory power, then the potential gain in efficiency from adjusting for these additional covariates may be limited. We further caution that care must be taken to ensure that the adjustment is performed in such a way that it guarantees a gain in efficiency: see \cite{bai2024covariate} and \cite{cytrynbaum2023covariate} for related discussion. Recent work has developed such methods of covariate adjustment for specific parameters of interest \citep[see, for instance,][]{bai2024covariate, bai2024inference-1, bai2025inference,cytrynbaum2023covariate}, but we leave the development of a method of covariate adjustment which applies at the level of generality considered in this paper to future work.

\clearpage

\appendix

\section{Proofs of Main Results}
\subsection{Proof of Theorem \ref{thm:normal}}
First note \eqref{eq:normal_convergence} follows from \eqref{eq:normal} and Lemma \ref{lem:clt}. In particular, the second component of the decomposition therein is zero because $E[\psi^\ast | X, A] = E[\psi^\ast | X]$. To show \eqref{eq:normal}, we first establish \eqref{eq:vdv}, i.e.,
\[ \sqrt n(\hat \theta_n -  \theta_0) = -M^{-1}\frac{1}{\sqrt{n}}\sum_{1 \le i \le n}m(X_i,A_i,R_i,\theta_0) + o_P(1)~. \]
By the proof of Theorem 5.21 in \cite{van_der_vaart1998asymptotic}, to show \eqref{eq:vdv}, it suffices to show
\begin{equation}\label{eq:se}
\mathbb L_n(\hat{\theta}_n) \stackrel{P}{\to} 0~,
\end{equation}
where $\mathbb L_n(\theta) = (\mathbb L^{(1)}_{n}(\theta), \dots, \mathbb L^{(d_\theta)}_{n}(\theta))'$ for
\begin{align*} 
\mathbb L_n^{(s)}(\theta) & = \frac{1}{\sqrt n} \sum_{1 \leq i \leq n} (m_s(X_i, A_i, R_i, \theta) - E_P[m_s(X_i, A_i, R_i, \theta)]) \\
& \hspace{1em} - \frac{1}{\sqrt n} \sum_{1 \leq i \leq n} (m_s(X_i, A_i, R_i, \theta_0) - E_P[m_s(X_i, A_i, R_i, \theta_0)])~.
\end{align*}
To accomplish this, we study $\mathbb L_n^{(s)}(\theta)$ for $1 \leq s \leq d_\theta$ separately. It follows from Assumption \ref{ass:normal}(c)--(d), Proposition 8.11 in \cite{kosorok2008introduction}, and the arguments to establish \eqref{eq:countable} that
\[ \sup_{\theta \in \Theta: \|\theta - \theta_0\| < \delta} |\mathbb L_n^{(s)}(\theta)| = \sup_{\theta \in \Theta^\ast: \|\theta - \theta_0\| < \delta} |\mathbb 
L_n^{(s)}(\theta)|~. \]
Therefore, since $\hat{\theta}_n \stackrel{P}{\to} \theta_0$ by Lemma \ref{lem:consistency}, to show \eqref{eq:se} it suffices to argue that for every $\epsilon > 0$ and every sequence $\delta_n \downarrow 0$ \citep[p.89 of][]{van_der_vaart1996weak},
\begin{equation}\label{eq:L_equicont}
\lim_{n \to \infty}P\left\{\sup_{\theta \in \Theta^\ast:\|\theta - \theta_0\|<\delta_n}\big|\mathbb{L}^{(s)}_n(\theta)\big| > \epsilon\right\} = 0~.
\end{equation}
Following the arguments in the proof of Lemma \ref{lem:consistency}, we decompose $\mathbb L_n^{(s)}(\theta) = \mathbb L_{n, 1}^{(s)}(\theta) + \mathbb L_{n, 0}^{(s)}(\theta)$, where
\begin{align*}
\mathbb L_{n, 1}^{(s)}(\theta) & = \frac{1}{\sqrt n} \sum_{1 \leq i \leq n} A_i (m_s(X_i, 1, R_i(1), \theta) - m_s(X_i, 1, R_i(1), \theta_0) \\
& \hspace{7.5em} - E[m(X_i, 1, R_i(1), \theta) - m_s(X_i, 1, R_i(1), \theta_0)]) \\
\mathbb L_{n, 0}^{(s)}(\theta) & = \frac{1}{\sqrt n} \sum_{1 \leq i \leq n} (1 - A_i) (m_s(X_i, 0, R_i(0), \theta) - m_s(X_i, 0, R_i(0), \theta_0) \\
& \hspace{7.5em} - E[m_s(X_i, 0, R_i(0), \theta) - m_s(X_i, 0, R_i(0), \theta_0)])~.
\end{align*}
Define
\begin{multline*}
\rho_Q(\theta,\theta_0) = E_Q[(m_s(X, a, R(a), \theta) - m_s(X, a, R(a), \theta_0) \\- E_Q[m_s(X, a, R(a), \theta) - m_s(X, a, R(a), \theta_0)])^2]^{1/2}~.    
\end{multline*}
Note by Assumption \ref{ass:normal}(c) that $\rho_Q(\theta, \theta_0)$ is continuous in $\theta$, i.e., as $\|\theta - \theta_0\| \to 0$,
\[ \rho_Q(\theta,\theta_0)\leq E_Q[(m_s(X, a, R(a), \theta) - m_s(X, a, R(a), \theta_0))^2]^{1/2} \to 0~. \]
Fix any sequence $\tilde \delta_n \downarrow 0$. For every $n$, there exists $n'$ such that $\{\theta \in \Theta^\ast: \|\theta - \theta_0\| < \delta_{n'}\} \subseteq \{\theta \in \Theta^\ast: \rho_Q(\theta, \theta_0) < \tilde \delta_n\}$. By Proposition C.1 in \cite{han2021complex},
\begin{align*}
& E \bigg [ \sup_{\rho_Q(\theta, \theta_0) < \tilde \delta_n} |\mathbb L_{n, a}^{(s)}(\theta)| \bigg ] E \bigg [ \sup_{\rho_Q(\theta, \theta_0) < \tilde \delta_n} \bigg | \frac{1}{\sqrt n} \sum_{1 \leq i \leq n} (m_s(X_i, 1, R_i(1), \theta) \\
& \hspace{3em} - m_s(X_i, 1, R_i(1), \theta_0) - E[(m_s(X_i, 1, R_i(1), \theta) - m_s(X_i, 1, R_i(1), \theta_0)]) \bigg | \bigg ] \to 0~.
\end{align*}
where the convergence follows from Assumption \ref{ass:normal}(e) and Corollary 2.3.12 in \cite{van_der_vaart1996weak}. We then obtain \eqref{eq:L_equicont} by Markov's inequality.

Finally, we derive \eqref{eq:normal} from \eqref{eq:vdv}. Note that
\begin{align*}
& \frac{1}{\sqrt n} \sum_{1 \leq i \leq n} m(X_i, A_i, R_i, \theta_0) \\
& = \frac{1}{\sqrt n} \sum_{1 \leq i \leq n} \Big ( \eta E[m(X_i, 1, R_i(1), \theta_0) | X_i] + (1 - \eta) E[m(X_i, 0, R_i(0), \theta_0) | X_i] \\
& \hspace{3.5em} + I\{A_i = 1\} (m(X_i, 1, R_i, \theta_0) - E[m(X_i, 1, R_i(1), \theta_0) | X_i]) \\
& \hspace{3.5em} + I\{A_i = 0\} (m(X_i, 0, R_i, \theta_0) - E[m(X_i, 0, R_i(0), \theta_0) | X_i]) \\
& \hspace{3.5em} + (A_i - \eta) (E[m(X_i, 1, R_i(1), \theta_0) - m(X_i, 0, R_i(0), \theta_0) | X_i]) \Big )~.
\end{align*}
Let $\Omega(X_i) = E[m(X_i, 1,R_i(1),\theta_0) - m(X_i, 0,R_i(0),\theta_0)|X_i]$ and note that by Assumption \ref{eq:unconfounded},
\[E\bigg[\frac{1}{\sqrt{n}}\sum_{1 \le i \le n}(A_i - \eta)\Omega(X_i) \bigg \vert X^{(n)}\bigg] = 0~.\]
Recall $\Omega^{(s)}(X_i)$ is the $s$th component of $\Omega(X_i)$. Next, it follows from Assumption \ref{ass:a}, \ref{ass:normal}(f), and equation (12.3) in \cite{lehmann2022testing} that for $1 \leq s \leq d_\theta$,
\begin{align*}
\var\bigg[\frac{1}{\sqrt{n}}\sum_{1 \le i \le n}(A_i - \eta)\Omega^{(s)}(X_i) \bigg \vert X^{(n)}\bigg] & = \frac{1}{n} \sum_{1 \le j \le n/k} \frac{\ell(k - \ell)}{k - 1} \sum_{i \in \lambda_j} (\Omega_i^{(s)} - \overline \Omega_j^{(s)})^2 \\
& \le C^2\frac{\ell(k - \ell)}{k - 1} \frac{1}{n}\sum_{1 \le j \le n/k} \max_{i, i' \in \lambda_j} \|X_i - X_{i'}\|^2~,
\end{align*}
where $\overline \Omega_j^{(s)} = \frac{1}{k} \sum_{i \in \lambda_j} \Omega^{(s)}(X_i)$, and so the conditional variance converges in probability to zero under Assumption \ref{ass:pair}. It then follows from Markov's inequality and the fact that probabilities are bounded and hence uniformly integrable that
\[ \frac{1}{\sqrt{n}}\sum_{1 \le i \le n}(A_i - \eta)\Omega(X_i) = o_P(1)~. \]
Therefore,
\[ - M^{-1} \frac{1}{\sqrt n} \sum_{1 \leq i \leq n} m(X_i, A_i, R_i, \theta_0) = - M^{-1} \frac{1}{\sqrt n} \sum_{1 \leq i \leq n} m^\ast(X_i, A_i, R_i, \theta_0) + o_P(1)~, \]
which, together with \eqref{eq:vdv}, implies the desired result in \eqref{eq:normal}.
\qed

\subsection{Proof of Theorem \ref{thm:var-est}}
By assumption $\widehat M_n \xrightarrow{P} M$. Therefore, it suffices to show that
\begin{align} 
\label{eq:sigmahat1-consistent} \hat \Sigma_{1, n} & \xrightarrow{P} \Sigma_1 \\
\label{eq:sigmahat2-consistent} \hat \Sigma_{2, n} & \xrightarrow{P} \Sigma_2~,
\end{align}
where
\begin{align*}
\Sigma_1 & = \eta \var[m(X_i, 1, R_i(1), \theta_0)] + (1 - \eta) \var[m(X_i, 0, R_i(0), \theta_0)] \\
\Sigma_2 & = - \eta (1 - \eta) \var \big [ E[m(X_i, 1, R_i(1), \theta_0) | X_i] - E[m(X_i, 1, R_i(1), \theta_0)] \\
& \hspace{7em} - (E[m(X_i, 0, R_i(0), \theta_0) | X_i] - E[m(X_i, 0, R_i(0), \theta_0)]) \big ]~.
\end{align*}
In what follows, we will show \eqref{eq:sigmahat1-consistent}. The proof of \eqref{eq:sigmahat2-consistent} will follow from similar steps, but with the calculations below replaced by the ones in the proof of Lemmas C.2--C.3 in \cite{bai2024inference}, along with Assumptions \ref{ass:pairsofblocks}--\ref{ass:var-est}. To that end, we show for $1 \leq s \leq d_\theta$,
\begin{equation} \label{eq:firstmoment-consistent}
\hat \mu_{1, n}^{(s)} := \frac{1}{\eta n} \sum_{1 \leq i \leq n} m_s(X_i, 1, R_i, \hat \theta_n) \xrightarrow{P} E[m_s(X_i, 1, R_i(1), \theta_0)] =: \mu_1^{(s)}~,
\end{equation}
and similar arguments will establish the results for $a = 0$ as well as for the second moments and therefore \eqref{eq:sigmahat1-consistent}. For $\theta \in \Theta$, define
\begin{equation} \label{eq:sigmahattheta}
\hat \mu_{1, n}^{(s)}(\theta) = \frac{1}{\eta n} \sum_{1 \leq i \leq n} m_s(X_i, 1, R_i, \theta)
\end{equation}
and note $\hat \mu_{1, n}^{(s)} = \hat \mu_{1, n}^{(s)}(\hat \theta_n)$. Suppose \eqref{eq:firstmoment-consistent} doesn't hold. Then, there exists a subsequence $\{n_k\}_{k \geq 1}$ and $\epsilon_1, \epsilon_2 > 0$, such that
\begin{equation} \label{eq:contra}
\lim_{k \to \infty} P \{|\hat \mu_{1, n}^{(s)} - \mu_1^{(s)}| > \epsilon_1\} \to \epsilon_2~.    
\end{equation}
Because $\hat \theta_{n_k} \xrightarrow{P} \theta_0$ by Theorem \ref{thm:normal}, there exists a further subsequence, which we still denote by $\{n_k\}_{k \geq 1}$ by an abuse of notation, along which $\hat \theta_{n_k} \to \theta_0$ with probability one. Along that subsequence, Lemma \ref{eq:varest-fixedsequence} implies that $\hat \mu_{1, n}^{(s)} \xrightarrow{P} \mu_1^{(s)}$, in contradiction to \eqref{eq:contra}. Therefore, \eqref{eq:firstmoment-consistent} holds, and the theorem follows as discussed above.
\qed

\subsection{Proofs for Section \ref{sec:semi}}\label{sec:SPEB}
Recall that $P_n$ denotes the distribution of the observed data $(X^{(n)}, A^{(n)}, R^{(n)})$, and $Q$ denotes the marginal distribution of the vector $(R_i(1), R_i(0), X_i)$. Note that any treatment assignment mechanism $A^{(n)}$ satisfying Assumption \ref{ass:unconfounded} can be represented as a function of $X^{(n)}$ and some additional exogenous randomization device $U_n \in \mathbf R$. Let $p_n^{U_n}$ denote the density function for $U_n$ with respect to a dominating measure $\mu^U$. In what follows, we consider a family $\{Q_{t}: t \in \mathbf R^{d_{\theta}}\}$ of marginal distributions indexed by $t$, and let $q_t^X$ denote the density function for $X_i$ with respect to a dominating measure $\mu^X$, $q_t^{R(a) | X}(r | x)$ denote the conditional density of $R_i(a)$ given $X_i$ with respect to a dominating measure $\mu^R$. With some abuse of notation, continue letting $P_{t, n}$ denote the distribution of $(U_n, X^{(n)}, R^{(n)})$. We require that $Q_0 = Q$ and $P_{0, n} = P_n$ and define $q^X = q_0^X$ and $q^{R(a) | X} = q_0^{R(a) | X}$. As a consequence, the density function of $P_{t,n}$ is given by
\begin{equation} \label{eq:density}
\ell_n = p_n^U(U_n) \prod_{1 \leq i \leq n} q_{t}^X(X_i) \prod_{1 \leq i \leq n} \prod_{a \in \{0, 1\}} q_{t}^{R(a) | X}(R_i | X_i)^{I \{A_i = a\}}~.
\end{equation}
Because the density $p_n^{U_n}$ does not depend on $t$, and in general we will only concern ourselves with the ratio of likelihoods at different values of $t$ (so that $p_n^{U_n}$ in the ratio will cancel), in what follows we suppress the dependence on $n$ and simply denote the distribution $P_{t,n}$ by $P_t$.

We consider parametric submodels $\{P_t: t \in \mathbf R^{d_\theta}\}$, where $P_0 = P$, such that the following holds for some $g = (g^X, g^{R(1) | X}, g^{R(0) | X})$, each component of which is a $d_\theta$-dimensional function:
\begin{enumerate}[\rm (a)]
\item As $t \to 0$,
\begin{equation} \label{eq:qmd-1}
\int \frac{1}{\|t\|^2} \Big ( q_{t}^X(x)^{1/2} - q^X(x)^{1/2}  - \frac{1}{2} q^X(x)^{1/2} t' g^X(x) \Big )^2 d \mu^X(x) \to 0~.
\end{equation}
\item For $a \in \{0, 1\}$, as $t \to 0$,
\begin{multline} \label{eq:qmd-2}
\frac{1}{\|t\|^2} \int\!\!\!\int \Big ( q_{t}^{R(a) | X}(r | x)^{1/2} - q^{R(a) | X}(r | x)^{1/2} - \frac{1}{2} q^{R(a) | X}(r | x)^{1/2} t' g^{R(a) | X}(r | x) \Big )^2 \\
\times d \mu^R(r) q^X(x) d \mu^X(x) \to 0~.
\end{multline}
\end{enumerate}
In what follows, we will index a parametric submodel by its associated function $g$, denoted by $P_{t, g}$, to emphasize the role of $g$. Similarly we denote the density of $Q_{t,g}$ by $q_{t, g}$. When writing expectations and variances, we suppress the subscripts $P$ and $Q$ whenever doing so does not lead to confusion. For completeness, we document the following properties of score functions which satisfy \eqref{eq:qmd-1}--\eqref{eq:qmd-2}:

\begin{lemma} \label{lem:qmd}
For a parametric submodel $\{P_{t, g}: t \in \mathbf R^{d_\theta}\}$ with $P_{0, g} = P$ that satisfies \eqref{eq:qmd-1}--\eqref{eq:qmd-2},
\begin{enumerate}[\rm (a)]
\item $E[g^X(X) g^X(X)'] < \infty$.
\item $E[g^X(X)] = 0$.
\item $E[g^{R(a) | X}(R(a) | X) g^{R(a) | X}(R(a) | X)'] < \infty$ and hence $I^{R(a) | X}(X) < \infty$ with probability one under $Q$.
\item $E[g^{R(a) | X}(R(a) | X) | X] = 0$ with probability one under $Q$.
\end{enumerate}
\end{lemma}

\begin{proof}
(a) and (b) follow from Lemma 14.2.1 in \cite{lehmann2022testing}. (c) follows from the same lemma. In order to show (d), fix $t_n \to 0$. Note \eqref{eq:qmd-2} and Markov's inequality imply that along a subsequence $t_{n_k}$,
\begin{multline*}
\frac{1}{\|t_{n_k}\|^2} \int \Big ( q_{t_{n_k}}^{R(a) | X}(r | x)^{1/2} - q^{R(a) | X}(r | x)^{1/2} - \frac{1}{2} q^{R(a) | X}(r | x)^{1/2} t_{n_k}' g^{R(a) | X}(r | x) \Big )^2 \\
\times d \mu^R(r) \to 0
\end{multline*}
for $Q$-almost every $x$. Along that subsequence, another application of Lemma 14.2.1 in \cite{lehmann2022testing} implies (d).
\end{proof}
Define the information of $X$ as $I^X = E[g^X(X) g^X(X)']$. Define the conditional information of $R(a)$ given $X=x$ as
\[ I^{R(a) | X}(x) = E[g^{R(a) | X}(R(a) | X) g^{R(a) | X}(R(a) | X)' | X = x]~. \]
Further define $I = I^X + \eta E[I^{R(1) | X}(X)] + (1 - \eta) E[I^{R(0) | X}(X)]$. We restrict ourselves to parametric submodels that satisfy \eqref{eq:qmd-1}--\eqref{eq:qmd-2} for a $g$ that satisfies the following conditions. These submodels exist by Lemma \ref{lem:path} below.

\begin{condition} \label{cond:path}
The function $g$ satisfies that
\begin{enumerate}[\rm (a)]
\item $E[g^X(X)] = 0$ and $\var[g^X(X)] < \infty$.
\item For $a \in \{0, 1\}$, $E[g^{R(a) | X}(R(a) | X) | X] = 0$, and $\var[g^{R(a) | X}(R(a) | X)] < \infty$ with probability one.
\item $I$ is nonsingular.
\end{enumerate}
\end{condition}

\begin{lemma} \label{lem:path}
For any $g$ that satisfies Condition \ref{cond:path}, there exists a parametric submodel $\{P_{t, g}: t \in \mathbf R^{d_\theta}\}$ such that \eqref{eq:qmd-1}--\eqref{eq:qmd-2} hold.  
\end{lemma}

\begin{proof}
We use a vector version of the construction in in Example 25.16 in \cite{van_der_vaart1998asymptotic}. Let $k(x)$ be any strictly positive function that is bounded from above and away from zero with a bounded derivative such that $k(0) = k'(0) = 1$; for example, take $k(x) = 2 (1 + e^{-2x})^{-1}$. Define
\[ q_t^X(x) = C(t) q^X(x) k(t' g^X(x))~, \]
where $C(t) = \big ( \int q^X(x) k(t' g^X(x)) d \mu^X(x) \big )^{-1}$, so that $q_t^X(x)$ is a probability density function. Differentiating both sides of $C(t) \int q^X(x) k(t' g^X(x)) d \mu^X(x) = 1$ at $t = 0$, we get that $\frac{\partial}{\partial t} \big \vert_{t = 0} C(t) = 0$. It can then be verified through direct calculation that
\[ \frac{\partial}{\partial t} \bigg\vert_{t = 0} \log q_t^X(x) = g^X(x)~. \]
The quadratic mean differentiability requirement in \eqref{eq:qmd-1} follows from Lemma 7.6 in \cite{van_der_vaart1998asymptotic}. Next, for each $x \in \mathbf R^{d_x}$, we define
\[ q_t^{R(1) | X}(r | x) = C(t) q^{R(1) | X}(r | x) k(t' g^{R(1) | X}(r | x))~. \]
As above, it can be verified through direct calculation that
\[ \frac{\partial}{\partial t} \bigg\vert_{t = 0} \log q_t^{R(1) | X}(r | x) = g^{R(1) | X}(r | x)~. \]
To show \eqref{eq:qmd-2} for $a = 1$ (and symmetric arguments apply for $a = 0$), we modify the arguments in the proof of Lemma 7.6 in \cite{van_der_vaart1998asymptotic}. Define $s_t(r | x) = q_t^{R(1) | X}(r | x)^{1/2}$ and denote the $\ell$th component by $s_t^{(\ell)}(r | x)$ for $1 \leq \ell \leq d_\theta$. By the mean-value theorem, we have $s_t(r | x) - s(r | x) = \int_0^1 t' \dot s_{ut}(r | x) d u$, where $\dot s_t = \frac{\partial}{\partial t} s_t$, so that it follows from Jensen's inequality that
\begin{align}
\nonumber & \frac{1}{\|t\|^2} \int\!\!\!\int (s_t(r | x) - s_0(r | x) - t' \dot s_0(r | x))^2\, d \mu^R(r) q^X(x) d \mu^X(x) \\
\nonumber & \leq \frac{1}{\|t\|^2} \int\!\!\!\int\!\!\!\int_0^1 \big ( t' (\dot s_{ut}(r | x) - \dot s_0(r | x) ) \big )^2\,d u d \mu^R(r) q^X(x) d \mu^X(x)  \\
\label{eq:vitali} & \leq \int\!\!\!\int\!\!\!\int_0^1 \|\dot s_{ut}(r | x) - \dot s_0(r | x)\|^2 d u d \mu^R(r) q^X(x) d \mu^X(x)
\end{align}
where the first inequality follows from Jensen's inequality and the second by Cauchy-Schwarz. It then suffices to show \eqref{eq:vitali} goes to zero as $t \to 0$. Analyzing componentwise, it suffices to show that for $1 \leq \ell \leq d_\theta$, as $t \to 0$,
\begin{equation} \label{eq:vitali-l}
\int\!\!\!\int\!\!\!\int_0^1 (\dot s_{ut}^{(\ell)}(r | x) - \dot s_0^{(\ell)}(r | x))^2 d u d \mu^R(r) q^X(x) d \mu^X(x) \to 0~.
\end{equation}
The integrand in \eqref{eq:vitali-l} obviously converges to zero as $t \to 0$ by continuous differentiability of $q_t$. Recall that $\int \dot s_{ut}^{(\ell)}(r | x)^2 d \mu^R(r) = \frac{1}{4} [I_{ut}^{R(1) | X}(x)]_{(\ell, \ell)}$, where $I_{ut}^{R(1) | X}(x)$ is the conditional information for $a = 1$ given $x$ at $P_{ut, g}$. Therefore, it follows from Fubini's theorem that
\begin{align*}
\int\!\!\!\int\!\!\!\int_0^1 \dot s_0^{(\ell)}(r | x)^2 d u d \mu^R(r) q^X(x) d \mu^X(x) & = \frac{1}{4} \int [I_0^{R(1) | X}(x)]_{(\ell, \ell)} q^X(x) d \mu^X(x) \\
\int\!\!\!\int\!\!\!\int_0^1 \dot s_{ut}^{(\ell)}(r | x)^2 d u d \mu^R(r) q^X(x) d \mu^X(x) & = \frac{1}{4} \int\!\!\!\int_0^1 [I_{ut}^{R(1) | X}(x)]_{(\ell, \ell)} du q^X(x) d \mu^X(x)~.    
\end{align*}
To apply Vitali's theorem \citep[Proposition 2.29 in][]{van_der_vaart1998asymptotic}, it suffices to show that
\[ \int\!\!\!\int_0^1 [I_{ut}^{R(1) | X}(x)]_{(\ell, \ell)} du q^X(x) d \mu^X(x) \to \int [I_0^{R(1) | X}(x)]_{(\ell, \ell)} q^X(x) d \mu^X(x) \]
as $t \to 0$. To do so, we fix any arbitrarily small $\delta > 0$ and note that at least for $t$ small enough, $\|ut\| \leq \delta$ for $u \in [0, 1]$, so we can apply the dominated convergence theorem with
\[ [I_{ut}^{R(1) | X}(x)]_{(\ell, \ell)} \leq \sup_{\|h\| \leq \delta} [I_h^{R(1) | X}(x)]_{(\ell, \ell)}~, \]
as long as
\begin{equation} \label{eq:dom}
\begin{split}
& \int\!\!\!\int_0^1 \sup_{\|h\| \leq \delta} [I_h^{R(1) | X}(x)]_{(\ell, \ell)} du q^X(x) d \mu^X(x) \\
& = \int \sup_{\|h\| \leq \delta} [I_h^{R(1) | X}(x)]_{(\ell, \ell)} q^X(x) d \mu^X(x) < \infty~.      
\end{split}
\end{equation}
To show \eqref{eq:dom}, we calculate the conditional information as
\[ E \bigg [ \bigg ( \frac{\frac{\partial}{\partial h_\ell} C(h)}{C(h)} + \frac{k'(h' g^{R(1) | X}(R(1) | X)}{k(h' g^{R(1) | X}(R(1) | X))} g_\ell^{R(1) | X}(R(1) | X) \bigg )^2 \bigg \vert X = x \bigg ]~. \]
Note that $k'$ is bounded above and $k$ is bounded below, $C(h)$ is continuously differentiable with $C(0) = 1$, and so is bounded for $\|h\| \leq \delta$. Therefore, an application of the Cauchy-Schwarz inequality implies the previous expectation is bounded by a constant plus a constant multiple of $[I_0^{R(1) | X}(x)]_{(\ell, \ell)}$. The desired conclusion in \eqref{eq:dom} then follows because $E[I_0^{R(1) | X}(X)] < \infty$, and the proof is complete.
\end{proof}

For $t \in \mathbf R^{d_\theta}$, the log-likelihood ratio between $P_{t/\sqrt n, g}$ and $P_0 = P$ is
\[ L_{t, n}(g) = \frac{1}{n} \sum_{1 \leq i \leq n} \log \frac{q_{t / \sqrt n, g}^X(X_i)}{q^X(X_i)} + \frac{1}{n} \sum_{1 \leq i \leq n} \sum_{a \in \{0, 1\}} I \{A_i = a\} \log \frac{q_{t / \sqrt n, g}^{R(a) | X}(R_i | X_i)}{q^{R(a) | X}(R_i | X_i)}~. \]
The following lemma establishes an expansion of the log-likelihood ratio and local asymptotic normality of $\{P_{t/\sqrt n, g}\}$.

\begin{lemma} \label{lem:lan}
Suppose the treatment assignment mechanism satisfies Assumption \ref{ass:unconfounded} and the path satisfies \eqref{eq:qmd-1}--\eqref{eq:qmd-2} for $g$ satisfying Condition \ref{cond:path}. Then,
\begin{align*}
L_{t, n}(g) & = \frac{1}{\sqrt n} \sum_{1 \leq i \leq n} t' s_g(X_i, A_i, R_i) - \frac{1}{2} t' I^X t \\
& \hspace{3em} - \frac{1}{2n} \sum_{1 \leq i \leq n} \sum_{a \in \{0, 1\}} I \{A_i = a\} t' I^{R(a) | X}(X_i) t + o_P(1)~,
\end{align*}
where
\begin{equation} \label{eq:score}
s_g(x, a, r) = g^X(x) + I \{a = 1\} g^{R(1) | X} (r | x) + I \{a = 0\} g^{R(0) | X} (r | x)~
\end{equation}
and $I = I^X + \eta E_Q[I^{R(1)|X}(X_i)] + (1 - \eta)E_Q[I^{R(0)|X}(X_i)]$.
If in addition the assignment mechanism satisfies Assumption \ref{ass:LLN}, then, under $P_0$,
\[ L_{t, n}(g) \stackrel{d}{\to} N \Big ( - \frac{1}{2} t' I t, t' I t \Big )~, \]
\end{lemma}

\begin{proof}
The first result follows from Theorem 3.1 of \cite{armstrong2022asymptotic}. The second result follows from Lemma \ref{lem:clt} given Assumption \ref{ass:LLN} and the assumption that each component of $I^{R(a)|X}(x)$ is integrable, noting that $E[s_g(X, 1, R(1)) - s_g(X, 0, R(0)) | X] = 0$.
\end{proof}

We emphasize that Lemma \ref{lem:marginal} implies
\[ \sum_{1 \leq i \leq n} s_g(X_i, A_i, R_i)\]
is the sum of $n$ identically distributed, although possibly dependent, random variables. Therefore, in what follows, quantities like $E_P[s_g]$ are well defined.

Let $\theta(P) \in \mathbf R^{d_\theta}$ be a parameter of interest. Further suppose that there exists a $d_\theta \times 1$ vector of functions $\psi^\ast \in L^2(P)$ such that for each $g$ satisfying Condition \ref{cond:path}, for all $t \in \mathbf R^{d_\theta}$, as $n \to \infty$,
\begin{equation} \label{eq:pathwise}
\sqrt n(\theta(P_{t / \sqrt n, g}) - \theta(P)) \to E_P[\psi^\ast s_g' t]~.
\end{equation}
In Lemma \ref{lem:differentiability} below, we provide explicit conditions which guarantee this is possible when $\theta(P)$ is defined by \eqref{eq:moments}.

We call an estimator $\tilde \theta_n$ for $\theta(P)$ regular if for all $g$ satisfying Condition \ref{cond:path} and $t \in \mathbf R^{d_\theta}$,
\begin{equation} \label{eq:regular}
\sqrt n(\tilde \theta_n - \theta(P_{t/\sqrt n, g})) \xrightarrow{P_{t/\sqrt n, g}} L 
\end{equation}
for a fixed probability measure $L$.

The following lemma establishes a convolution theorem for regular estimators:

\begin{lemma} \label{lem:convolution}
Suppose $\theta$ satisfies \eqref{eq:pathwise}. Let $\tilde \theta_n$ be a regular estimator for $\theta$. Further suppose that $\psi^\ast = s_g$ for some function $g$ satisfying Condition \ref{cond:path}. Then,
\[ L = N(0, E_P[\psi^\ast \psi^{\ast\prime}]) \ast B~, \]
where $B$ is a fixed probability measure.
\end{lemma}

\begin{proof}
In what follows, for each $g$ satisfying Condition \ref{cond:path}, we consider the linear subspace given by 
\[ \mathcal{M}_g = \{t' s_g: t \in \mathbf R^{d_\theta}\}~. \]
Note that $t' s_g$ appears in the expansion of the log-likelihood ratio between $P_{t/\sqrt n, g}$ and $P$. To align our setting with Theorem 3.11.2 in \cite{van_der_vaart1996weak}, we first characterize the adjoint map (viewed as a mapping into $\mathcal{M}_g$) of the function $v \mapsto E[\psi^* v] \in \mathbf R^{d_\theta}$, where $v \in \mathcal{M}_g$. To that end, implicitly identifying each $b \in \mathbf R^{d_\theta}$ with the functional $b^*:\mathbf R^{d_\theta} \rightarrow \mathbf R$ given by $x \mapsto b'x$, we construct a $w(b) \in \mathbf R^{d_\theta}$ such that
\[ b' E_P[\psi^\ast s_g' t] = E_P[w(b)' s_g s_g' t] \]
for all $t\in \mathbf R^{d_\theta}$, where we note $w(b)' s_g \in \mathcal{M}_g$, and is thus the output of the adjoint map when applied to the functional $b^*$. Because $E[s_g s_g']$ is invertible, we immediately obtain
\[ w(b) = E[s_g s_g']^{-1} E[s_g \psi^{\ast\prime}] b~. \]
Here we use the assumption that $I$ is nonsingular. It then follows from the local asymptotic normality established in Lemma \ref{lem:lan} and Theorem 3.11.2 in \cite{van_der_vaart1996weak} that
\[ L = N(0, V_g) \ast B_g~, \]
where $B_g$ is a fixed probability measure and $b' V_g b = E[(w(b)' s_g)^2]$, so that we have
\[ V_g = E_P[\psi^\ast s_g'] E_P[s_g s_g']^{-1} E[s_g \psi^{\ast\prime}]~. \]
Furthermore, by a standard projection argument, in particular the fact that the second moment of $\psi^\ast - E_P[\psi^\ast s_g'] E_P[s_g s_g']^{-1} s_g$ is positive semi-definite, it can be shown that $V_g$ is maximized in the matrix sense when $s_g = \psi^\ast$. Note this maximum is attained by our assumption that $\psi^\ast = s_g$ for some $g$ satisfying Condition \ref{cond:path}. The conclusion then follows.
\end{proof}

To apply Lemma \ref{lem:convolution} to the setting in Section \ref{sec:semi}, we we establish the form of $\psi^\ast$ in \eqref{eq:pathwise} for the parameter $\theta_0 = \theta(P)$ defined by \eqref{eq:moments}. Define $\eta(X_i) = P \{A_i = 1 | X_i\}$. Note that
\begin{equation} \label{eq:moments-P}
\begin{split}
0 & = E_P[m(X_i, A_i, R_i, \theta(P))] \\
& = E_Q[m(X, 1, R(1), \theta(P)) \eta(X)] + E_Q[m(X, 0, R(0), \theta(P)) (1 - \eta(X))]~.
\end{split}
\end{equation}

\begin{lemma} \label{lem:differentiability}
Suppose the treatment assignment mechanism satisfies Assumptions \ref{ass:unconfounded} and \ref{ass:LLN}. Fix a function $g$ that satisfies Condition \ref{cond:path}. Suppose \eqref{eq:qmd-1}--\eqref{eq:qmd-2} holds. Fix $t \in \mathbf R^{d_\theta}$ and consider a one-dimensional submodel $\{P_{t / \sqrt n, g}\}$ such that
\begin{align}
\label{eq:bounded_path1} E_{Q_{t / \sqrt n}}[m(X, a, R(a), \theta(P))^2] & = O(1) \\
\label{eq:bounded_path2} E_{Q^X}[E_{Q_{t / \sqrt n}^{R(a) | X}}[m(X, a, R(a), \theta(P))^2 | X]] & = O(1) \\
\label{eq:bounded_path3} E_{Q_{t / \sqrt n}^X}[E_{Q^{R(a) | X}}[m(X, a, R(a), \theta(P))^2 | X]] & = O(1)  
\end{align}
as $n \to \infty$ and $\theta(P_{t / \sqrt n, g})$ is uniquely determined by \eqref{eq:moments-P}. Then, $\theta(P_{t / \sqrt n, g})$ defined by \eqref{eq:moments-P} satisfies
\begin{align*}
& \sqrt n ( \theta(P_{t / \sqrt n, g}) - \theta(P)) \\
& \to M^{-1} E_P[m(X_i, A_i, R_i, \theta(P))(g^X(X_i) \\
& \hspace{3em} + I \{A_i = 1\}g^{R(1) | X}(R_i | X_i) + I \{A_i = 0\}g^{R(0) | X}(R_i | X_i))'] t \\
& = E_P[\psi^\ast(X_i, A_i, R_i, \theta(P)) (g^X(X_i) \\
& \hspace{3em} + I \{A_i = 1\}g^{R(1) | X}(R_i | X_i) + I \{A_i = 0\} g^{R(0) | X}(R_i | X_i))'] t~,
\end{align*}
where
\begin{align*}
& \psi^\ast(X_i, A_i, R_i, \theta(P)) \\
& = - M^{-1} \Big ( \eta(X_i) E_Q[m(X_i, 1, R_i(1), \theta(P)) | X_i] \\
& \hspace{5em} + (1 - \eta(X_i)) E_Q[m(X_i, 0, R_i(0), \theta(P)) | X_i] \\
& \hspace{5em} + I\{A_i = 1\} (m(X_i, 1, R_i, \theta(P)) - E_Q[m(X_i, 1, R_i(1), \theta(P)) | X_i]) \\
& \hspace{5em} + I\{A_i = 0\} (m(X_i, 0, R_i, \theta(P)) - E_Q[m(X_i, 0, R_i(0), \theta(P)) | X_i]) \Big )~.
\end{align*}
\end{lemma}

\begin{proof}
In what follows, we only use the property that the quadratic mean derivative of $P_{t / \sqrt n, g}$ is given by $s_g' t$. Therefore, for ease of notation we consider a generic one-dimensional submodel $\{P_\nu: \nu \in [-\epsilon, \epsilon]\}$ that satisfies \eqref{eq:qmd-1}--\eqref{eq:qmd-2} for some $g = (g^X, g^{R(1) | X}, g^{R(0) | X})$, each component of which is a one-dimensional function. \eqref{eq:moments-P} implies
\begin{multline*}
0 = \int m(x, 1, r, \theta(P_\nu)) q_\nu^{R(1) | X}(r | x) d\mu^R(r) \eta(x) q_\nu^X(x) d\mu^X(x) \\
+ \int m(x, 0, r, \theta(P_\nu)) q_\nu^{R(0) | X}(r | x) d\mu^R(r) (1 - \eta(x)) q_\nu^X(x) d\mu^X(x)
\end{multline*}
Note that
\begin{align*}
& \int m(x, 1, r, \theta(P)) q_\nu^{R(1) | X}(r | x) d\mu^R(r) \eta(x) q_\nu^X(x) d\mu^X(x) \\
& \hspace{2em} - \int m(x, 1, r, \theta(P)) q^{R(1) | X}(r | x) d\mu^R(r) \eta(x) q^X(x) d\mu^X(x) \\
& = \gamma_1(\nu) + \gamma_2(\nu) + \gamma_3(\nu) + \gamma_4(\nu)~,
\end{align*}
where
\begin{align*}
\gamma_1(\nu) & = \int m(x, 1, r, \theta(P)) \big ( q_\nu^{R(1) | X}(r | x)^{1/2} - q^{R(1) | X}(r | x)^{1/2} \big ) q_\nu^{R(1) | X}(r | x)^{1/2} d\mu^R(r) \\
& \hspace{5em} \times \eta(x) q_\nu^X(x) d\mu^X(x) \\
\gamma_2(\nu) & = \int m(x, 1, r, \theta(P)) \big ( q_\nu^{R(1) | X}(r | x)^{1/2} - q^{R(1) | X}(r | x)^{1/2} \big ) q^{R(1) | X}(r | x)^{1/2} d\mu^R(r)  \\
& \hspace{5em} \times \eta(x) q_\nu^X(x) d\mu^X(x) \\
\gamma_3(\nu) & = \int m(x, 1, r, \theta(P)) q^{R(1) | X}(r | x) d\mu^R(r) \\
& \hspace{5em} \times \eta(x) \big ( q_\nu^X(x)^{1/2} - q^X(x)^{1/2} \big ) q_\nu^X(x)^{1/2} d\mu^X(x) \\
\gamma_4(\nu) & = \int m(x, 1, r, \theta(P)) q^{R(1) | X}(r | x) d\mu^R(r) \\
& \hspace{5em} \times \eta(x) \big ( q_\nu^X(x)^{1/2} - q^X(x)^{1/2} \big ) q^X(x)^{1/2} d\mu^X(x)~.
\end{align*}
It follows from the Cauchy-Schwarz inequality that
\begin{align*}
& \frac{1}{\nu} \gamma_4(\nu) - \int m(x, 1, r, \theta(P)) q^{R(1) | X}(r | x) d\mu^R(r) \\
& \hspace{5em} \times \eta(x) \frac{1}{2} g^X(x) q^X(x)^{1/2} \times q^X(x)^{1/2}  d\mu^X(x) \\
& \leq \int \bigg ( m(x, 1, r, \theta(P))^2 q^{R(1) | X}(r | x) d\mu^R(r) \eta(x)^2 q^X(x) d \mu^X(x) \bigg )^{1/2} \\
& \times \bigg ( \int q^{R(1) | X}(r | x) d\mu^R(r) \\
& \hspace{5em} \times \Big ( \frac{1}{\nu}(q_\nu^X(x)^{1/2} - q^X(x)^{1/2}) - \frac{1}{2} g^X(x) q^X(x)^{1/2} \Big )^2 d \mu^X(x) \bigg )^{1/2} \\
& \to 0
\end{align*}
by the assumption that $E_P[m(X, a, R(a), \theta(P))^2] < \infty$, the facts that $0 \leq \eta(x) \leq 1$, $\int q^{R(1) | X}(r | x) d\mu^R(r) = 1$, and \eqref{eq:qmd-1}. Similar arguments implies as $\nu \to 0$,
\[ \frac{1}{\nu} \gamma_1(\nu) - \int m(x, 1, r, \theta(P)) \frac{1}{2} g^{R(1) | X}(r | x) q^{R(1) | X}(r | x) d\mu^R(r) \eta(x) q^X(x) d\mu^X(x) \to 0 \]
because $E_{P_\nu}[m(X, a, R(a), \theta(P))^2] = O(1)$ as $\nu \to 0$. The limits of $\gamma_2(\nu)$ and $\gamma_3(\nu)$ can be derived following similar arguments using \eqref{eq:bounded_path2}--\eqref{eq:bounded_path3}. Combining all previous results yields
\begin{align*}
& \frac{\partial}{\partial \nu} E_{P_\nu}[m(X, A, R, \theta(P))] \Big |_{\nu = 0} \\
& = E_Q[m(X, 1, R(1), \theta(P)) (g^X(X) + g^{R(1) | X}(R | X)) \eta(X)] \\
& \hspace{3em} + E_Q[m(X, 0, R(0), \theta(P)) (g^X(X) + g^{R(0) | X}(R | X)) (1 - \eta(X))] \\
& = E_P[m(X, A, R, \theta(P))(g^X(X) + I \{A = 1\}g^{R(1) | X}(R) + I \{A = 0\}g^{R(0) | X}(R))]~.
\end{align*}
On the other hand, by definition
\[ M = \frac{\partial}{\partial \theta'} E_P[m(X, A, R, \theta)] \Big |_{\theta = \theta(P)}~. \]
The formula for the derivative therefore follows from the implicit function theorem (in particular, because we have assumed the existence of $\theta(P_\nu)$ along the path, it follows from the last part of the proof of Theorem 3.2.1 in \cite{krantz2013implicit}). The second equality follows from Lemma \ref{lem:marginal} together with Condition \ref{cond:path}.
\end{proof}

Finally, to prove Theorem \ref{thm:efficiencybound} we require the following additional regularity condition:

\begin{assumption}\label{ass:bounded_path}
For every function $g$ satisfying Condition \ref{cond:path} and every $t \in \mathbf R^{d_\theta}$ there exists a submodel $P_{t / \sqrt n,g}$ for which \eqref{eq:bounded_path1}--\eqref{eq:bounded_path3} hold as $n \to \infty$, and $\theta(P_{t / \sqrt n, g})$ is uniquely determined by \eqref{eq:moments-P}.
\end{assumption}

This assumption guarantees that every element satisfying Condition \ref{cond:path} has a corresponding path for which we can apply Lemma \ref{lem:differentiability}. A similar assumption appears in \cite{chen2018overidentification} (see their Assumption 4.1(iv)). Note that a simple sufficient condition for the first part of Assumption \ref{ass:bounded_path} is that $m(x, a, r, \theta(P))$ is a bounded function in $(x, r)$ on the support of $(X, R(a))$. The second part of Assumption \ref{ass:bounded_path} can be verified easily in specific examples (see, for instance, Examples \ref{ex:ATE}--\ref{ex:log_odds} in the main text). Alternatively, Assumption \ref{ass:bounded_path} could be avoided by assuming that we can differentiate under the integral in the final step of the proof of Lemma \ref{lem:differentiability}, from which we would immediately obtain the expression for the pathwise derivative. See, for instance, \cite{newey1994asymptotic} and \cite{chen2008semiparametric}.

\begin{proof}[\sc Proof of Theorem \ref{thm:efficiencybound}]
First note $\theta$ satisfies \eqref{eq:pathwise} because of Lemma \ref{lem:differentiability} and Assumption \ref{ass:bounded_path}. The result then follows from Lemma \ref{lem:convolution} upon noting that $\psi^\ast = s_g$ for some $g$ that satisfies Condition \ref{cond:path}.
\end{proof}

To study regular estimators, we need the following lemma.

\begin{lemma} \label{lem:clt}
Suppose the treatment assignment mechanism satisfies Assumptions \ref{ass:unconfounded}, \ref{ass:LLN}, and \ref{ass:imbalance}. Let $f(x, a, r)$ be a vector-valued function such that $E[f(X, A, R)] = 0$ and $E[f^2(X, a, R(a))] < \infty$ for $a \in \{0, 1\}$. Then,
\[ \frac{1}{\sqrt n} \sum_{1 \leq i \leq n} f(X_i, A_i, R_i) \xrightarrow{d} N(0, V_f)~, \]
where $V_f = V_{1, f} + V_{2, f} + V_{3, f}$ for
\begin{align*}
V_{1, f} & = \var[E[f(X, A, R) | X]] = \var[\eta E[f(X, 1, R(1)) | X] + (1 - \eta) E[f(X, 0, R(0)) | X]] \\
V_{2, f} & = V_{E[f(X, 1, R(1)) - f(X, 0, R(0)) | X]}^{\rm imb} \\
V_{3, f} & = E[\eta \var[f(X, 1, R(1)) | X] + (1 - \eta) \var[f(X, 0, R(0)) | X]]~.
\end{align*}
\end{lemma}

\begin{proof}
Note that
\[ \mathbb C_n := \frac{1}{\sqrt n} \sum_{1 \leq i \leq n} f(X_i, A_i, R_i) = \mathbb C_{1, n} + \mathbb C_{2, n} + \mathbb C_{3, n}~, \]
where
\begin{align*}
\mathbb C_{1, n} & = \frac{1}{\sqrt n} \sum_{1 \leq i \leq n} E[f(X_i, A_i, R_i) | X_i] \\
& = \frac{1}{\sqrt n} \sum_{1 \leq i \leq n} (\eta E[f(X_i, 1, R_i(1)) | X_i] + (1 - \eta) E[f(X_i, 0, R_i(0)) | X_i]) \\
\mathbb C_{2, n} & = \frac{1}{\sqrt n} \sum_{1 \leq i \leq n} (A_i - \eta) E[f(X_i, 1, R_i(1)) - f(X_i, 0, R_i(0)) | X_i] \\
\mathbb C_{3, n} & = \frac{1}{\sqrt n} \sum_{1 \leq i \leq n} \big ( A_i (f(X_i, 1, R_i(1)) - E[f(X_i, 1, R_i(1)) | X_i]) \\
& \hspace{5em} + (1 - A_i) (f(X_i, 0, R_i(0)) - E[f(X_i, 0, R_i(0)) | X_i]) \big )~.
\end{align*}
Note that $\mathbb C_{1, n}$ has zero mean because $E[E[f(X_i, A_i, R_i) | X_i]] = E[f(X_i, A_i, R_i)] = 0$. Further note that $\mathbb C_{1, n} = E[\mathbb C_n | X^{(n)}]$, $\mathbb C_{2, n} = E[\mathbb C_n | X^{(n)}, A^{(n)}] - E[\mathbb C_n | X^{(n)}]$, and $\mathbb C_{3, n} = \mathbb C_n - E[\mathbb C_n | X^{(n)}, A^{(n)}]$. It follows from the central limit theorem and $E[f^2(X, a, R(a))] < \infty$ that
\[ \mathbb C_{1, n} \xrightarrow{d} N(0, V_{1, f})~. \]
Next, it follows from Assumption \ref{ass:imbalance} that
\[ \rho(\mathcal L(\mathbb C_{2, n} | X^{(n)}), N(0, V_{2, f})) \xrightarrow{P} 0~. \]
For $\mathbb C_{3, n}$, it follows from Assumption \ref{ass:LLN} that
\[ \var[\mathbb C_{3, n} | X^{(n)}, A^{(n)}] \xrightarrow{P} V_{3, f}~. \]
As a result, one can verify the Lindeberg condition conditional on $X^{(n)}$ and $A^{(n)}$ as in the proof of Lemma S.1.4 of \cite{bai2022inference}, and obtain
\[ \rho(\mathcal L(\mathbb C_{3, n} | X^{(n)}, A^{(n)}), N(0, V_{3, f})) \xrightarrow{P} 0~. \]
The proof can then be completed by the subsequencing argument at the end of the proof of Lemma S.1.4 of \cite{bai2022inference}.
\end{proof}

\begin{proof}[\sc Proof of Theorem \ref{thm:regular}]
Recall from Lemma \ref{lem:lan} that $L_{t, n}(g)$ is also asymptotically linear. Because the treatment assignment mechanism satisfies Assumptions \ref{ass:unconfounded}, \ref{ass:LLN}, and \ref{ass:imbalance}, the path satisfies Condition \ref{cond:path}, and hence $\psi$ and $s_g$ jointly satisfy the conditions in Lemma \ref{lem:clt} (in particular, note that $E[s_g(X, 1, R(1)) - s_g(X, 0, R(0)) | X] = 0$),
\[ \begin{pmatrix}
\sqrt n (\tilde \theta_n - \theta(P)) \\
L_{t, n}(g)
\end{pmatrix} = \begin{pmatrix}
\frac{1}{\sqrt n} \sum_{1 \leq i \leq n} \psi(X_i, A_i, R_i, \theta(P)) \\
\frac{1}{\sqrt n} \sum_{1 \leq i \leq n} t' s_g(X_i, A_i, R_i)
\end{pmatrix} - \begin{pmatrix}
0 \\
- \frac{1}{2} t' I t
\end{pmatrix} + o_P(1) \]
are jointly asymptotically normal. Because $E[s_g(X, 1, R(1)) - s_g(X, 0, R(0)) | X] = 0$, the covariance in the second term in the joint variance in Lemma \ref{lem:clt} vanishes, and thus the overall covariance is given by
\begin{align*}
& E[E[\psi | X] g^X(X)'] t \\
& \hspace{3em} + \eta E[(\psi(X, 1, R(1), \theta(P)) - E[\psi(X, 1, R(1), \theta(P)) | X]) g^{R(1) | X}(R(1) | X)'] t \\
& \hspace{3em} + (1 - \eta) E[(\psi(X, 0, R(0), \theta(P)) - E[\psi(X, 0, R(0), \theta(P)) | X]) g^{R(0) | X}(R(0) | X)'] t \\
& = E[\psi s_g'] t~.
\end{align*}
It then follows from Le Cam's third lemma that under $P_{t / \sqrt n, g}$, $\sqrt n (\tilde \theta_n - \theta_0)$ converges in distribution to a normal distribution with mean $E[\psi s_g'] t$ and the same variance as in the limit under $P$. At the same time,
\[ \sqrt n(\tilde \theta_n - \theta(P_{t / \sqrt n, g})) = \sqrt n(\tilde \theta_n - \theta(P)) - \sqrt n(\theta(P_{t / \sqrt n, g} - \theta(P))~. \]
Therefore, \eqref{eq:regular} holds if and only if
\[ E[\psi s_g'] t = E[\psi^\ast s_g'] t~. \]
Furthermore, $\tilde \theta_n$ is regular if the equality holds for all $t$ and all $g$ satisfying Condition \ref{cond:path}, if and only if
\[ E[(\psi - \psi^\ast) s_g'] = 0 \]
for all $g$ satisfying Condition \ref{cond:path}. Note that
\begin{align*}
& E[(\psi - \psi^\ast) s_g'] \\
& = E[E[\psi(X, A, R, \theta(P)) - \psi^\ast(X, A, R, \theta(P)) | X] g^X(X)'] \\
& \hspace{3em} + E[A_i (\psi(X, 1, R(1), \theta(P)) - \psi^\ast(X, 1, R(1), \theta(P)) \\
& \hspace{6em} - E[\psi(X, 1, R(1), \theta(P)) - \psi^\ast(X, 1, R(1), \theta(P)) | X]) g^{R(1) | X}(R(1) | X)'] \\
& \hspace{3em} + E[(1 - A_i) (\psi(X, 0, R(0), \theta(P)) - \psi^\ast(X, 0, R(0), \theta(P)) \\
& \hspace{6em} - E[\psi(X, 0, R(0), \theta(P)) - \psi^\ast(X, 0, R(0), \theta(P)) | X]) g^{R(0) | X}(R(0) | X)']~.
\end{align*}
By setting $g^{R(1) | X} = 0$, $g^{R(0) | X} = 0$, and $g^X(X) = E[\psi(X, A, R, \theta(P)) - \psi^\ast(X, A, R, \allowbreak\theta(P)) | X]$, we get that
\begin{equation} \label{eq:psi-psi*}
E[\psi(X, A, R, \theta(P)) | X] = E[\psi^\ast(X, A, R, \theta(P)) | X]~.
\end{equation}
Setting $g^X(X) = 0 = g^{R(0) | X}$ and $g^{R(1) | X} = \psi(X, 1, R(1), \theta(P)) - \psi^\ast(X, 1, R(1), \theta(P)) - E[\psi(X, 1, R(1), \theta(P)) - \psi^\ast(X, 1, R(1), \theta(P)) | X]$, we get
\begin{align*}
& \psi(X, 1, R(1), \theta(P)) - \psi^\ast(X, 1, R(1), \theta(P)) \\
& \hspace{3em} - E[\psi(X, 1, R(1), \theta(P)) - \psi^\ast(X, 1, R(1), \theta(P)) | X] = 0~,	
\end{align*}
which implies that $\psi(X, 1, R(1), \theta(P)) - \psi^\ast(X, 1, R(1), \theta(P))$ can only be a function of $X$. Denote it by $\psi^\perp(X, 1)$. Similarly, $\psi(X, 0, R(0), \theta(P)) - \psi^\ast(X, 0, R(0), \theta(P))$ can only be a function of $X$. Denote it by $\psi^\perp(X, 0)$. We have
\begin{align*}
\psi(X, A, R, \theta(P)) & = A \psi(X, 1, R(1), \theta(P)) + (1 - A) \psi(X, 0, R(0), \theta(P)) \\
& = A (\psi^\ast(X, 1, R(1), \theta(P)) + \psi^\perp(X, 1, \theta(P))) \\
& \hspace{3em} + (1 - A) (\psi^\ast(X, 0, R(0), \theta(P)) + \psi^\perp(X, 0, \theta(P))) \\
& = \psi^\ast(X, A, R, \theta(P)) + \psi^\perp(X, A, \theta(P))~,
\end{align*}
and it follows from \eqref{eq:psi-psi*} that $E[\psi^\perp(X, A, \theta(P)) | X] = 0$. Finally, to show that $\tilde\theta_n$ is efficient if and only if \eqref{eq:fast} holds, we apply Lemma \ref{lem:clt} with $f = \psi$. If \eqref{eq:fast} holds, then $V_\psi = V_{1, \psi} + V_{3, \psi} = V_{\psi^\ast}$ because $V_{1, \psi} = V_{1, \psi^\ast}$, $V_{3, \psi} = V_{3, \psi^\ast}$ and $V_{2, \psi^\ast} = 0$. On the other hand, if $\tilde \theta_n$ is efficient, then $V_{2, \psi} = 0$, so \eqref{eq:fast} holds by Markov's inequality combined with the fact that probabilities are bounded and thus uniformly integrable.
\end{proof}

\subsection{Auxiliary Lemmas}
\begin{lemma} \label{lem:marginal}
Suppose \eqref{eq:unconfounded} holds and $\mathrm{Pr} \{A_i = 1 | X_i = x\}$ as a function is identical across $1 \leq i \leq n$. Then,
\begin{equation} \label{eq:marginal-indep}
(R_i(1), R_i(0)) \indep A_i | X_i~.
\end{equation}
Moreover, $(X_i, A_i, R_i)$ is identically distributed across $1 \leq i \leq n$.
\end{lemma}

\begin{proof}
Fix $a \in \{0, 1\}$ and any Borel sets $B \in \mathbf R^{d_r} \times \mathbf R^{d_r}$ and $C \in \mathbf R^{d_x}$.
\begin{align*}
& E[\mathrm{Pr} \{(R_i(1), R_i(0)) \in B, A_i = a | X_i\} I \{X_i \in C\}] \\
& = E[E[\mathrm{Pr} \{(R_i(1), R_i(0)) \in B, A_i = a | X^{(n)}\} | X_i] I \{X_i \in C\}] \\
& = E[E[\mathrm{Pr} \{(R_i(1), R_i(0)) \in B | X^{(n)}\} \mathrm{Pr} \{A_i = a | X^{(n)}\} | X_i] I \{X_i \in C\}] \\
& = E[\mathrm{Pr} \{(R_i(1), R_i(0)) \in B | X_i\} \mathrm{Pr} \{A_i = a | X_i\} I \{X_i \in C\}]~,
\end{align*}
where the first equality follows from the law of iterated expectations, the second equality follows from \eqref{eq:unconfounded}, the third equality follows from the law of iterated expectations as well as the fact that $Q_n = Q^n$. The first statement of the lemma then follows from the definition of a conditional expectation.

To prove the second statement, fix units $i$ and $i'$. Clearly $X_i$ and $X_{i'}$ are identically distributed. Conditional on $X_i$, for any Borel set $C \in \mathbf R^{d_r}$ and $a \in \{0, 1\}$, it follows (a) that
\[ \mathrm{Pr} \{R_i \in C, A_i = a | X_i\} = \mathrm{Pr} \{A_i = a | X_i\} \mathrm{Pr} \{R_i(a) \in C | X_i\}~. \]
The conclusion then follows because $\mathrm{Pr} \{A_i = 1 | X_i = x\}$ is identical across $1 \leq i \leq n$ and $Q_n = Q^n$.
\end{proof}

\begin{lemma} \label{lem:consistency}
Suppose the treatment assignment mechanism satisfies Assumptions \ref{ass:a}--\ref{ass:pair} and the moment functions satisfy Assumption \ref{ass:normal}. Then, $\hat \theta_n \stackrel{P}{\to} \theta_0$.
\end{lemma}
\begin{proof}
It follows from Assumption \ref{ass:normal}(a) and Theorem 5.9 in \cite{van_der_vaart1998asymptotic} that we only need to establish for each $1 \leq s \leq d_\theta$,
\begin{equation} \label{eq:uc}
\sup_{\theta \in \Theta}~ \bigg | \frac{1}{n} \sum_{1 \leq i \leq n} (m_s(X_i, A_i, R_i, \theta) - E_P[m_s(X_i, A_i, R_i, \theta)]) \bigg | \stackrel{P}{\to} 0~.
\end{equation}
To begin, note it follows from Assumption \ref{ass:normal}(d) and the dominated convergence theorem that if $m_s(x, a, r, \theta_m) \allowbreak \to m_s(x, a, r, \theta)$ as $m \to \infty$ for $\{\theta_m\} \subset \Theta^\ast$, then $E_P[m_s(X_i, A_i, R_i, \allowbreak\theta_m)] \to E_P[m_s(X_i, A_i, R_i, \theta)]$. Here, the dominating function exists by Problem 2.4.1 in \cite{van_der_vaart1996weak}. Assumption \ref{ass:normal}(c) then implies
\begin{multline} \label{eq:countable}
\sup_{\theta \in \Theta}~ \bigg | \frac{1}{n} \sum_{1 \leq i \leq n} (m_s(X_i, A_i, R_i, \theta) - E_P[m_s(X_i, A_i, R_i, \theta)]) \bigg | \\
= \sup_{\theta \in \Theta^\ast}~ \bigg | \frac{1}{n} \sum_{1 \leq i \leq n} (m_s(X_i, A_i, R_i, \theta) - E_P[m_s(X_i, A_i, R_i, \theta)]) \bigg |~,
\end{multline}
which is measurable. Next, note that
\begin{equation} \label{eq:half-m}
m(X_i, A_i, R_i, \theta) = A_i m(X_i, 1, R_i(1), \theta) + (1 - A_i) m(X_i, 0, R_i(0), \theta)~.
\end{equation}
and it follows from Lemma \ref{lem:marginal} that
\begin{equation} \label{eq:half-E}
E_P[m(X_i, A_i, R_i, \theta)] = \eta E_Q[m(X_i, 1, R_i(1), \theta)] + (1 - \eta) E_Q[m(X_i, 0, R_i(0), \theta)]~,
\end{equation}
from which we obtain that
\begin{align*}
& E \bigg [ \sup_{\theta \in \Theta}~ \bigg | \frac{1}{n} \sum_{1 \leq i \leq n} (m_s(X_i, A_i, R_i, \theta) - E_P[m_s(X_i, A_i, R_i, \theta)]) \bigg | \bigg ] \\
& = E \bigg [ \sup_{\theta \in \Theta^\ast}~ \bigg | \frac{1}{n} \sum_{1 \leq i \leq n} (m_s(X_i, A_i, R_i, \theta) - E_P[m_s(X_i, A_i, R_i, \theta)]) \bigg | \bigg ] \\
& \le E \bigg [ \sup_{\theta \in \Theta^\ast}~ \bigg | \frac{1}{n} \sum_{1 \leq i \leq n} A_i (m_s(X_i, 1, R_i(1), \theta) - \eta E[m_s(X_i, 1, R_i(1), \theta)]) \bigg | \bigg ] \\
& \hspace{1em} + E \bigg [ \sup_{\theta \in \Theta^\ast}~ \bigg | \frac{1}{n} \sum_{1 \leq i \leq n} (1 - A_i) (m_s(X_i, 0, R_i(0), \theta) - (1 - \eta) E[m_s(X_i, 0, R_i(0), \theta)]) \bigg | \bigg ] \\
& = E \bigg [ \sup_{\theta \in \Theta^\ast}~ \bigg | \frac{1}{n} \sum_{1 \leq i \leq n} A_i (m_s(X_i, 1, R_i(1), \theta) - E[m_s(X_i, 1, R_i(1), \theta)]) \bigg | \bigg ] \\
& \hspace{1em} + E \bigg [ \sup_{\theta \in \Theta^\ast}~ \bigg |  \frac{1}{n} \sum_{1 \leq i \leq n}(1 - A_i) (m_s(X_i, 0, R_i(0), \theta) -  E[m_s(X_i, 0, R_i(0), \theta)]) \bigg | \bigg ] \\
& \lesssim E \bigg [ \sup_{\theta \in \Theta^\ast}~ \bigg | \frac{1}{n} \sum_{1 \leq i \leq n} (m_s(X_i, 1, R_i(1), \theta) - E[m_s(X_i, 1, R_i(1), \theta)]) \bigg | \bigg ] \\
& \hspace{1em} + E \bigg [ \sup_{\theta \in \Theta^\ast}~ \bigg | \frac{1}{n} \sum_{1 \leq i \leq n} (m_s(X_i, 0, R_i(0), \theta) - E[m_s(X_i, 0, R_i(0), \theta)]) \bigg | \bigg ] \rightarrow 0~.
\end{align*}
where the first equality follows from \eqref{eq:countable}, the first inequality follows from the triangle inequality, the second equality follows because $\sum_{1 \leq i \leq n} A_i = n \eta$ and $\sum_{1 \leq i \leq n} (1 - A_i) = n (1 - \eta)$, and the last inequality follows from Proposition C.1 in \cite{han2021complex}. The convergence follows from Assumption \ref{ass:normal}(e) and an application of the backward submartingale convergence theorem \citep[see, for instance, Theorem 12.30 in][]{le2022measure}, as detailed in the proof of Theorem 3.1 in \cite{han2021complex}. The desired result in \eqref{eq:uc} then follows by Markov's inequality.
\end{proof}

\begin{lemma} \label{eq:varest-fixedsequence}
Suppose the treatment assignment mechanism satisfies Assumptions \ref{ass:a}, \ref{ass:pair}, and \ref{ass:pairsofblocks}, and the moment functions satisfy Assumptions \ref{ass:normal} and \ref{ass:var-est}. Then, for each deterministic sequence $\theta_n \to \theta_0$, $\hat \mu_{1, n}^{(s)}(\theta_n) \xrightarrow{P} \mu_1^{(s)}$ for $\mu_1^{(s)}$ in \eqref{eq:firstmoment-consistent} and $\hat \mu_{1, n}^{(s)}(\theta_n)$ in \eqref{eq:sigmahattheta}.
\end{lemma}

\begin{proof}
The conclusion could follow from similar arguments to those in the proof of Lemma \ref{lem:consistency} but we provide a different proof because it will apply generally to other components of the variance estimator. It suffices to prove that
\begin{align}
   \label{eq:varest-conv1} \hat \mu_{1, n}^{(s)}(\theta_n) - E[\hat \mu_{1, n}^{(s)}(\theta_n) | X^{(n)}, A^{(n)}] & \xrightarrow{P} 0 \\
   \label{eq:varest-conv2} E[\hat \mu_{1, n}^{(s)}(\theta_n) | X^{(n)}, A^{(n)}] & \xrightarrow{P} \mu_1^{(s)}~.
\end{align}
The desired result in \eqref{eq:varest-conv1} follows from similar arguments to those in the proof of equation (S.27) in \cite{bai2022inference}, where the uniform integrability condition is replaced by Assumption \ref{ass:var-est}(a). To show \eqref{eq:varest-conv2}, note
\begin{align}
\nonumber & E[\hat \mu_{1, n}^{(s)}(\theta_n) | X^{(n)}, A^{(n)}] - \mu_1^{(s)} \\
\nonumber & = \frac{1}{\eta n} \sum_{1 \leq i \leq n} A_i E[m_s(X_i, 1, R_i(1), \theta_n) | X_i] - E[m_s(X_i, 1, R_i(1), \theta_0)] \\
\label{eq:varest-fixedseq-1} & = \frac{1}{\eta n} \sum_{1 \leq i \leq n} A_i E[m_s(X_i, 1, R_i(1), \theta_n) | X_i] - \frac{1}{n} \sum_{1 \leq i \leq n} E[m_s(X_i, 1, R_i(1), \theta_n) | X_i] \\
\label{eq:varest-fixedseq-2} & \hspace{3em} + \frac{1}{n} \sum_{1 \leq i \leq n} E[m_s(X_i, 1, R_i(1), \theta_n) | X_i] - E[m_s(X_i, 1, R_i(1), \theta_0)]~.
\end{align}
The difference in \eqref{eq:varest-fixedseq-2} converges in probability to zero by Assumption \ref{ass:var-est}(b) and because $E[m(X_i, 1, R_i(1), \theta_n)] \allowbreak \to E[m(X_i, 1, R_i(1), \theta_0)]$ from $\theta_n \to \theta_0$ and Assumption \ref{ass:normal}(c). Next, note the absolute value of \eqref{eq:varest-fixedseq-1} can be written as
\begin{align*}
& \bigg | \frac{k}{n} \sum_{1 \leq j \leq n / k} \frac{1}{\eta k} \sum_{i \in \lambda_j} A_i \bigg ( E[m_s(X_i, 1, R_i(1), \theta_n) | X_i] \\
& \hspace{10em} - \frac{1}{k} \sum_{i' \in \lambda_j} E[m_s(X_{i'}, 1, R_{i'}(1), \theta_n) | X_{i'}] \bigg ) \bigg | \\
& \leq \frac{k}{n} \sum_{1 \leq j \leq n / k} \frac{1}{\eta k} \sum_{i \in \lambda_j} A_i \bigg | E[m_s(X_i, 1, R_i(1), \theta_n) | X_i] - E[m_s(X_{i'}, 1, R_{i'}(1), \theta_n) | X_{i'}] \bigg | \\
& \lesssim \frac{1}{n} \sum_{1 \leq j \leq n / k} \max_{i, i' \in \lambda_j} \|X_i - X_{i'}\|^2 \xrightarrow{P} 0~,
\end{align*}
where the inequality follows from the uniform Lipschitzness condition in Assumption \ref{ass:var-est}(c) and the convergencef follows from Assumption \ref{ass:pair}. Therefore, \eqref{eq:varest-conv1} holds.
\end{proof}

\section{Details for Simulations}

\subsection{Local Average Treatment Effect}
In this section, we present the model specifications and estimators for estimating the LATE as in Example \ref{ex:LATE}. Recall that in this case the moment condition we consider is given by 
\begin{equation*}
    m(X_i, A_i, R_i, \theta) = \frac{Y_i A_i}{\eta} - \frac{Y_i (1 - A_i)}{1 - \eta} - \theta \left ( \frac{D_i A_i}{\eta} - \frac{D_i (1 - A_i)}{1 - \eta} \right )~,
\end{equation*}
with $R_i = (Y_i, D_i)$. The outcome is determined by the relationship $Y_i = D_i Y_i(1) + (1-D_i)Y_i(0)$, where $Y_i(d)$ is again given by \eqref{eq:sims-outcome}. In addition, the take-up decision is determined as $D_i = A_i D_i(1) + (1-A_i)D_i(0)$, where
\begin{align*}
    D_i(0)&=I\left\{\alpha_0+\alpha\left(X_i\right)> \varepsilon_{1, i}\right\}~, \\
    D_i(1)&=\begin{cases}
I\left\{\alpha_1+\alpha\left(X_i\right)> \varepsilon_{2, i}\right\} & \text { if } D_i(0)=0 \\
1 & \text { otherwise }
\end{cases}~,
\end{align*}
where $\alpha_0 = 0.2$, $\alpha_1 = 4$, $\alpha(X_i) = \sum_{1 \leq l \leq 8} \big ( X_{i, l} + \frac{1}{3}(X_{i, l}^2 - 1) \big )$, $\varepsilon_{1, i}, \varepsilon_{2, i} \sim N(0, 4)$, $(X_i, \epsilon_i, \varepsilon_{1, i}, \varepsilon_{2, i})$, $1 \leq i \leq n$ are i.i.d., and for each $1 \leq i \leq n$, $(X_i, \epsilon_i, \varepsilon_{1, i}, \varepsilon_{2, i})$ are independent.

We consider the following three estimators for the LATE:
\begin{itemize}
    \item[] \textbf{Unadjusted Estimator}:
    \begin{equation*}
        \hat\theta_n^{\rm unadj} = \frac{\displaystyle\sum_{1\leq i \leq n} \left ( Y_i A_i - Y_i (1 - A_i) \right )}{\displaystyle\sum_{1\leq i \leq n} \left ( D_i A_i - D_i (1 - A_i) \right )}~.
    \end{equation*}
    \item[] \textbf{Adjusted Estimator 1}:
    \begin{equation*}
        \hat\theta_n^{\rm adj,1} = \frac{\displaystyle\sum_{1 \leq i \leq n} \left ( 2 A_i(Y_i - \hat\mu_1^Y(X_i)) - 2(1- A_i)(Y_i - \hat\mu_0^Y(X_i)) + \hat\mu_1^Y(X_i) - \hat\mu_0^Y(X_i) \right )}{\displaystyle\sum_{1 \leq i \leq n} \left ( 2 A_i(D_i - \hat\mu_1^Y(X_i)) - 2(1- A_i)(D_i - \hat\mu_0^Y(X_i)) + \hat\mu_1^D(X_i) - \hat\mu_0^D(X_i) \right )}~,
    \end{equation*}
    where $\hat\mu_a^Y(X_i)$ is the linear projection of $Y_i$ on $(1, (X_{i, l}, X_{i, l}^2: 1 \leq l \leq T))$ in the subsample with $A_i=a$, and $\hat\mu_a^D(X_i)$ is estimated using from a logistic regression model with the same set of regressors in the subsample with $A_i=a$.
    \item[] \textbf{Adjusted Estimator 2}: 
    As in Adjusted Estimator 1, but in $\hat\mu_a^Y(X_i)$ and $\hat\mu_a^D(X_i)$ the regressors are instead $(1, (X_{i, l}, X_{i, l}^2, X_{i, l} 1\{X_{i, l} > \hat t_l\}: 1 \leq l \leq T))$, where $\hat t_l$ is the sample median of $X_{i, l}$, $1 \leq i \leq n$. 
\end{itemize}
Similarly to Section \ref{sec:sims_ATE}, $\hat{\theta}_n^{\rm unadj}$ solves \eqref{eq:est} for the moment condition given in \eqref{eq:moments-late}. The second and third estimators are covariate adjusted estimators which can be obtained as two-step method of moments estimators from solving an ``augmented'' version of the moment condition \eqref{eq:moments-late} \citep[see, for instance,][]{chernozhukov2018doubledebiased, jiang2022improving}.

\subsection{Additional Simulation Tables}
Table \ref{table:sims2-bias} displays the absolute value of the bias for each design and estimator pair computed across $4000$ Monte Carlo replications. We find that all estimators have fairly low and comparable bias.

Table \ref{table:coverage_ci} reports the coverage and average length of confidence intervals constructed using the estimator described in Section \ref{sec:var-est}, again computed across $4000$ Monte Carlo replications. We find that the confidence intervals have appropriate coverage at all sample sizes and for any number of covariates. The average length of the intervals follows a pattern similar to what was observed for the MSE, again because the first 4 covariates are much stronger predictors of the control outcome than the last 4 covariates, which are almost uninformative. Although we found in Table \ref{table:sims2} that the MSE of the matched tuples design was worse than that of matched pairs with 8 covariates, this does not seem to translate to longer confidence intervals for matched tuples: this may be due to the fact that the matched tuples variance estimator is a ``within block" estimator, as described in Section \ref{sec:var-est}. Similar findings have been reported by \cite{bai2024inference} for related problems in factorial designs.
\begin{table}[ht!]
\centering
\begin{threeparttable}
\caption{Absolute value of bias of estimators for different designs and estimators}
\setlength{\tabcolsep}{3pt} 
\begin{tabular}{cccccccccc}
\toprule
& & & \multicolumn{3}{c}{\textbf{i.i.d.\ assignment}} & & \textbf{Matched pairs} & & \textbf{Matched tuples} \\ 
\cmidrule{4-6} \cmidrule{8-8} \cmidrule{10-10}
& & \# of covariates ($T$) & Unadjusted & Adjusted 1 & Adjusted 2 & & Unadjusted & & Unadjusted \\ 
\midrule

\multicolumn{10}{c}{$n=100$} \\ \addlinespace[0.5em]
& \multirow{3}{*}{ATE} & 2 & 0.0083 & 0.0065 & 0.0051 & & 0.0212 & & 0.0068 \\
& & 4 & 0.0083 & 0.0090 & 0.0090 & & 0.0073 & & 0.0124 \\
& & 8 & 0.0083 & 0.0134 & 0.0131 & & 0.0010 & & 0.0209 \\ \addlinespace[0.5em]
& \multirow{3}{*}{LATE} & 2 & 0.0905 & 0.0436 & 0.0396 & & 0.0341 & & 0.0076 \\
& & 4 & 0.0905 & 0.0276 & 0.0264 & & 0.0301 & & 0.0427 \\
& & 8 & 0.0905 & 0.0401 & 0.0376 & & 0.0264 & & 0.0297 \\ \addlinespace[1.0em]

\multicolumn{10}{c}{$n=200$} \\ \addlinespace[0.5em]
& \multirow{3}{*}{ATE} & 2 & 0.0075 & 0.0068 & 0.0074 & & 0.0084 & & 0.0081 \\
& & 4 & 0.0075 & 0.0052 & 0.0067 & & 0.0040 & & 0.0040 \\
& & 8 & 0.0075 & 0.0045 & 0.0055 & & 0.0145 & & 0.0051 \\ \addlinespace[0.5em]
& \multirow{3}{*}{LATE} & 2 & 0.0545 & 0.0326 & 0.0339 & & 0.0361 & & 0.0328 \\
& & 4 & 0.0545 & 0.0182 & 0.0217 & & 0.0123 & & 0.0144 \\
& & 8 & 0.0545 & 0.0159 & 0.0180 & & 0.0447 & & 0.0216 \\ \addlinespace[1.0em]

\multicolumn{10}{c}{$n=400$} \\ \addlinespace[0.5em]
& \multirow{3}{*}{ATE} & 2 & 0.0001 & 0.0011 & 0.0009 & & 0.0041 & & 0.0055 \\
& & 4 & 0.0001 & 0.0001 & 0.0001 & & 0.0001 & & 0.0016 \\
& & 8 & 0.0001 & 0.0006 & 0.0006 & & 0.0056 & & 0.0041 \\ \addlinespace[0.5em]
& \multirow{3}{*}{LATE} & 2 & 0.0142 & 0.0064 & 0.0059 & & 0.0179 & & 0.0198 \\
& & 4 & 0.0142 & 0.0012 & 0.0015 & & 0.0004 & & 0.0020 \\
& & 8 & 0.0142 & 0.0021 & 0.0019 & & 0.0129 & & 0.0064 \\ \addlinespace[1.0em]

\multicolumn{10}{c}{$n=1000$} \\ \addlinespace[0.5em]
& \multirow{3}{*}{ATE} & 2 & 0.0021 & 0.0037 & 0.0039 & & 0.0014 & & 0.0025 \\
& & 4 & 0.0021 & 0.0021 & 0.0022 & & 0.0020 & & 0.0006 \\
& & 8 & 0.0021 & 0.0021 & 0.0023 & & 0.0023 & & 0.0027 \\ \addlinespace[0.5em]
& \multirow{3}{*}{LATE} & 2 & 0.0118 & 0.0120 & 0.0123 & & 0.0019 & & 0.0083 \\
& & 4 & 0.0118 & 0.0055 & 0.0057 & & 0.0051 & & 0.0007 \\
& & 8 & 0.0118 & 0.0054 & 0.0058 & & 0.0057 & & 0.0059 \\ \addlinespace[1.0em]

\multicolumn{10}{c}{$n=2000$} \\ \addlinespace[0.5em]
& \multirow{3}{*}{ATE} & 2 & 0.0013 & 0.0010 & 0.0010 & & 0.0011 & & 0.0025 \\
& & 4 & 0.0013 & 0.0015 & 0.0015 & & 0.0003 & & 0.0021 \\
& & 8 & 0.0013 & 0.0015 & 0.0015 & & 0.0001 & & 0.0014 \\ \addlinespace[0.5em]
& \multirow{3}{*}{LATE} & 2 & 0.0062 & 0.0035 & 0.0035 & & 0.0018 & & 0.0050 \\
& & 4 & 0.0062 & 0.0037 & 0.0038 & & 0.0011 & & 0.0055 \\
& & 8 & 0.0062 & 0.0038 & 0.0037 & & 0.0008 & & 0.0032 \\ 
\bottomrule
\end{tabular}
\begin{tablenotes}
\footnotesize
\item Note: Entries report the absolute bias of each estimator under different sample sizes, designs, and covariate adjustment strategies. Results are averaged across 2000 replications.
\end{tablenotes}
\label{table:sims2-bias}
\end{threeparttable}
\end{table}

\begin{table}[ht!]
\centering
\begin{threeparttable}
\caption{Coverage rates and lengths of 95\% confidence intervals based on unadjusted estimator for matched pairs and matched tuples}
\setlength{\tabcolsep}{6pt}
\begin{tabular}{ccc cc cc}
\toprule
& & & \multicolumn{2}{c}{\textbf{Matched pairs}} & \multicolumn{2}{c}{\textbf{Matched tuples}} \\
\cmidrule(lr){4-5} \cmidrule(lr){6-7}
& & \# of covariates ($T$) & Coverage & CI length & Coverage & CI length \\
\midrule

\multicolumn{7}{c}{$n=100$} \\ \addlinespace[0.5em]
& \multirow{3}{*}{ATE} & 2 & 0.9403 & 2.3191 & 0.9490 & 2.3224 \\
& & 4 & 0.9465 & 2.2744 & 0.9485 & 2.2830 \\
& & 8 & 0.9487 & 2.3532 & 0.9395 & 2.3672 \\ \addlinespace[0.5em]
& \multirow{3}{*}{LATE} & 2 & 0.9590 & 6.4075 & 0.9590 & 6.3865 \\
& & 4 & 0.9620 & 6.3184 & 0.9615 & 6.3176 \\
& & 8 & 0.9610 & 6.5497 & 0.9527 & 6.5311 \\ \addlinespace[1.0em]

\multicolumn{7}{c}{$n=200$} \\ \addlinespace[0.5em]
& \multirow{3}{*}{ATE} & 2 & 0.9495 & 1.6401 & 0.9497 & 1.6418 \\
& & 4 & 0.9485 & 1.5990 & 0.9440 & 1.6049 \\
& & 8 & 0.9503 & 1.6529 & 0.9425 & 1.6617 \\ \addlinespace[0.5em]
& \multirow{3}{*}{LATE} & 2 & 0.9550 & 4.4046 & 0.9563 & 4.4264 \\
& & 4 & 0.9557 & 4.3315 & 0.9530 & 4.3381 \\
& & 8 & 0.9597 & 4.4833 & 0.9503 & 4.5038 \\ \addlinespace[1.0em]

\multicolumn{7}{c}{$n=400$} \\ \addlinespace[0.5em]
& \multirow{3}{*}{ATE} & 2 & 0.9487 & 1.1602 & 0.9483 & 1.1614 \\
& & 4 & 0.9497 & 1.1255 & 0.9455 & 1.1285 \\
& & 8 & 0.9495 & 1.1609 & 0.9445 & 1.1667 \\ \addlinespace[0.5em]
& \multirow{3}{*}{LATE} & 2 & 0.9530 & 3.0955 & 0.9493 & 3.0963 \\
& & 4 & 0.9547 & 3.0122 & 0.9517 & 3.0110 \\
& & 8 & 0.9555 & 3.1284 & 0.9505 & 3.1239 \\ \addlinespace[1.0em]

\multicolumn{7}{c}{$n=1000$} \\ \addlinespace[0.5em]
& \multirow{3}{*}{ATE} & 2 & 0.9485 & 0.7349 & 0.9550 & 0.7352 \\
& & 4 & 0.9500 & 0.7086 & 0.9505 & 0.7102 \\
& & 8 & 0.9417 & 0.7285 & 0.9430 & 0.7318 \\ \addlinespace[0.5em]
& \multirow{3}{*}{LATE} & 2 & 0.9507 & 1.9468 & 0.9560 & 1.9452 \\
& & 4 & 0.9515 & 1.8836 & 0.9507 & 1.8857 \\
& & 8 & 0.9467 & 1.9469 & 0.9463 & 1.9480 \\ \addlinespace[1.0em]

\multicolumn{7}{c}{$n=2000$} \\ \addlinespace[0.5em]
& \multirow{3}{*}{ATE} & 2 & 0.9490 & 0.5194 & 0.9530 & 0.5193 \\
& & 4 & 0.9483 & 0.4998 & 0.9485 & 0.5004 \\
& & 8 & 0.9505 & 0.5121 & 0.9483 & 0.5139 \\ \addlinespace[0.5em]
& \multirow{3}{*}{LATE} & 2 & 0.9495 & 1.3714 & 0.9530 & 1.3721 \\
& & 4 & 0.9493 & 1.3250 & 0.9483 & 1.3258 \\
& & 8 & 0.9533 & 1.3663 & 0.9505 & 1.3658 \\
\bottomrule
\end{tabular}
\label{table:coverage_ci}
\end{threeparttable}
\end{table}

\clearpage

\bibliography{efficiency}

@article{hanneke2021universal,
	title = {Universal {Bayes} consistency in metric spaces},
	volume = {49},
	issn = {0090-5364, 2168-8966},
	url = {https://projecteuclid.org/journals/annals-of-statistics/volume-49/issue-4/Universal-Bayes-consistency-in-metric-spaces/10.1214/20-AOS2029.full},
	doi = {10.1214/20-AOS2029},
	number = {4},
	urldate = {2025-02-28},
	journal = {The Annals of Statistics},
	author = {Hanneke, Steve and Kontorovich, Aryeh and Sabato, Sivan and Weiss, Roi},
	month = aug,
	year = {2021},
	note = {Publisher: Institute of Mathematical Statistics},
	pages = {2129--2150},
}

@article{fogarty2018mitigating,
	title = {On mitigating the analytical limitations of finely stratified experiments},
	volume = {80},
	copyright = {© 2018 Royal Statistical Society},
	issn = {1467-9868},
	url = {https://rss.onlinelibrary.wiley.com/doi/abs/10.1111/rssb.12290},
	doi = {https://doi.org/10.1111/rssb.12290},
	language = {en},
	number = {5},
	urldate = {2020-12-03},
	journal = {Journal of the Royal Statistical Society: Series B (Statistical Methodology)},
	author = {Fogarty, Colin B.},
	year = {2018},
	note = {\_eprint: https://rss.onlinelibrary.wiley.com/doi/pdf/10.1111/rssb.12290},
	pages = {1035--1056},
}

@article{fogarty2018regression-assisted,
	title = {Regression-assisted inference for the average treatment effect in paired experiments},
	volume = {105},
	issn = {0006-3444},
	url = {https://academic.oup.com/biomet/article/105/4/994/5047363},
	doi = {10.1093/biomet/asy034},
	language = {en},
	number = {4},
	urldate = {2020-12-03},
	journal = {Biometrika},
	author = {Fogarty, Colin B.},
	month = dec,
	year = {2018},
	note = {Publisher: Oxford Academic},
	pages = {994--1000},
}

@article{imai2008variance,
	title = {Variance identification and efficiency analysis in randomized experiments under the matched-pair design},
	volume = {27},
	number = {24},
	journal = {Statistics in medicine},
	author = {Imai, Kosuke},
	year = {2008},
	note = {Publisher: Wiley Online Library},
	pages = {4857--4873},
}

@article{imai2009essential,
	title = {The {Essential} {Role} of {Pair} {Matching} in {Cluster}-{Randomized} {Experiments}, with {Application} to the {Mexican} {Universal} {Health} {Insurance} {Evaluation}},
	volume = {24},
	issn = {0883-4237, 2168-8745},
	url = {https://projecteuclid.org/euclid.ss/1255009008},
	doi = {10.1214/08-STS274},
	language = {EN},
	number = {1},
	urldate = {2020-12-03},
	journal = {Statistical Science},
	author = {Imai, Kosuke and King, Gary and Nall, Clayton},
	month = feb,
	year = {2009},
	mrnumber = {MR2561126},
	zmnumber = {1327.62061},
	note = {Publisher: Institute of Mathematical Statistics},
	pages = {29--53},
}

@article{zelen1974randomization,
	title = {The randomization and stratification of patients to clinical trials},
	volume = {27},
	issn = {0021-9681},
	url = {https://www.sciencedirect.com/science/article/pii/0021968174900150},
	doi = {10.1016/0021-9681(74)90015-0},
	number = {7},
	urldate = {2025-02-19},
	journal = {Journal of Chronic Diseases},
	author = {Zelen, M.},
	month = sep,
	year = {1974},
	pages = {365--375},
}

@article{li2017general,
	title = {General {Forms} of {Finite} {Population} {Central} {Limit} {Theorems} with {Applications} to {Causal} {Inference}},
	volume = {112},
	issn = {0162-1459},
	url = {https://doi.org/10.1080/01621459.2017.1295865},
	doi = {10.1080/01621459.2017.1295865},
	number = {520},
	urldate = {2025-02-19},
	journal = {Journal of the American Statistical Association},
	author = {Li, Xinran and Ding, Peng},
	month = oct,
	year = {2017},
	note = {Publisher: ASA Website
\_eprint: https://doi.org/10.1080/01621459.2017.1295865},
	pages = {1759--1769},
}

@article{hu2012asymptotic,
	title = {Asymptotic properties of covariate-adaptive randomization},
	volume = {40},
	issn = {0090-5364, 2168-8966},
	url = {https://projecteuclid.org/journals/annals-of-statistics/volume-40/issue-3/Asymptotic-properties-of-covariate-adaptive-randomization/10.1214/12-AOS983.full},
	doi = {10.1214/12-AOS983},
	number = {3},
	urldate = {2025-02-19},
	journal = {The Annals of Statistics},
	author = {Hu, Yanqing and Hu, Feifang},
	month = jun,
	year = {2012},
	note = {Publisher: Institute of Mathematical Statistics},
	pages = {1794--1815},
}

@article{li2020rerandomization,
	title = {Rerandomization and {Regression} {Adjustment}},
	volume = {82},
	issn = {1369-7412},
	url = {https://doi.org/10.1111/rssb.12353},
	doi = {10.1111/rssb.12353},
	number = {1},
	urldate = {2025-02-19},
	journal = {Journal of the Royal Statistical Society Series B: Statistical Methodology},
	author = {Li, Xinran and Ding, Peng},
	month = feb,
	year = {2020},
	pages = {241--268},
}

@article{li2020rerandomization-1,
	title = {Rerandomization in \$2{\textasciicircum}\{{K}\}\$ factorial experiments},
	volume = {48},
	issn = {0090-5364, 2168-8966},
	url = {https://projecteuclid.org/journals/annals-of-statistics/volume-48/issue-1/Rerandomization-in-2K-factorial-experiments/10.1214/18-AOS1790.full},
	doi = {10.1214/18-AOS1790},
	number = {1},
	urldate = {2025-02-19},
	journal = {The Annals of Statistics},
	author = {Li, Xinran and Ding, Peng and Rubin, Donald B.},
	month = feb,
	year = {2020},
	note = {Publisher: Institute of Mathematical Statistics},
	pages = {43--63},
}

@article{li2018asymptotic,
	title = {Asymptotic theory of rerandomization in treatment–control experiments},
	volume = {115},
	copyright = {© 2018 . http://www.pnas.org/site/aboutpnas/licenses.xhtmlPublished under the PNAS license.},
	issn = {0027-8424, 1091-6490},
	url = {https://www.pnas.org/content/115/37/9157},
	doi = {10.1073/pnas.1808191115},
	language = {en},
	number = {37},
	urldate = {2020-12-03},
	journal = {Proceedings of the National Academy of Sciences},
	author = {Li, Xinran and Ding, Peng and Rubin, Donald B.},
	month = sep,
	year = {2018},
	pmid = {30150408},
	note = {Publisher: National Academy of Sciences
Section: Physical Sciences},
	pages = {9157--9162},
}

@article{pashley2021insights,
  title={Insights on variance estimation for blocked and matched pairs designs},
  author={Pashley, Nicole E and Miratrix, Luke W},
  journal={Journal of Educational and Behavioral Statistics},
  volume={46},
  number={3},
  pages={271--296},
  year={2021},
  publisher={SAGE Publications Sage CA: Los Angeles, CA}
}

@article{ye2022inference,
	title = {Inference on the average treatment effect under minimization and other covariate-adaptive randomization methods},
	volume = {109},
	issn = {1464-3510},
	url = {https://doi.org/10.1093/biomet/asab015},
	doi = {10.1093/biomet/asab015},
	number = {1},
	urldate = {2025-02-19},
	journal = {Biometrika},
	author = {Ye, Ting and Yi, Yanyao and Shao, Jun},
	month = mar,
	year = {2022},
	pages = {33--47},
}

@article{liu2020regression-adjusted,
	title = {Regression-adjusted average treatment effect estimates in stratified randomized experiments},
	volume = {107},
	issn = {0006-3444},
	url = {https://doi.org/10.1093/biomet/asaa038},
	doi = {10.1093/biomet/asaa038},
	number = {4},
	urldate = {2025-02-19},
	journal = {Biometrika},
	author = {Liu, Hanzhong and Yang, Yuehan},
	month = dec,
	year = {2020},
	pages = {935--948},
}

@article{ma2020statistical,
	title = {Statistical {Inference} for {Covariate}-{Adaptive} {Randomization} {Procedures}},
	volume = {115},
	issn = {0162-1459},
	url = {https://doi.org/10.1080/01621459.2019.1635483},
	doi = {10.1080/01621459.2019.1635483},
	number = {531},
	urldate = {2025-02-19},
	journal = {Journal of the American Statistical Association},
	author = {Ma, Wei and Qin, Yichen and Li, Yang and Hu, Feifang},
	month = jul,
	year = {2020},
	note = {Publisher: ASA Website
\_eprint: https://doi.org/10.1080/01621459.2019.1635483},
	pages = {1488--1497},
}

@article{ma2024new,
	title = {A {New} and {Unified} {Family} of {Covariate} {Adaptive} {Randomization} {Procedures} and {Their} {Properties}},
	volume = {119},
	issn = {0162-1459},
	url = {https://doi.org/10.1080/01621459.2022.2102986},
	doi = {10.1080/01621459.2022.2102986},
	number = {545},
	urldate = {2025-02-14},
	journal = {Journal of the American Statistical Association},
	author = {Ma, Wei and Li, Ping and Zhang, Li-Xin and Hu, Feifang},
	month = jan,
	year = {2024},
	note = {Publisher: ASA Website
\_eprint: https://doi.org/10.1080/01621459.2022.2102986},
	pages = {151--162},
}

@book{le2022measure,
  title={Measure theory, probability, and stochastic processes},
  author={Le Gall, Jean-Fran{\c{c}}ois},
  year={2022},
  publisher={Springer}
}

@article{han2021complex,
	title = {Complex sampling designs: {Uniform} limit theorems and applications},
	volume = {49},
	issn = {0090-5364, 2168-8966},
	shorttitle = {Complex sampling designs},
	url = {https://projecteuclid.org/journals/annals-of-statistics/volume-49/issue-1/Complex-sampling-designs-Uniform-limit-theorems-and-applications/10.1214/20-AOS1964.full},
	doi = {10.1214/20-AOS1964},
	number = {1},
	urldate = {2025-02-03},
	journal = {The Annals of Statistics},
	author = {Han, Qiyang and Wellner, Jon A.},
	month = feb,
	year = {2021},
	note = {Publisher: Institute of Mathematical Statistics},
	pages = {459--485},
}

@misc{cai2022performance,
	title = {On the {Performance} of the {Neyman} {Allocation} with {Small} {Pilots}},
	url = {http://arxiv.org/abs/2206.04643},
	doi = {10.48550/arXiv.2206.04643},
	urldate = {2024-01-16},
	publisher = {arXiv},
	author = {Cai, Yong and Rafi, Ahnaf},
	month = aug,
	year = {2022},
	note = {arXiv:2206.04643 [econ]},
}

@techreport{imbens2011experimental,
	title = {Experimental {Design} for {Unit} and {Cluster} {Randomized} {Trials}},
	author = {Imbens, Guido W},
	year = {2011},
}

@article{hirano2001estimation,
	title = {Estimation of {Causal} {Effects} using {Propensity} {Score} {Weighting}: {An} {Application} to {Data} on {Right} {Heart} {Catheterization}},
	volume = {2},
	issn = {1572-9400},
	shorttitle = {Estimation of {Causal} {Effects} using {Propensity} {Score} {Weighting}},
	url = {https://doi.org/10.1023/A:1020371312283},
	doi = {10.1023/A:1020371312283},
	language = {en},
	number = {3},
	urldate = {2021-01-31},
	journal = {Health Services and Outcomes Research Methodology},
	author = {Hirano, Keisuke and Imbens, Guido W.},
	month = dec,
	year = {2001},
	pages = {259--278},
}

@article{newey1994large,
  title={Large sample estimation and hypothesis testing},
  author={Newey, Whitney K and McFadden, Daniel},
  journal={Handbook of econometrics},
  volume={4},
  pages={2111--2245},
  year={1994},
  publisher={Elsevier}
}

@article{newey1994asymptotic,
	title = {The {Asymptotic} {Variance} of {Semiparametric} {Estimators}},
	volume = {62},
	issn = {0012-9682},
	url = {https://www.jstor.org/stable/2951752},
	doi = {10.2307/2951752},
	number = {6},
	urldate = {2021-04-28},
	journal = {Econometrica},
	author = {Newey, Whitney K.},
	year = {1994},
	note = {Publisher: [Wiley, Econometric Society]},
	pages = {1349--1382},
}

@article{hirano2003efficient,
	title = {Efficient {Estimation} of {Average} {Treatment} {Effects} {Using} the {Estimated} {Propensity} {Score}},
	volume = {71},
	issn = {1468-0262},
	url = {https://onlinelibrary.wiley.com/doi/abs/10.1111/1468-0262.00442},
	doi = {https://doi.org/10.1111/1468-0262.00442},
	language = {en},
	number = {4},
	urldate = {2021-02-02},
	journal = {Econometrica},
	author = {Hirano, Keisuke and Imbens, Guido W. and Ridder, Geert},
	year = {2003},
	note = {\_eprint: https://onlinelibrary.wiley.com/doi/pdf/10.1111/1468-0262.00442},
	pages = {1161--1189},
}

@book{van_der_vaart1996weak,
	address = {New York},
	series = {Springer {Series} in {Statistics}},
	title = {Weak {Convergence} and {Empirical} {Processes}: {With} {Applications} to {Statistics}},
	isbn = {978-0-387-94640-5},
	shorttitle = {Weak {Convergence} and {Empirical} {Processes}},
	url = {https://www.springer.com/gp/book/9780387946405},
	language = {en},
	urldate = {2020-12-03},
	publisher = {Springer-Verlag},
	author = {van der Vaart, A. W. and Wellner, Jon},
	year = {1996},
	doi = {10.1007/978-1-4757-2545-2},
}

@book{kosorok2008introduction,
	address = {New York},
	series = {Springer {Series} in {Statistics}},
	title = {Introduction to {Empirical} {Processes} and {Semiparametric} {Inference}},
	isbn = {978-0-387-74977-8},
	url = {https://www.springer.com/gp/book/9780387749778},
	language = {en},
	urldate = {2020-12-03},
	publisher = {Springer-Verlag},
	author = {Kosorok, Michael R.},
	year = {2008},
	doi = {10.1007/978-0-387-74978-5},
}

@book{van_der_vaart1998asymptotic,
	series = {Cambridge {Series} in {Statistical} and {Probabilistic} {Mathematics}},
	title = {Asymptotic statistics},
	volume = {3},
	isbn = {0-521-49603-9 0-521-78450-6},
	url = {https://doi-org.proxy.lib.umich.edu/10.1017/CBO9780511802256},
	publisher = {Cambridge University Press, Cambridge},
	author = {van der Vaart, A. W.},
	year = {1998},
	mrnumber = {1652247},
	doi = {10.1017/CBO9780511802256},
}

@article{bugni2018inference,
	title = {Inference {Under} {Covariate}-{Adaptive} {Randomization}},
	volume = {113},
	issn = {0162-1459},
	url = {https://doi.org/10.1080/01621459.2017.1375934},
	doi = {10.1080/01621459.2017.1375934},
	number = {524},
	urldate = {2020-12-03},
	journal = {Journal of the American Statistical Association},
	author = {Bugni, Federico A. and Canay, Ivan A. and Shaikh, Azeem M.},
	month = oct,
	year = {2018},
	pmid = {30906087},
	note = {Publisher: Taylor \& Francis
\_eprint: https://doi.org/10.1080/01621459.2017.1375934},
	pages = {1784--1796},
}

@article{bugni2019inference,
	title = {Inference under covariate-adaptive randomization with multiple treatments},
	volume = {10},
	copyright = {Copyright © 2019 The Authors.},
	issn = {1759-7331},
	url = {https://onlinelibrary.wiley.com/doi/abs/10.3982/QE1150},
	doi = {https://doi.org/10.3982/QE1150},
	language = {en},
	number = {4},
	urldate = {2020-12-03},
	journal = {Quantitative Economics},
	author = {Bugni, Federico A. and Canay, Ivan A. and Shaikh, Azeem M.},
	year = {2019},
	note = {\_eprint: https://onlinelibrary.wiley.com/doi/pdf/10.3982/QE1150},
	pages = {1747--1785},
}

@article{chernozhukov2017doubledebiasedneyman,
	title = {Double/{Debiased}/{Neyman} {Machine} {Learning} of {Treatment} {Effects}},
	volume = {107},
	issn = {0002-8282},
	url = {http://www.aeaweb.org/articles?id=10.1257/aer.p20171038},
	doi = {10.1257/aer.p20171038},
	language = {en},
	number = {5},
	urldate = {2020-12-03},
	journal = {American Economic Review},
	author = {Chernozhukov, Victor and Chetverikov, Denis and Demirer, Mert and Duflo, Esther and Hansen, Christian and Newey, Whitney},
	month = may,
	year = {2017},
	pages = {261--265},
}

@article{hahn1998role,
	title = {On the {Role} of the {Propensity} {Score} in {Efficient} {Semiparametric} {Estimation} of {Average} {Treatment} {Effects}},
	volume = {66},
	issn = {0012-9682},
	url = {https://www.jstor.org/stable/2998560},
	doi = {10.2307/2998560},
	number = {2},
	urldate = {2020-12-03},
	journal = {Econometrica},
	author = {Hahn, Jinyong},
	year = {1998},
	note = {Publisher: [Wiley, Econometric Society]},
	pages = {315--331},
}

@misc{bugni2022inference,
	title = {Inference for {Cluster} {Randomized} {Experiments} with {Non}-ignorable {Cluster} {Sizes}},
	url = {http://arxiv.org/abs/2204.08356},
	doi = {10.48550/arXiv.2204.08356},
	urldate = {2022-10-14},
	publisher = {arXiv},
	author = {Bugni, Federico and Canay, Ivan and Shaikh, Azeem and Tabord-Meehan, Max},
	month = jun,
	year = {2022},
	note = {arXiv:2204.08356 [econ, stat]},
}

@article{robins1995analysis,
	title = {Analysis of {Semiparametric} {Regression} {Models} for {Repeated} {Outcomes} in the {Presence} of {Missing} {Data}},
	volume = {90},
	issn = {0162-1459},
	url = {https://www.jstor.org/stable/2291134},
	doi = {10.2307/2291134},
	number = {429},
	urldate = {2022-10-24},
	journal = {Journal of the American Statistical Association},
	author = {Robins, James M. and Rotnitzky, Andrea and Zhao, Lue Ping},
	year = {1995},
	note = {Publisher: [American Statistical Association, Taylor \& Francis, Ltd.]},
	pages = {106--121},
}

@article{bai2022optimality,
	title = {Optimality of {Matched}-{Pair} {Designs} in {Randomized} {Controlled} {Trials}},
	volume = {112},
	issn = {0002-8282},
	url = {https://www.aeaweb.org/articles?id=10.1257/aer.20201856},
	doi = {10.1257/aer.20201856},
	language = {en},
	number = {12},
	urldate = {2022-11-28},
	journal = {American Economic Review},
	author = {Bai, Yuehao},
	month = dec,
	year = {2022},
	pages = {3911--3940},
}

@article{bai2022inference,
	title = {Inference in {Experiments} {With} {Matched} {Pairs}},
	volume = {117},
	issn = {0162-1459},
	url = {https://doi.org/10.1080/01621459.2021.1883437},
	doi = {10.1080/01621459.2021.1883437},
	number = {540},
	urldate = {2023-01-30},
	journal = {Journal of the American Statistical Association},
	author = {Bai, Yuehao and Romano, Joseph P. and Shaikh, Azeem M.},
	month = oct,
	year = {2022},
	note = {Publisher: Taylor \& Francis
\_eprint: https://doi.org/10.1080/01621459.2021.1883437},
	pages = {1726--1737},
}

@article{tsiatis2008covariate,
	title = {Covariate adjustment for two-sample treatment comparisons in randomized clinical trials: a principled yet flexible approach},
	volume = {27},
	issn = {0277-6715},
	shorttitle = {Covariate adjustment for two-sample treatment comparisons in randomized clinical trials},
	doi = {10.1002/sim.3113},
	language = {eng},
	number = {23},
	journal = {Statistics in Medicine},
	author = {Tsiatis, Anastasios A. and Davidian, Marie and Zhang, Min and Lu, Xiaomin},
	month = oct,
	year = {2008},
	pmid = {17960577},
	pmcid = {PMC2562926},
	pages = {4658--4677},
}

@article{van_der_vaart1989asymptotic,
	title = {On the {Asymptotic} {Information} {Bound}},
	volume = {17},
	issn = {0090-5364},
	url = {https://projecteuclid.org/journals/annals-of-statistics/volume-17/issue-4/On-the-Asymptotic-Information-Bound/10.1214/aos/1176347377.full},
	doi = {10.1214/aos/1176347377},
	language = {en},
	number = {4},
	urldate = {2023-04-04},
	journal = {The Annals of Statistics},
	author = {van der Vaart, Aad},
	month = dec,
	year = {1989},
}

@book{lehmann2022testing,
	address = {Cham},
	series = {Springer {Texts} in {Statistics}},
	title = {Testing {Statistical} {Hypotheses}},
	isbn = {978-3-030-70577-0 978-3-030-70578-7},
	url = {https://link.springer.com/10.1007/978-3-030-70578-7},
	language = {en},
	urldate = {2023-04-13},
	publisher = {Springer International Publishing},
	author = {Lehmann, E.L. and Romano, Joseph P.},
	year = {2022},
	doi = {10.1007/978-3-030-70578-7},
}

@article{firpo2007efficient,
	title = {Efficient {Semiparametric} {Estimation} of {Quantile} {Treatment} {Effects}},
	volume = {75},
	issn = {0012-9682},
	url = {https://www.jstor.org/stable/4123114},
	number = {1},
	urldate = {2023-04-18},
	journal = {Econometrica},
	author = {Firpo, Sergio},
	year = {2007},
	note = {Publisher: [Wiley, Econometric Society]},
	pages = {259--276},
}

@book{krantz2013implicit,
	address = {New York, NY},
	title = {The {Implicit} {Function} {Theorem}: {History}, {Theory}, and {Applications}},
	isbn = {978-1-4614-5980-4 978-1-4614-5981-1},
	shorttitle = {The {Implicit} {Function} {Theorem}},
	url = {https://link.springer.com/10.1007/978-1-4614-5981-1},
	language = {en},
	urldate = {2023-05-18},
	publisher = {Springer},
	author = {Krantz, Steven G. and Parks, Harold R.},
	year = {2013},
	doi = {10.1007/978-1-4614-5981-1},
}

@misc{armstrong2022asymptotic,
	title = {Asymptotic {Efficiency} {Bounds} for a {Class} of {Experimental} {Designs}},
	url = {http://arxiv.org/abs/2205.02726},
	doi = {10.48550/arXiv.2205.02726},
	urldate = {2023-01-27},
	publisher = {arXiv},
	author = {Armstrong, Timothy B.},
	month = may,
	year = {2022},
	note = {arXiv:2205.02726 [stat]},
}

@misc{jiang2022regression-adjusted,
	title = {Regression-{Adjusted} {Estimation} of {Quantile} {Treatment} {Effects} under {Covariate}-{Adaptive} {Randomizations}},
	url = {http://arxiv.org/abs/2105.14752},
	doi = {10.48550/arXiv.2105.14752},
	urldate = {2023-06-13},
	publisher = {arXiv},
	author = {Jiang, Liang and Phillips, Peter C. B. and Tao, Yubo and Zhang, Yichong},
	month = sep,
	year = {2022},
	note = {arXiv:2105.14752 [econ, stat]},
}

@misc{jiang2022improving,
	title = {Improving {Estimation} {Efficiency} via {Regression}-{Adjustment} in {Covariate}-{Adaptive} {Randomizations} with {Imperfect} {Compliance}},
	url = {http://arxiv.org/abs/2201.13004},
	doi = {10.48550/arXiv.2201.13004},
	urldate = {2023-06-13},
	publisher = {arXiv},
	author = {Jiang, Liang and Linton, Oliver B. and Tang, Haihan and Zhang, Yichong},
	month = sep,
	year = {2022},
	note = {arXiv:2201.13004 [econ, stat]},
}

@article{imbens1994identification,
	title = {Identification and {Estimation} of {Local} {Average} {Treatment} {Effects}},
	volume = {62},
	issn = {0012-9682},
	url = {https://www.jstor.org/stable/2951620},
	doi = {10.2307/2951620},
	number = {2},
	urldate = {2020-12-03},
	journal = {Econometrica},
	author = {Imbens, Guido W. and Angrist, Joshua D.},
	year = {1994},
	note = {Publisher: [Wiley, Econometric Society]},
	pages = {467--475},
}

@article{zhang2008improving,
  title={Improving efficiency of inferences in randomized clinical trials using auxiliary covariates},
  author={Zhang, Min and Tsiatis, Anastasios A and Davidian, Marie},
  journal={Biometrics},
  volume={64},
  number={3},
  pages={707--715},
  year={2008},
  publisher={Wiley Online Library}
}

@article{rafi2023efficient,
  title={Efficient Semiparametric Estimation of Average Treatment Effects Under Covariate Adaptive Randomization},
  author={Rafi, Ahnaf},
  journal={arXiv preprint arXiv:2305.08340},
  year={2023}
}

@article{jiang2021bootstrap,
  title={Bootstrap inference for quantile treatment effects in randomized experiments with matched pairs},
  author={Jiang, Liang and Liu, Xiaobin and Phillips, Peter CB and Zhang, Yichong},
  journal={Review of Economics and Statistics},
  pages={1--47},
  year={2021},
  publisher={MIT Press One Rogers Street, Cambridge, MA 02142-1209, USA journals-info~…}
}

@article{lin2013agnostic,
	title = {Agnostic notes on regression adjustments to experimental data: {Reexamining} {Freedman}’s critique},
	volume = {7},
	issn = {1932-6157, 1941-7330},
	shorttitle = {Agnostic notes on regression adjustments to experimental data},
	url = {https://projecteuclid.org/euclid.aoas/1365527200},
	doi = {10.1214/12-AOAS583},
	language = {EN},
	number = {1},
	urldate = {2020-12-03},
	journal = {Annals of Applied Statistics},
	author = {Lin, Winston},
	month = mar,
	year = {2013},
	mrnumber = {MR3086420},
	zmnumber = {06171273},
	note = {Publisher: Institute of Mathematical Statistics},
	pages = {295--318},
}

@article{chen2018overidentification,
	title = {Overidentification in {Regular} {Models}},
	volume = {86},
	issn = {0012-9682},
	url = {https://www.jstor.org/stable/44955258},
	number = {5},
	urldate = {2023-04-14},
	journal = {Econometrica},
	author = {Chen, Xiaohong and Santos, Andres},
	year = {2018},
	note = {Publisher: [Wiley, The Econometric Society]},
	pages = {1771--1817},
}

@article{cytrynbaum2023designing,
  title={Designing representative and balanced experiments by local randomization},
  author={Cytrynbaum, Max},
  journal={arXiv preprint arXiv:2111.08157},
  year={2023}
}

@article{abadie2008estimation,
	title = {Estimation of the {Conditional} {Variance} in {Paired} {Experiments}},
    volume = {1},
	issn = {0769-489X},
	doi = {10.2307/27917244},
	number = {91/92},
	urldate = {2020-06-15},
	journal = {Annales d'Économie et de Statistique},
	author = {Abadie, Alberto and Imbens, Guido W.},
	year = {2008},
	pages = {175--187},
}

@article{chen2008semiparametric,
	title = {Semiparametric {Efficiency} in {GMM} {Models} with {Auxiliary} {Data}},
	volume = {36},
	issn = {0090-5364},
	url = {https://www.jstor.org/stable/25464647},
	number = {2},
	urldate = {2023-06-29},
	journal = {The Annals of Statistics},
	author = {Chen, Xiaohong and Hong, Han and Tarozzi, Alessandro},
	year = {2008},
	note = {Publisher: Institute of Mathematical Statistics},
	pages = {808--843},
}

@article{angrist2009effects,
	title = {The {Effects} of {High} {Stakes} {High} {School} {Achievement} {Awards}: {Evidence} from a {Randomized} {Trial}},
	volume = {99},
	issn = {0002-8282},
	shorttitle = {The {Effects} of {High} {Stakes} {High} {School} {Achievement} {Awards}},
	url = {https://www.aeaweb.org/articles?id=10.1257/aer.99.4.1384},
	doi = {10.1257/aer.99.4.1384},
	language = {en},
	number = {4},
	urldate = {2020-12-03},
	journal = {American Economic Review},
	author = {Angrist, Joshua and Lavy, Victor},
	month = sep,
	year = {2009},
	pages = {1384--1414},
}

@article{banerjee2015miracle,
	title = {The {Miracle} of {Microfinance}? {Evidence} from a {Randomized} {Evaluation}},
	volume = {7},
	issn = {1945-7782},
	shorttitle = {The {Miracle} of {Microfinance}?},
	url = {https://www.aeaweb.org/articles?id=10.1257/app.20130533},
	doi = {10.1257/app.20130533},
	language = {en},
	number = {1},
	urldate = {2020-12-03},
	journal = {American Economic Journal: Applied Economics},
	author = {Banerjee, Abhijit and Duflo, Esther and Glennerster, Rachel and Kinnan, Cynthia},
	month = jan,
	year = {2015},
	pages = {22--53},
}

@article{bruhn2016impact,
	title = {The {Impact} of {High} {School} {Financial} {Education}: {Evidence} from a {Large}-{Scale} {Evaluation} in {Brazil}},
	volume = {8},
	issn = {1945-7782},
	shorttitle = {The {Impact} of {High} {School} {Financial} {Education}},
	url = {https://www.aeaweb.org/articles?id=10.1257/app.20150149},
	doi = {10.1257/app.20150149},
	language = {en},
	number = {4},
	urldate = {2020-12-03},
	journal = {American Economic Journal: Applied Economics},
	author = {Bruhn, Miriam and Leão, Luciana de Souza and Legovini, Arianna and Marchetti, Rogelio and Zia, Bilal},
	month = oct,
	year = {2016},
	pages = {256--295},
}

@article{chernozhukov2018doubledebiased,
	title = {Double/debiased machine learning for treatment and structural parameters},
	volume = {21},
	copyright = {© 2017 Royal Economic Society.},
	issn = {1368-423X},
	url = {https://onlinelibrary.wiley.com/doi/abs/10.1111/ectj.12097},
	doi = {10.1111/ectj.12097},
	language = {en},
	number = {1},
	urldate = {2023-07-03},
	journal = {The Econometrics Journal},
	author = {Chernozhukov, Victor and Chetverikov, Denis and Demirer, Mert and Duflo, Esther and Hansen, Christian and Newey, Whitney and Robins, James},
	year = {2018},
	note = {\_eprint: https://onlinelibrary.wiley.com/doi/pdf/10.1111/ectj.12097},
	pages = {C1--C68},
}

@article{casaburi2022using,
	title = {Using {Individual}-{Level} {Randomized} {Treatment} to {Learn} about {Market} {Structure}},
	volume = {14},
	issn = {1945-7782},
	url = {https://www.aeaweb.org/articles?id=10.1257/app.20200306},
	doi = {10.1257/app.20200306},
	language = {en},
	number = {4},
	urldate = {2023-07-03},
	journal = {American Economic Journal: Applied Economics},
	author = {Casaburi, Lorenzo and Reed, Tristan},
	month = oct,
	year = {2022},
	pages = {58--90},
}

@article{farrell2015robust,
  title={Robust inference on average treatment effects with possibly more covariates than observations},
  author={Farrell, Max H},
  journal={Journal of Econometrics},
  volume={189},
  number={1},
  pages={1--23},
  year={2015},
  publisher={Elsevier}
}

@article{belloni2017program,
  title={Program evaluation and causal inference with high-dimensional data},
  author={Belloni, Alexandre and Chernozhukov, Victor and Fernandez-Val, Ivan and Hansen, Christian},
  journal={Econometrica},
  volume={85},
  number={1},
  pages={233--298},
  year={2017},
  publisher={Wiley Online Library}
}

@article{bruhn2009pursuit,
  title={In pursuit of balance: Randomization in practice in development field experiments},
  author={Bruhn, Miriam and McKenzie, David},
  journal={American economic journal: applied economics},
  volume={1},
  number={4},
  pages={200--232},
  year={2009},
  publisher={American Economic Association}
}

@misc{cytrynbaum2023covariate,
	title = {Covariate {Adjustment} in {Stratified} {Experiments}},
	url = {http://arxiv.org/abs/2302.03687},
	doi = {10.48550/arXiv.2302.03687},
	urldate = {2024-02-08},
	publisher = {arXiv},
	author = {Cytrynbaum, Max},
	month = sep,
	year = {2023},
	note = {arXiv:2302.03687 [econ, stat]},
}

@article{freedman2008regression,
  title={On regression adjustments to experimental data},
  author={Freedman, David A},
  journal={Advances in Applied Mathematics},
  volume={40},
  number={2},
  pages={180--193},
  year={2008},
  publisher={Elsevier}
}

@misc{chernozhukov2024long,
	title = {Long {Story} {Short}: {Omitted} {Variable} {Bias} in {Causal} {Machine} {Learning}},
	shorttitle = {Long {Story} {Short}},
	url = {http://arxiv.org/abs/2112.13398},
	doi = {10.48550/arXiv.2112.13398},
	urldate = {2024-08-14},
	publisher = {arXiv},
	author = {Chernozhukov, Victor and Cinelli, Carlos and Newey, Whitney and Sharma, Amit and Syrgkanis, Vasilis},
	month = may,
	year = {2024},
	note = {arXiv:2112.13398 [cs, econ, stat]},
}

@article{bai2024inference,
	title = {Inference for {Matched} {Tuples} and {Fully} {Blocked} {Factorial} {Designs}},
	volume = {15},
	number = {2},
	journal = {Quantitative Economics},
	author = {Bai, Yuehao and Liu, Jizhou and Tabord-Meehan, Max},
	year = {2024},
	pages = {279--330},
}

@article{bai2024covariate,
	title = {Covariate adjustment in experiments with matched pairs},
	volume = {241},
	issn = {0304-4076},
	url = {https://www.sciencedirect.com/science/article/pii/S0304407624000861},
	doi = {10.1016/j.jeconom.2024.105740},
	number = {1},
	urldate = {2024-05-15},
	journal = {Journal of Econometrics},
	author = {Bai, Yuehao and Jiang, Liang and Romano, Joseph P. and Shaikh, Azeem M. and Zhang, Yichong},
	month = apr,
	year = {2024},
	pages = {105740},
}

@article{bai2024inference-1,
	title = {Inference in cluster randomized trials with matched pairs},
	volume = {245},
	issn = {0304-4076},
	url = {https://www.sciencedirect.com/science/article/pii/S0304407624002185},
	doi = {10.1016/j.jeconom.2024.105873},
	number = {1},
	urldate = {2024-11-19},
	journal = {Journal of Econometrics},
	author = {Bai, Yuehao and Liu, Jizhou and Shaikh, Azeem M. and Tabord-Meehan, Max},
	month = oct,
	year = {2024},
	pages = {105873},
}

@article{bai2025inference,
	title = {Inference in {Experiments} with {Matched} {Pairs} and {Imperfect} {Compliance}},
	volume = {43},
	issn = {0735-0015},
	url = {https://doi.org/10.1080/07350015.2024.2416972},
	doi = {10.1080/07350015.2024.2416972},
	number = {3},
	urldate = {2025-08-25},
	journal = {Journal of Business \& Economic Statistics},
	author = {Bai, Yuehao and Guo, Hongchang and Shaikh, Azeem M. and Tabord-Meehan, Max},
	month = jul,
	year = {2025},
	note = {Publisher: ASA Website
\_eprint: https://doi.org/10.1080/07350015.2024.2416972},
	pages = {627--642},
}

@article{cattaneo2010efficient,
  title={Efficient semiparametric estimation of multi-valued treatment effects under ignorability},
  author={Cattaneo, Matias D},
  journal={Journal of Econometrics},
  volume={155},
  number={2},
  pages={138--154},
  year={2010},
  publisher={Elsevier}
}

@article{efron1971forcing,
  title={Forcing a sequential experiment to be balanced},
  author={Efron, Bradley},
  journal={Biometrika},
  volume={58},
  number={3},
  pages={403--417},
  year={1971},
  publisher={Oxford University Press}
}

@article{wei1978adaptive,
  title={The adaptive biased coin design for sequential experiments},
  author={Wei, Lee-Jen},
  journal={The Annals of Statistics},
  volume={6},
  number={1},
  pages={92--100},
  year={1978},
  publisher={Institute of Mathematical Statistics}
}

@article{cytrynbaum2024finely,
  title={Finely Stratified Rerandomization Designs},
  author={Cytrynbaum, Max},
  journal={arXiv preprint arXiv:2407.03279},
  year={2024}
}

@article{pocock1975sequential,
  title={Sequential treatment assignment with balancing for prognostic factors in the controlled clinical trial},
  author={Pocock, Stuart J and Simon, Richard},
  journal={Biometrics},
  pages={103--115},
  year={1975},
  publisher={JSTOR}
}

@article{bai2025new,
  title={A New Design-Based Variance Estimator for Finely Stratified Experiments},
  author={Bai, Yuehao and Huang, Xun and Romano, Joseph P and Shaikh, Azeem M and Tabord-Meehan, Max},
  journal={arXiv preprint arXiv:2503.10851},
  year={2025}
}

\end{document}